\documentclass[10pt,a4paper]{article}
\usepackage{hyperref}
\usepackage[english]{babel}
\usepackage[utf8]{inputenc}
\usepackage[T1]{fontenc}
\usepackage{graphicx}
\usepackage{geometry}
\usepackage{tabularx}
\usepackage{psfrag}
\usepackage{fancyhdr}
\usepackage{xcolor}
\usepackage{sistyle}
\usepackage{amsfonts}
\usepackage{amsmath}
\usepackage{amssymb}
\usepackage{booktabs}
\usepackage{multirow}
\usepackage{amsthm}
\usepackage{pifont}
\usepackage{algorithm}
\usepackage{algorithmicx}
\usepackage{enumerate}
\usepackage{enumitem}
\usepackage{bbold}
\usepackage{tikz}
\usepackage{hyperref}
\usepackage{mathrsfs}
\usepackage{hhline}
\usepackage[numbers,sort&compress]{natbib}
\graphicspath{{../}}

\theoremstyle{theorem}
\newtheorem{proposition}{Proposition}

\theoremstyle{definition}
\newtheorem{definition}{Definition}
\newtheorem{problem}{Problem}

\theoremstyle{remark}

\usepackage[caption=false,font=footnotesize]{subfig}

\begin{document}

\title{Hyperparameter selection for Discrete Mumford-Shah \thanks{This work is supported by ANR-19-CE48-0009s MULTISC-IN.}}

\author{Charles-G\'erard Lucas, Barbara Pascal, Nelly Pustelnik, Patrice Abry
\thanks{C.-G. Lucas, N. Pustelnik and P. Abry are with ENSL, CNRS, Laboratoire de physique, F-69342 Lyon, France (e-mail: firstname.surname@ens-lyon.fr).
B. Pascal is with Nantes Université, École Centrale Nantes, CNRS, LS2N, UMR 6004, F-44000 Nantes, France (e-mail: barbara.pascal@cnrs.fr).}
}
\maketitle

\begin{abstract}
This work focuses on a parameter-free joint piecewise smooth image denoising and contour detection. Formulated as the minimization of a discrete Mumford-Shah functional and estimated \textit{via} a theoretically grounded alternating minimization scheme,  the bottleneck of such a variational approach lies in the need to fine-tune their hyperparameters, while not having access to ground truth data.
To that aim, a Stein-like strategy providing optimal hyperparameters is designed, based on the minimization of an unbiased estimate of the quadratic risk.
Efficient and automated minimization of the estimate of the risk crucially relies on an unbiased estimate of the gradient of the risk with respect to hyperparameters. Its practical implementation is performed using a \textit{forward} differentiation of the alternating scheme minimizing the Mumford-Shah functional, requiring exact differentiation of the proximity operators involved. 
Intensive numerical experiments are performed on synthetic images with different geometry and noise levels, assessing the accuracy and the robustness of the proposed procedure.
The resulting \textit{parameter-free piecewise-smooth estimation} and contour detection procedure, not requiring prior image processing expertise nor annotated data, can then be applied to real-world images.
\end{abstract}

\section{Introduction}

\noindent \textbf{Context} -- Image processing is characterized by several key tasks such as image recovery (e.g., debluring and/or denoising), feature extraction, segmentation, and contour detection, to name a few. 
To provide the user with the requested information, it is standard to perform successively a certain number of these tasks.  
A first major drawback of cascading tasks, is that important information might be thrown away at each stage.
A second key issue is that each task might introduce estimation variance and/or regularization bias, which may accumulate and lead to subsequent errors on the target quantity.
Finally, the selection of hyperparameters, e.g., regularization parameters, needs to be performed for each task independently, which might turn sub-optimal overall in minimizing the final error on the output estimate. 

The benefit of performing jointly several steps has been illustrated in the context of texture segmentation~\cite{Pascal_B_2021_j-acha}, providing a comparison between a two-step procedure (extract relevant local texture features  followed  by segmentation) against an original single-step procedure intertwining the estimation of relevant features and the segmentation procedure. Both strategies lead to strongly convex optimization schemes and fair comparisons can be provided by having recourse to an automatic hyperparameters selection procedure relying on Stein Unbiased Risk Estimator \cite{Stein_C_1981_j-annals-statistics_estimation_mmnd}.

Along this line, recent contributions in the image processing literature were dedicated to joint image denoising/restoration and contour detection \cite{Storath_M_2014_ieee-tsp_jssrupf, zach_discretized_2017, Kiefer_L_2020_ipol} via biconvex proximal optimization schemes.  However, the automatic tuning of hyperparameters in this context has not been dealt with yet: this is the object of the present contribution. 

\noindent \textbf{D-MS} -- This work focuses on a bi-convex  formulation, refered as Discrete Mumford-Shah (D-MS) functional, that can trace back to the Mumford-Shah~\cite{mumford1989optimal} or Geman and Geman functionals~\cite{geman1984stochastic}, aiming to perform joint image denoising and contour detection, which may be written in the discrete variational formulation setting as:
\begin{equation} 
\label{eq:dms_univariate}
\underset{\boldsymbol{u} \in \mathbb{R}^{\lvert \Omega \rvert}, \boldsymbol{e}\in \mathbb{R}^{\vert \mathcal{E}\vert}}{\textrm{min}}   \frac{1}{2} \Vert \boldsymbol{u} - \boldsymbol{z} \rVert_2^2 + \beta \lVert (1-\boldsymbol{e}) \odot \boldsymbol{\mathrm{D}} \boldsymbol{u} \rVert_2^2 +  \lambda h(\boldsymbol{e}),
\end{equation}
where $\boldsymbol{z} = \overline{\boldsymbol{u}} + \sigma \boldsymbol{\zeta} \in \mathbb{R}^{\lvert \Omega \rvert}$ with $\boldsymbol{\zeta}\sim \mathcal{N}(\boldsymbol{0}_{\lvert \Omega \rvert}, \textbf{I}_{\lvert \Omega \rvert})$ denotes the observed degraded image, defined on a grid of pixels $\Omega$ such that $\lvert \Omega \rvert=N$, and $\sigma>0$ is the \textit{known} standard-deviation of the noise. The variable $\boldsymbol{e}$ is a discrete field defined on the lattice $\mathcal{E}$, encapsulating the contour information, whose values are 1 when a contour is detected and 0 otherwise, 
$\boldsymbol{\mathrm{D}}\colon \mathbb{R}^{\lvert \Omega \rvert} \to \mathbb{R}^{\vert \mathcal{E}\vert}$ is a discrete difference operator such that $\boldsymbol{\mathrm{D}}\boldsymbol{u}$ lives on a lattice of contours $\mathcal{E}$,  
$\odot$ denotes the component-wise product, $h$ denotes a convex separable function enforcing sparsity having its minimum in $0$ and such that, for every $\boldsymbol{e} = (e_i)_{1\leq i \leq \mathcal{E}}$, $h(\boldsymbol{e}) = \sum_i h_i(e_i)$. $\beta>0$ and $\lambda>0$ are regularization parameters. The minimization is performed with SL-PAM, a nonconvex alternated minimization scheme,  with descent parameters genuinely chosen to ensure fast convergence \cite{foare2019semi} and recalled in Algorithm~\ref{alg:dms0} in the Appendix~\ref{alg:dms0}. An exhaustive state-of-the-art regarding variations of the minimization problem~\eqref{eq:dms_univariate} is provided in Appendix~\ref{ssec:SL-PAM}.

\noindent \textbf{Hyperparameter selection }-- The aforementioned procedure for image denoising and contour detection involves \textit{hyperparameters}, e.g. $\beta$ and $\lambda$ in~\eqref{eq:dms_univariate}.
To reach satisfactory performance, the fine-tuning of these parameters is crucial.
Although central in signal and image processing, this difficult task is still an ongoing challenge, particularly for {variational} methods.

The difficult problem of the selection of the regularization parameters of the D-MS functional is addressed considering a Stein Unbiased Risk Estimate (SURE), combined with a Finite Difference Monte Carlo (FDMC) strategy making its practical computation tractable. For a detailed state-of-the-art dedicated to hyperparameter selection, the reader could refer to Appendix~\ref{sec:soahyper}.  
The optimal regularization parameters obtained by minimizing FDMC SURE \textit{via} exhaustive grid search are shown to lead to denoised estimates with high \textit{signal-to-noise ratio} and relevant contours.

Then, to provide a fast procedure selecting the regularization parameters, a Stein Unbiased GrAdient Risk estimate (SUGAR) adapted to D-MS functional~\eqref{eq:dms_univariate} is designed, involving the Jacobian of the parametric estimator obtained from~\eqref{eq:dms_univariate}.
Practical implementation of SUGAR requires iterative differentiation of the SL-PAM minimization scheme, for which closed-form formulas are provided. 
An averaging Monte Carlo strategy is discussed, providing a robust FDMC SUGAR estimator.
The resulting procedure compares favorably against exhaustive grid search in terms of  \textit{signal-to-noise ratio}, while requiring a significantly smaller computational cost.
To the best of our knowledge, the proposed automated D-MS bi-level scheme constitutes a first automated, prior-free and fast discrete Mumford-Shah-like formalism with automated selection of regularization parameters.

\noindent \textbf{Outline} --   
The proposed automated and fast procedure is described in Section~\ref{sec:hyperparameters}. Numerical experiments are provided in Section~\ref{sec:num_exp_BP}.

\section{Hyperparameter selection for D-MS}
\label{sec:hyperparameters}

\subsection{Stein estimators for D-MS }\label{sec:specifDMS}

The minimization of the D-MS functional of Eq.~\eqref{eq:dms_univariate} provides both a piecewise smooth image reconstruction, denoted $\widehat{\boldsymbol{u}}(\boldsymbol{z}; \beta, \lambda)$ and a set of detected contours, encapsulated into $\widehat{\boldsymbol{e}}(\boldsymbol{z}; \beta, \lambda)$, depending on the choice of hyperparameters $\Theta = (\beta, \lambda) \in \mathbb{R}_+ \times \mathbb{R}_+$. 
In such a context, we should ideally minimize a \textit{global} error criterion 
\begin{align}
\label{eq:def_quad_riskue}
\widehat{\Theta} \in \underset{\Theta}{\mathrm{Argmin}} \; \mathrm{d}( \widehat{\boldsymbol{x}}(\boldsymbol{z};\Theta), \overline{\boldsymbol{x}}),
\end{align}
measuring the ability of the piecewise-smooth image and contour estimates \\
$\widehat{\boldsymbol{x}}(\boldsymbol{z};\Theta) = \left(\widehat{\boldsymbol{u}}(\boldsymbol{z}; \beta, \lambda),\widehat{\boldsymbol{e}}(\boldsymbol{z}; \beta, \lambda)\right)$ to approximate the original data $\overline{\boldsymbol{x}} = \left(\overline{\boldsymbol{u}},\overline{\boldsymbol{e}}\right)$, using a measure of similarity $\mathrm{d}$.
If $\mathrm{d}$ is chosen to be a quadratic risk, it reads:
\begin{align}
\label{eq:def_quad_riskue}
\mathrm{d}( \widehat{\boldsymbol{x}}(\boldsymbol{z};\Theta), \overline{x}) =  \mathbb{E} [\Vert \widehat{\boldsymbol{u}}(\boldsymbol{z};\Theta) - \overline{\boldsymbol{u}} \Vert_2^2 ] + \zeta \mathbb{E}[\Vert \widehat{\boldsymbol{e}}(\boldsymbol{z};\Theta) - \overline{\boldsymbol{e}} \Vert_2^2],
\end{align}
where $\zeta\geq 0$. 
However, the degradation model only describes how the observed image $\boldsymbol{z}$ relates to the ground truth image $\overline{\boldsymbol{u}}$,  no prior knowledge about how the ground truth contours $\overline{\boldsymbol{e}}$ are affected by the observation noise being assumed. 
Further, measuring the accuracy of the contours is a tedious task, involving complicated criteria, such as the Jaccard index~\cite{jaccard1901distribution}. 
For these reasons, in the present work, the quadratic error on which the choice of hyperparameter relies is chosen to be the \textit{quadratic estimation error on the reconstructed image} (i.e. $\zeta=0$).

The present work focuses on a strategy combining Finite Difference approximated differentiation and Monte Carlo averaging, which was first described by~\cite{ramani2008monte}. 
Making use of a Finite Difference step  $\epsilon > 0$ and a Monte Carlo vector $ \boldsymbol{\delta} \in \mathbb{R}^N$  drawn from $ \mathcal{N}(\boldsymbol{0}_N,\textbf{I}_N)$,  Finite Difference Monte Carlo (FDMC) SURE  is defined as:
    \begin{align}
    \label{eq:fdmcsure}
        \mathrm{SURE}_{\epsilon, \boldsymbol{\boldsymbol{\delta}}}(\boldsymbol{z};\Theta \lvert \sigma^2) := \| ( \widehat{\boldsymbol{u}}(\boldsymbol{z};\Theta) - \boldsymbol{z}) \|_2^2 +  \frac{2}{\epsilon} \langle \widehat{\boldsymbol{u}}(\boldsymbol{z} + \epsilon \boldsymbol{\delta};\Theta) - \widehat{\boldsymbol{u}}(\boldsymbol{z};\Theta), \sigma^2 \boldsymbol{\delta} \rangle - \sigma^2 N.
    \end{align}
    It involves the denoised image $\widehat{\boldsymbol{u}}(\boldsymbol{z};\beta, \lambda)$, obtained from the minimization of \eqref{eq:dms_univariate}, and the standard deviation of the noise $\sigma$. 
Under technical assumptions detailed in Appendix~\ref{ssec:SURE}, the \textit{true} inaccessible quadratic risk estimator~\eqref{eq:def_quad_riskue} when $\zeta=0$ satisfies the following asymptotic unbiasedness property:
    \begin{equation}
    \label{eq:sure_thm}
        \underset{\epsilon \longrightarrow 0}{\lim} \mathbb{E} [\mathrm{SURE}_{\epsilon, \boldsymbol{\delta}}(\boldsymbol{z};\Theta \lvert \sigma^2)] = \mathbb{E}[\Vert \widehat{\boldsymbol{u}}(\boldsymbol{z};\Theta) - \overline{\boldsymbol{u}} \Vert_2^2].
    \end{equation}
Then, the design of a fast \textit{gradient}-based hyperparameter selection strategy providing optimal hyperparameter from the minimization of~\eqref{eq:fdmcsure}  requires an unbiased estimate $\partial_{\Theta} \mathrm{SURE}_{\epsilon,\boldsymbol{\boldsymbol{\delta}}}(\boldsymbol{z};\Theta \lvert \sigma^2)$ of the \textit{gradient} of the quadratic risk with respect to hyperparameters $\Theta$.
Such a general procedure is sketched in Algorithm~\ref{alg:mininize_sure} and relies on the FDMC Stein Unbiased GrAdient Risk (SUGAR) estimate  defined as:
    \begin{align}
    \label{eq:fdmcsugar2}
           \!\!\!\! \!\!\mathrm{SUGAR}_{\epsilon, \boldsymbol{\delta}}(\boldsymbol{z};\Theta \lvert \sigma^2 ) = 2\partial_{\Theta}\widehat{\boldsymbol{u}}(\boldsymbol{z};\Theta)^*(\widehat{\boldsymbol{u}}(\boldsymbol{z};\Theta)-\boldsymbol{z}) +  \frac{2}{\epsilon} \left(\partial_{\Theta}\widehat{\boldsymbol{u}}(\boldsymbol{z} + \epsilon \boldsymbol{\delta};\Theta) -\partial_{\Theta}\widehat{\boldsymbol{u}}(\boldsymbol{z};\Theta) \right)^* \sigma^2 \boldsymbol{\delta},
        \end{align}
where $\partial_{\Theta}\widehat{\boldsymbol{u}}(\boldsymbol{z};\Theta)$ denotes the Jacobian of the parametric estimator $\widehat{\boldsymbol{u}}(\boldsymbol{z};\Theta)$ with respect to the hyperparameters $\Theta$. The first proposal of such Stein Unbiased GrAdient Risk (SUGAR) estimate  was formulated by~\cite{deledalle2014stein} for i.i.d. Gaussian noise, and then extended in~\cite{pascal2020automated} for correlated noise. The main difficulty when it comes to practical implementation  is to evaluate the Jacobian matrices. 

{\color{black}
      \begin{algorithm}[htbp]
            \caption{Automated selection of hyperparameters.\label{alg:mininize_sure}}
            \begin{algorithmic}
            \State \textbf{Input:} Data $\boldsymbol{z}$, $\epsilon$, $\boldsymbol{\delta}$, and true or estimated $\sigma^2$.  
            \State \textbf{Initialization:} Set $\Theta^{[0]} \in \mathbb{R}^L $. 
            \State \textbf{For } $t= 0$ \textbf{to} $T_{\max}-1$ \textbf{do}
            
            $\left \lfloor \begin{array}{l} 
             \mbox{Compute $ \mathrm{SURE}_{\epsilon, {\boldsymbol{\delta}}}(\boldsymbol{z};\Theta^{[t]} \lvert \sigma^2)$} \\
             \mbox{Compute $\mathrm{SUGAR}_{\epsilon, \boldsymbol{\delta}}(\boldsymbol{z};\Theta^{[t]} \lvert \sigma^2 )$} \\
             \mbox{Update $\Theta^{[t]} $ to $\Theta^{[t+1]}$ \textit{via} a gradient} \\
             \mbox{descent step}
              \end{array} \right.$ \\
                   \State \textbf{Output:}  $\Theta^* =  \Theta^{[T_{\rm{max}}]}$
        	\end{algorithmic}
        \end{algorithm}

\subsection{Differentiated SL-PAM}
 
Practical evaluation of the risk and gradient of the risk estimates from Eq.~\eqref{eq:fdmcsure}~and~\eqref{eq:fdmcsugar2}, requires to compute the Jacobian of the D-MS estimator.
No closed-form expression being available for $\widehat{\boldsymbol{u}}(\boldsymbol{z}; \beta, \lambda)$, the derivatives are obtained from the iterative differentiation of the recursive scheme of DMS-SLPAM, Algorithm~\ref{alg:dms0} in Appendix~\ref{ssec:SL-PAM}. 
This strategy raises several technical issues.
Indeed, following~\cite{deledalle2014stein}~and~\cite{pascal2020automated}, the Jacobian matrices $\partial_{\Theta} \widehat{\boldsymbol{u}}(\boldsymbol{z}; \beta, \lambda)$ and $\partial_{\Theta} \widehat{\boldsymbol{u}}(\boldsymbol{z}+ \epsilon \boldsymbol{\delta}; \beta, \lambda)$  are  computed iteratively from a \textit{differentiated} recursive scheme.
Particularized to the case of D-MS estimates,  the chain differentiation of the SL-PAM scheme (cf. Algorithm~\ref{alg:dms0} in Appendix~\ref{ssec:SL-PAM}) is derived in Algorithm~\ref{dslpam}.
For ease of computation, a specific choice of the step-size $d_k = \eta \beta \Vert  \boldsymbol{\mathrm{D}} \Vert^2$ involved in the update of the variable $\boldsymbol{e}$ is considered, without inducing any loss of generality.

                \begin{algorithm}[htbp]
            \caption{Iterative differentiation of SL-PAM \label{dslpam}}
            \begin{algorithmic}
            \State  \textbf{Input:} Data $\widetilde{\boldsymbol{z}}=\{\boldsymbol{z},\boldsymbol{z}+ \epsilon \boldsymbol{\delta}\}$. Set $\Theta = (\beta, \lambda) \in \mathbb{R}_+ \times \mathbb{R}_+$.
             \State  \textbf{Initialization}: $\boldsymbol{u}^{[0]}=\widetilde{\boldsymbol{z}}$, $\boldsymbol{e}^{[0]}= \boldsymbol{1}_{\lvert \mathcal{E}\rvert}$, \\
             $\partial_\Theta  \widetilde{\boldsymbol{u}}^{[0]} = \partial_\Theta \boldsymbol{u}^{[0]} =  \boldsymbol{0}_N$, $ \partial_\Theta  \widetilde{\boldsymbol{e}}^{[0]} = \partial_\Theta  \boldsymbol{e}^{[0]} =\boldsymbol{0}_{\lvert \mathcal{E}\rvert}$.
             \State  Set $\gamma > 1$ and $\eta > 0$. 
            \State \textbf{While} $\vert \Psi(\boldsymbol{u}^{[k+1]},\boldsymbol{e}^{[k+1]})-\Psi(\boldsymbol{u}^{[k]},\boldsymbol{e}^{[k]}) \vert > \xi$ \\
            	$\left \lfloor \begin{array}{l}
            	
            \mbox{Set $c_k = \gamma  \beta \Vert  \boldsymbol{\mathrm{D}} \Vert^2$ and $d_k = \eta \beta \Vert  \boldsymbol{\mathrm{D}} \Vert^2$}\; \\
            \widetilde{\boldsymbol{u}}^{[k+1]} = \boldsymbol{u}^{(k)}-\frac{1}{c_k}\nabla_{\boldsymbol{u}} g(\boldsymbol{u}^{[k]},\boldsymbol{e}^{[k]})\; \\
            \boldsymbol{u}^{[k+1]} = \mathrm{prox}_{\frac{1}{c_k} f(\cdot; \widetilde{\boldsymbol{z}})} \left(\widetilde{\boldsymbol{u}}^{[k+1]}\right)\;  \\
\mbox{Compute $\partial_\Theta  \widetilde{\boldsymbol{u}}^{[k+1]}$ from Eq. \eqref{eq:u_tilde_diff}}\; \\
            \mbox{Compute $\partial_\Theta \boldsymbol{u}^{[k+1]}$ from Eq. \eqref{eq:u_diff}}\; \\
                       \mbox{For all $  i \in \{1,\ldots ,\vert \mathcal{E} \vert\} $} \\
          \left \lfloor \begin{array}{l} 
             \widetilde{e}_i^{[k]} =\frac{\beta \displaystyle (\boldsymbol{\mathrm{D}}_i \boldsymbol{u}^{[k+1]})^2 + \frac{d_k e_i^{[k]} }{2}}{\beta \displaystyle (\boldsymbol{\mathrm{D}}_i \boldsymbol{u}^{[k+1]})^2 + \frac{d_k}{2}} \\
             e_i^{[k+1]} = \mathrm{prox}_{ \frac{\lambda }{2 \beta (\boldsymbol{\mathrm{D}}_i \boldsymbol{u}^{[k+1]})^2 + d_k} {h_i}} (\widetilde{e}_i^{[k]}) \\
            \mbox{Compute $\partial_\Theta  \widetilde{e}_i^{[k+1]}$ from Eq. \eqref{eq:e_tilde_diff_distinct}}\; \\
            \mbox{Compute $\partial_\Theta  e_i^{[k+1]}$  from Eq.  \eqref{eq:e_diff_distinct}}
                                       \end{array} \right. \\
            \end{array} \right.$
            \end{algorithmic}
        \end{algorithm}

\noindent \textbf{General procedure --}          
The purpose is to differentiate the mapping $\Theta \mapsto \left(\widehat{\boldsymbol{u}}(\boldsymbol{z};\Theta),\widehat{\boldsymbol{e}}(\boldsymbol{z};\Theta)\right)$, where the estimates $\left(\widehat{\boldsymbol{u}}(\boldsymbol{z};\Theta),\widehat{\boldsymbol{e}}(\boldsymbol{z};\Theta)\right)$ are obtained solving~ \eqref{eq:dms_univariate} for \textit{fixed} $\boldsymbol{z}$.

The recursive \textit{chain differentiation} consists in differentiating step by step DMS-SLPAM each update of which can  
be written as $\mathbf{v}(\boldsymbol{z};\Theta) = \Gamma(\textbf{u}(\boldsymbol{z};\Theta),\textbf{e}(\boldsymbol{z};\Theta),\tau(\Theta))$, where  $\textbf{u}: \mathbb{R}^N \times\mathbb{R}^L \rightarrow \mathbb{R}^N$, $\textbf{e}: \mathbb{R}^N \times \mathbb{R}^L \rightarrow \mathbb{R}^{\vert  \mathcal{E} \vert}$ and $\tau: \mathbb{R}^L \rightarrow \mathbb{R}$ are functions of the observed noisy image $\boldsymbol{z}$ and of the hyperparameters $\Theta$, with respect to which the differentiation is to be performed and $\mathbf{v}(\boldsymbol{z};\Theta) \in \mathcal{K}$, where $\mathcal{K} = \mathbb{R}^N$ when updating $\boldsymbol{u}$ or $\widetilde{\boldsymbol{u}}$, and  $\mathcal{K} = \mathbb{R}^{\lvert \mathcal{E} \rvert}$ when updating $\boldsymbol{e}$ or $\widetilde{\boldsymbol{e}}$.

Then, applying the chain rule differentiation principle yields the following partial derivative expression, for every component $\theta$ of the hyperparameter vector $\Theta$ and for every index $j \in\{1,\ldots, \mathrm{dim}(\mathcal{K})\}$:
        \begin{align}
        \label{eq:chain_rule}
             \partial_{\theta} \mathrm{v}_{j}   = \sum_{\ell = 1}^{N} (\partial_{ \mathrm{u}_\ell } \Gamma_j) ( \partial_{\theta} \mathrm{u}_\ell) +  \sum_{m = 1}^{\vert  \mathcal{E} \vert} (\partial_{ \mathrm{e}_m} \Gamma_j) (\partial_{\theta} \mathrm{e}_m) + (\partial_{\tau} \Gamma_j )(\partial_{\theta} \tau) 
        \end{align}
        leading to the following closed form expression for $\partial_\Theta  \widetilde{\boldsymbol{u}}^{[k+1]}$, $\partial_\Theta \boldsymbol{u}^{[k+1]}$, $\partial_\Theta  \widetilde{e}_i^{[k+1]}$ and $\partial_\Theta  e_i^{[k+1]}$, for $  i \in \{1,\ldots ,\vert \mathcal{E} \vert\} $.

        \noindent \textbf{Iterative differentiation of DMS-SLPAM --}
        \label{sec:computation}
Applying Formula~\eqref{eq:chain_rule} to each step of DMS-SLPAM leads to Algorithm~\ref{dslpam}.
Proposition~\ref{prop:derpar} provides closed-form expressions of the updates of the Jacobian matrices of the iterates involved in the minimization of a D-MS functional with $h_i= \vert \cdot \vert$, 
thus allowing an easy and direct implementation of FDMC SURE and FDMC SUGAR estimates of Eq.~\eqref{eq:fdmcsure}~and~\eqref{eq:fdmcsugar2}.

\begin{proposition} 
\label{prop:derpar}
        Considering the D-MS functional~\eqref{eq:dms_univariate} when $h_i= \vert \cdot \vert$ and its minimization via SL-PAM  Algorithm~\ref{alg:SLPAM} in Appendix~\ref{ssec:SL-PAM} with  
$d_k = \beta \overline{d}  $, $\overline{d} = \eta \Vert \boldsymbol{\mathrm{D}} \Vert_2^2$, for every $\theta \in \{\beta,\lambda\}$:
        \begin{align}
            \label{eq:u_tilde_diff}
            \partial_{\theta}\widetilde{\boldsymbol{u}}^{[k]}  = 
             \partial_{\theta}\boldsymbol{u}^{[k]} - \frac{2\beta}{c_k}  \displaystyle \sum_{i=1}^{\vert \mathcal{E} \vert}  (1-e^{[k]}_i)^2 \ \boldsymbol{\mathrm{D}}_i^*\boldsymbol{\mathrm{D}}_i \partial_{\theta} \boldsymbol{u}^{[k]} 
                + \frac{4 \beta}{c_k} \displaystyle \sum_{i=1}^{\vert \mathcal{E} \vert}  (1-e^{[k]}_i) \partial_{\theta} e_i^{[k]}\ \boldsymbol{\mathrm{D}}_i^*\boldsymbol{\mathrm{D}}_i \boldsymbol{u}^{[k]},
            \end{align}

        \begin{equation}\label{eq:u_diff}
            \partial_{\theta} \boldsymbol{u}^{[k+1]}  =  \frac{c_k}{c_k+1}\partial_{\theta} \widetilde{\boldsymbol{u}}^{[k]} + \frac{\widetilde{\boldsymbol{u}}^{[k]}  -\boldsymbol{z}}{(\beta \overline{c}+1)^2} \partial_\theta c_k ,
            \end{equation}
 where $\partial_{\beta} c_k = \gamma \|D\|^2$ and $\partial_{\lambda} c_k = 0$, and for every $i \in\{1,\ldots, \vert \mathcal{E} \vert\}$:

    \begin{align}\label{eq:e_tilde_diff_distinct}
          \partial_{\theta} \widetilde{e}_i^{[k]}  
            =  \frac{2 \boldsymbol{\mathrm{D}}_i \boldsymbol{u}^{[k+1]} \boldsymbol{\mathrm{D}}_i \partial_{\theta}\boldsymbol{u}^{[k+1]}  \frac{ \overline{d}}{2}(1- e_i^{[k]} )}{ \left[ \left(\boldsymbol{\mathrm{D}}_i \boldsymbol{u}^{[k+1]} \right)^2 + \frac{ \overline{d}}{2} \right]^2} 
            +  \frac{ \frac{\overline{d}}{2}\partial_{\theta} e_i^{[k]}}{ \left(\boldsymbol{\mathrm{D}}_i \boldsymbol{u}^{[k+1]} \right)^2 + \frac{\overline{d}}{2} },
            \end{align}
            
            \begin{align}           
        \partial_{\theta}e_i^{[k+1]}   =   - \partial_{\boldsymbol{u}} \phi_i^{[k+1]} \partial_{\theta} \boldsymbol{u}^{[k+1]} \frac{\widetilde{e}_i^{[k]}  }{\vert \widetilde{e}_i^{[k]} \vert} \mathcal{I}_{\vert \widetilde{e}_i^{[k]}  \vert > \phi_i^{[k+1]} } 
         + \partial_{\theta} \widetilde{e}_i^{[k]} \mathcal{I}_{\vert \widetilde{e}_i^{[k]} \vert > \phi_i^{[k+1]} } 
         - \frac{\partial_{\tau} \phi_i^{[k+1]} \partial_{\theta} \tau}{\vert \widetilde{e}_i^{[k]} \vert} \widetilde{e}_i^{[k]}  \mathcal{I}_{\vert \widetilde{e}_i^{[k]} \vert > \phi_i^{[k+1]} }\label{eq:e_diff_distinct} ,
          \end{align}
          where 
        \begin{equation}
        \begin{cases}
        \partial_{\boldsymbol{u}} \phi_i^{[k+1]} \partial_{\theta} \boldsymbol{u}^{[k+1]} = -\frac{4 \tau  \boldsymbol{\mathrm{D}}_i \boldsymbol{u}^{[k+1]}   \boldsymbol{\mathrm{D}}_i \partial_{\theta} \boldsymbol{u}^{[k+1]} }{ \left[ 2  \left( \boldsymbol{\mathrm{D}}_i \boldsymbol{u}^{[k+1]}  \right)^2 + \overline{d} \right]^2},   \\
        \partial_\tau \phi_i^{[k+1]} = \frac{1}{ 2 \left(\boldsymbol{\mathrm{D}}_i \boldsymbol{u}^{[k+1]} \right)^2 + \overline{d} }, \\
        \partial_\beta \tau = - \frac{\lambda}{\beta^2}, \qquad \partial_\lambda \tau =  \frac{1}{\beta}.
        \end{cases}
          \end{equation}
        \end{proposition}

\begin{proof}
The proof is given in Appendix~\ref{sec:iterdiff} and the notation $\mathcal{I}$ is defined in Appendix~\ref{sec:notations}.        
\end{proof} 

\subsection{Monte Carlo averaging strategy}
\label{ssec:MC_averaging}

Following~\cite{pascal2020automated}, the risk and gradient of the risk FDMC Stein estimators, introduced in Eq.~\eqref{eq:fdmcsure}~and~\eqref{eq:fdmcsugar2}, are defined from \textit{one} realization of the Monte Carlo vector $\boldsymbol{\delta}$.
Yet, in the context of a parametric estimator $\widehat{\boldsymbol{u}}(\boldsymbol{z};\Theta)$ obtained from the minimization of a nonconvex objective functional, such as \eqref{eq:dms_univariate}, it can be necessary to go further, and to consider \textit{Monte Carlo averaging} strategies to get more robust risk and gradient of the risk estimates.

The \textit{Monte Carlo averaging} strategy consists in averaging the FDMC Stein estimators of Eq.~\eqref{eq:fdmcsure}~and~\eqref{eq:fdmcsugar2} over a certain number $R$ of random Monte Carlo vectors $\boldsymbol{\delta}^{(r)} \in \mathbb{R}^N$, independently sampled from the standard Gaussian distribution as stated properly in Definition~\ref{def:MCaverage}.

\begin{definition}[\textit{Monte Carlo averaged} Stein estimators]
\label{def:MCaverage}
Let $\boldsymbol{z} = \overline{\boldsymbol{u}} + \sigma \boldsymbol{\zeta} \in \mathbb{R}^{\lvert \Omega \rvert}$ with $\boldsymbol{\zeta}\sim \mathcal{N}(\boldsymbol{0}_{\lvert \Omega \rvert}, \textbf{I}_{\lvert \Omega \rvert})$ 
and let $\widehat{\boldsymbol{u}}(\boldsymbol{z};\Theta)$ a parametric estimator of the underlying ground truth $\overline{\boldsymbol{u}}$, depending on some hyperparameters stored in $\Theta \in \mathbb{R}^L$.
For $\epsilon > 0$ a Finite Difference step and $\boldsymbol{\Delta} = [ \boldsymbol{\delta}^{(1)},\ldots,  \boldsymbol{\delta}^{(R)}]$ a concatenation of independent 
Monte Carlo vectors sampled 
from the standard Gaussian distribution.
The \textit{Monte Carlo averaged} SURE is defined as
\begin{align}
\label{eq:mcsure}
\overline{\mathrm{SURE}}_{\epsilon, \boldsymbol{\Delta}}^R(\boldsymbol{z};\Theta) := \frac{1}{R} \sum_{r = 1}^R \mathrm{SURE}_{\epsilon, \boldsymbol{\delta}^{(r)}}(\boldsymbol{z};\Theta),
\end{align}
where $ \mathrm{SURE}_{\epsilon, \boldsymbol{\delta}^{(r)}}(\boldsymbol{z};\Theta)$ is the FDMC SURE~\eqref{eq:fdmcsure}.
Similarly, the \textit{Monte Carlo averaged} SUGAR estimator writes
\begin{align}
\label{mcsugar}
\overline{\mathrm{SUGAR}}_{\epsilon, \boldsymbol{\Delta}}^R(\boldsymbol{z};\Theta) := \frac{1}{R} \sum_{r = 1}^R \mathrm{SUGAR}_{\epsilon, \boldsymbol{\delta}^{(r)}}(\boldsymbol{z};\Theta),
\end{align}
involving $\mathrm{SUGAR}_{\epsilon, \boldsymbol{\delta}^{(r)}}(\boldsymbol{z};\Theta)$, the FDMC SUGAR estimate~\eqref{eq:fdmcsugar2}.
\end{definition}

\begin{proposition}
\label{prop:unbavsure}
Let $\widehat{\boldsymbol{u}}(\boldsymbol{z};\Theta)$ an estimator being uniformly Lipschitz w.r.t the observations $\boldsymbol{z}$ and w.r.t the hyperparameters $\Theta$, with a Lipschitz modulus $L_{\widehat{\boldsymbol{u}}(\boldsymbol{z};\cdot)}$ independent of $\boldsymbol{z}$, and satisfying the univocity condition $\widehat{\boldsymbol{u}}(\boldsymbol{0}_N;\Theta) = \boldsymbol{0}_N$. 
For $\;\epsilon$ a infinitesimal positive Finite Difference step and $\boldsymbol{\Delta} = \left[ \boldsymbol{\delta}^{(1)}, \ldots, \boldsymbol{\delta}^{(R)} \right]$ a collection of independent standard Gaussian Monte Carlo vectors, the \textit{Monte Carlo averaged} estimates $\overline{\mathrm{SURE}}_{\epsilon, \boldsymbol{\Delta}}^R(\boldsymbol{z};\Theta) $ and $\overline{\mathrm{SUGAR}}_{\epsilon, \boldsymbol{\Delta}}^R(\boldsymbol{z};\Theta) $ are asymptotically unbiased estimates of respectively the risk and of the gradient of the risk with respect to hyperparameters
\begin{equation}
\label{eq:unbavsure}
\underset{\epsilon \longrightarrow 0}{\textrm{lim}} \mathbb{E}[\overline{\mathrm{SURE}}^R_{\epsilon, \boldsymbol{\Delta}}(\boldsymbol{z};\Theta \lvert \sigma^2 )] = Q[\widehat{\boldsymbol{u}}](\Theta)
\end{equation} 
and
\begin{align}
\underset{\epsilon \longrightarrow 0}{\textrm{lim}} \mathbb{E}[\overline{\mathrm{SUGAR}}^R_{\epsilon, \boldsymbol{\Delta}}(\boldsymbol{z};\Theta \lvert \sigma^2 )] =  \partial_\Theta Q[\widehat{\boldsymbol{u}}](\Theta).
\end{align}
Moreover,  $\overline{\mathrm{SUGAR}}^R_{\epsilon, \boldsymbol{\Delta}}(\boldsymbol{z};\Theta \lvert \sigma^2 )$ is \textit{exactly} the gradient of $\overline{\mathrm{SURE}}^R_{\epsilon, \boldsymbol{\Delta}}(\boldsymbol{z};\Theta \lvert \sigma^2 )$ with respect to the hyperparameters $\Theta$.
\end{proposition}

\begin{proof}
The proof is provided in Appendix~\ref{sec:proof2}.
\end{proof}

\subsection{Averaged SUGAR D-MS}
\label{ss:algobfgs}
The framework presented in Section~\ref{sec:hyperparameters}, combined with Proposition~\ref{prop:unbavsure}, enable us to design an automated strategy to select the D-MS regularization parameters described below and assessed in Section~\ref{sec:num_exp_BP}. 
First, $R$ \textit{independent} Monte Carlo vectors $\boldsymbol{\delta}^{(r)}$ are sampled.
The set $\boldsymbol{\Delta} = \lbrace \boldsymbol{\delta}^{(1)}, \hdots, \boldsymbol{\delta}^{(R)}\rbrace$ is kept fixed throughout the procedure. 
Then, Algorithm~\ref{alg:mininize_sure} with the averaged estimates $\widehat{Q}(\boldsymbol{z}; \Theta \lvert \sigma^2) =\overline{ \mathrm{SURE}}_{\epsilon, \boldsymbol{\Delta}}^R(\boldsymbol{z}; \Theta \lvert \sigma^2)$ and $\partial_{\Theta} \widehat{Q}(\boldsymbol{z}, \Theta \lvert \sigma^2) = \overline{\mathrm{SUGAR}}_{\epsilon, \boldsymbol{\Delta}}^R(\boldsymbol{z}; \Theta \lvert \sigma^2)$, defined respectively in Eq.~\eqref{eq:mcsure}~and~\eqref{mcsugar}, whose practical implementation is based on Algorithm~\ref{dslpam}, provides the optimal hyperparameters.
The overall procedure is referred to as \textit{Averaged SUGAR D-MS}.
Note that, for $R = 1$, $\boldsymbol{\Delta} = \lbrace \boldsymbol{\delta}^{(1)}\rbrace$, and one retrieves the standard SURE and SUGAR estimates. 
In the case when $R = 1$, the automated hyperparameter strategy is hence referred to as \textit{Standard SUGAR D-MS}.

\section{Performance assessment}
\label{sec:num_exp_BP}

\subsection{Settings}

\noindent\textbf{Data -- } 
To assess the relevance of SURE~\eqref{eq:fdmcsure} in the context of interface detection, as well as the efficiency of the proposed automated minimization making use of the SUGAR proposed in Section~\ref{sec:hyperparameters}, systematic experiments are performed on the test data displayed in Fig.~\ref{fig:denoisedgeometries} and 
 several noise levels are explored, corresponding to  
$\sigma \in \{ 0.01,0.05,0.1\}$. 
Additional experiments with other geometries and level of noise are provided in Appendix~\ref{sec:add_exp}.

\noindent\textbf{Algorithmic setup -- } See Supplementary materials.
 
\noindent\textbf{Performance criteria --} 
In practice, standard and averaged SURE are compared to the following \textit{quadratic error}: 
    $
    	 \mathcal{Q}( \widehat{\boldsymbol{u}} \lvert  \overline{\boldsymbol{u}} ) := \Vert \widehat{\boldsymbol{u}} -  \overline{\boldsymbol{u}}  \Vert_2^2.
    $ 
    To assess the performance of D-MS denoising with automatically selected hyperparameters, the quality of the reconstruction is quantified by the peak signal-to-noise ratio defined as:
    $
        \mathrm{PSNR}(\widehat{\boldsymbol{u}}\lvert \overline{\boldsymbol{u}})  = 20 \mathrm{log}_{10}\left( \frac{\Vert \overline{\boldsymbol{u}} \Vert}{\Vert \widehat{\boldsymbol{u}} - \overline{\boldsymbol{u}} \Vert} \right).
 $

\subsection{SURE for D-MS}

We first illustrate in Fig. \ref{fig:sure_grids} the asymptotic unbiasedness of the standard and averaged SURE 
on the example $\boldsymbol{z}$ displayed in Fig.~\ref{fig:denoisedgeometries} (top-middle) with noise level $\sigma = 0.05$. 
To better locate and compare the minima, three level sets of SURE are displayed by the \textsc{Matlab} function \emph{contour}.

Even though the overall shape of the standard SURE 
maps are similar to the quadratic error profile, 
Fig.~\ref{fig:sure_grids}(a-c) shows that the location of the minimum varies significantly with the Monte Carlo vector $\boldsymbol{\delta}^{(r)}$. 
Averaged SURE  
also well reproduces the quadratic error map while being more robust to achieve the minimum (cf. Fig.~\ref{fig:sure_grids}(d)). 

This first set of experiments illustrate that the proposed averaged SURE reproduces accurately the quadratic risk as expected from Proposition~\ref{prop:unbavsure}

\begin{figure}[!t]
\centering
\subfloat{\includegraphics[width=0.2\linewidth]{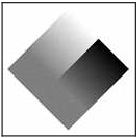}} \hspace{1cm}
\subfloat{\includegraphics[width=0.2\linewidth]{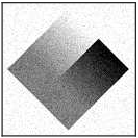}}\hspace{1cm}
\subfloat{\includegraphics[width=0.2\linewidth]{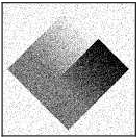}}\\
\subfloat{\includegraphics[width=0.2\linewidth]{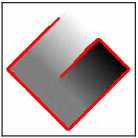}} \hspace{1cm}
\subfloat{\includegraphics[width=0.2\linewidth]{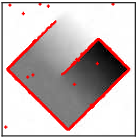}} \hspace{1cm}
\subfloat{\includegraphics[width=0.2\linewidth]{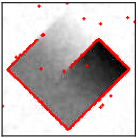}}\\
\caption{Piecewise smooth grey level image corrupted by i.i.d. Gaussian noise with level $\sigma \in \{0.01,0.05,0.1\}$. Associated D-MS estimates $\widehat{\boldsymbol{u}}$ and $\widehat{\boldsymbol{e}}$ (superimposed in red)  
The D-MS hyperparameters are selected with the proposed  \textit{Averaged SUGAR D-MS} (with $R=5$) using the true standard deviation $\sigma$. }
\label{fig:denoisedgeometries}
\end{figure}

\begin{figure*}[!t]
\centerline
{
{\footnotesize
\hspace{-.0cm} $\mathrm{SURE}_{\epsilon, \boldsymbol{\delta}^{(1)}}(\boldsymbol{z}; \Theta \lvert \sigma^2) $
\hspace{.5cm}$\mathrm{SURE}_{\epsilon, \boldsymbol{\delta}^{(2)}}(\boldsymbol{z}; \Theta \lvert \sigma^2)$ 
\hspace{.5cm} $\mathrm{SURE}_{\epsilon, \boldsymbol{\delta}^{(3)}}(\boldsymbol{z}; \Theta \lvert \sigma^2) $
\hspace{.5cm} 
$ \mathrm{SURE}_{\epsilon, \boldsymbol{\Delta}}(\boldsymbol{z}; \Theta \lvert \sigma^2)$ \hspace{.7cm}
$\mathbb{E}[\Vert \widehat{\boldsymbol{u}}(\boldsymbol{z};\Theta) - \overline{\boldsymbol{u}} \Vert_2^2] $
}
}
\centerline{
\subfloat[]{\includegraphics[width=.2\linewidth]{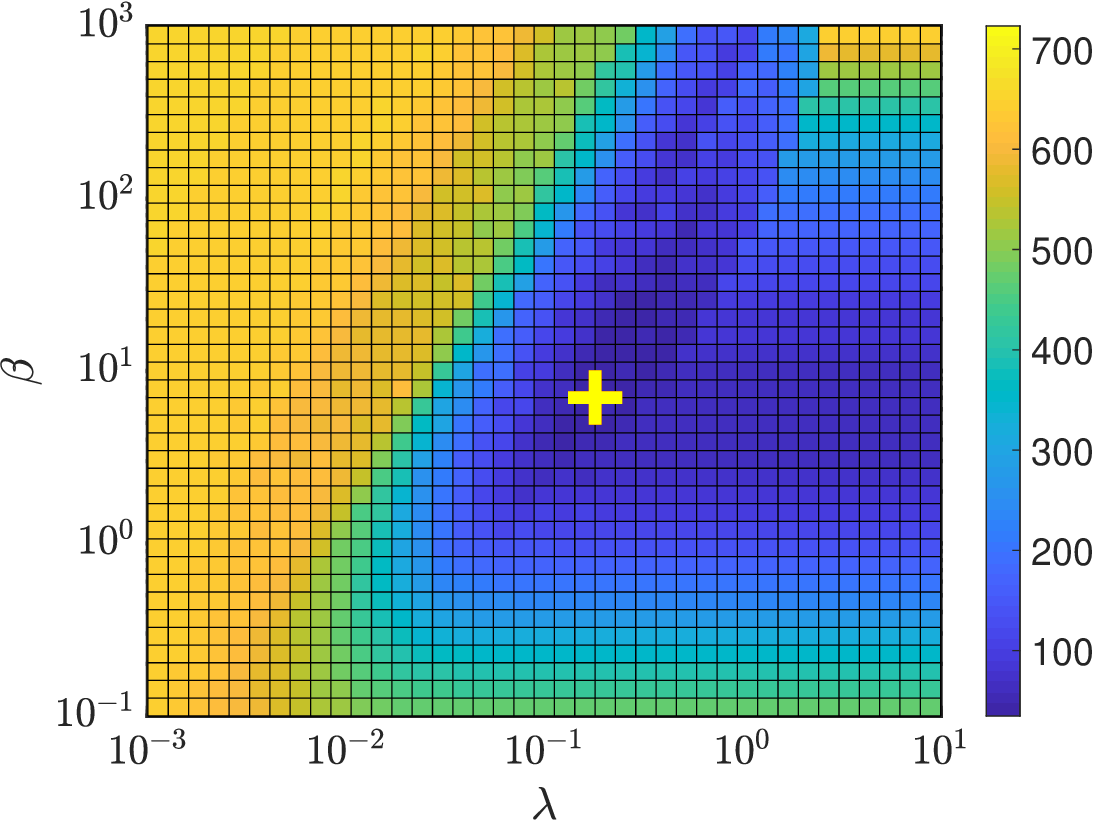}}
\subfloat[]{\includegraphics[width=.2\linewidth]{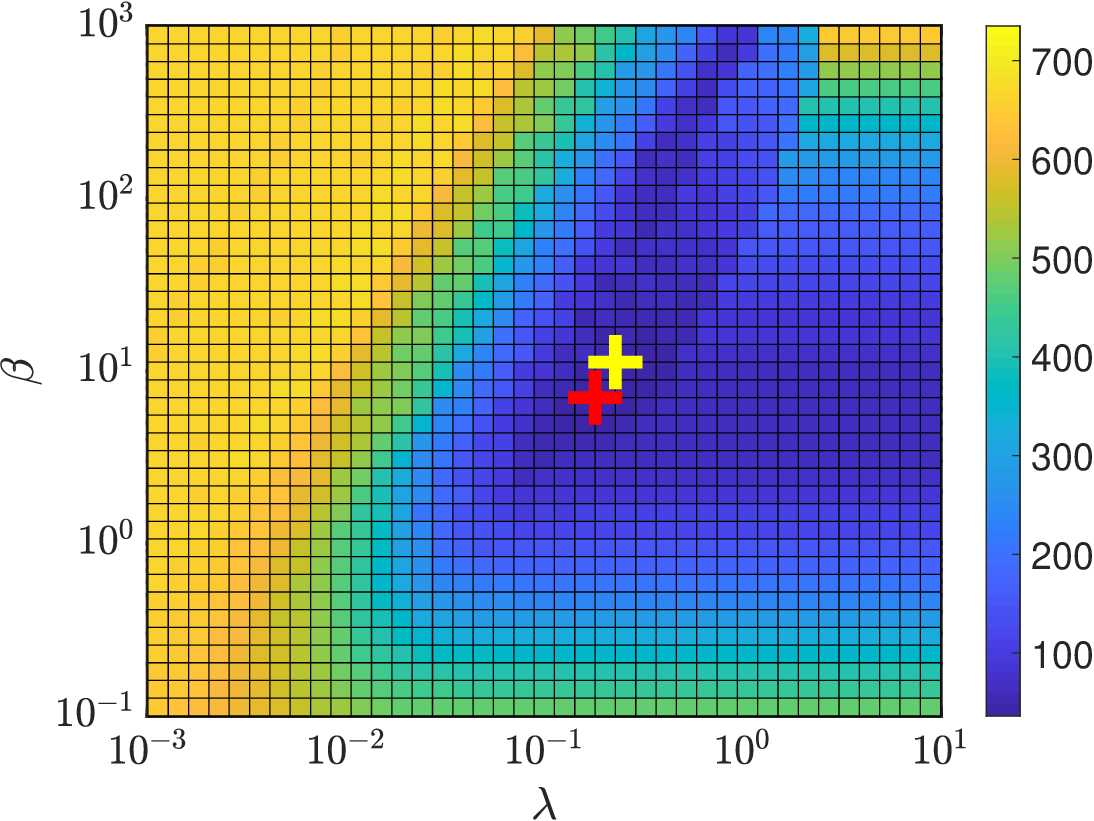}}
\subfloat[]{\includegraphics[width=.2\linewidth]{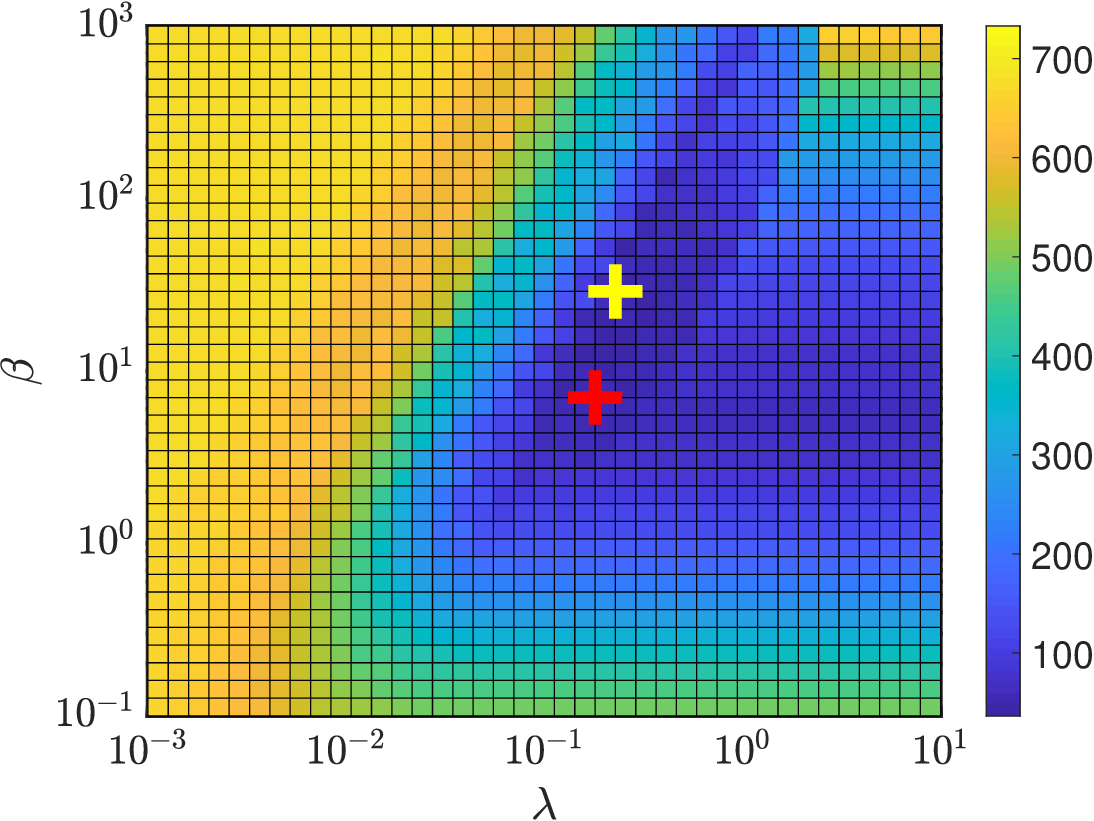}}
\subfloat[]{\includegraphics[width=.2\linewidth]{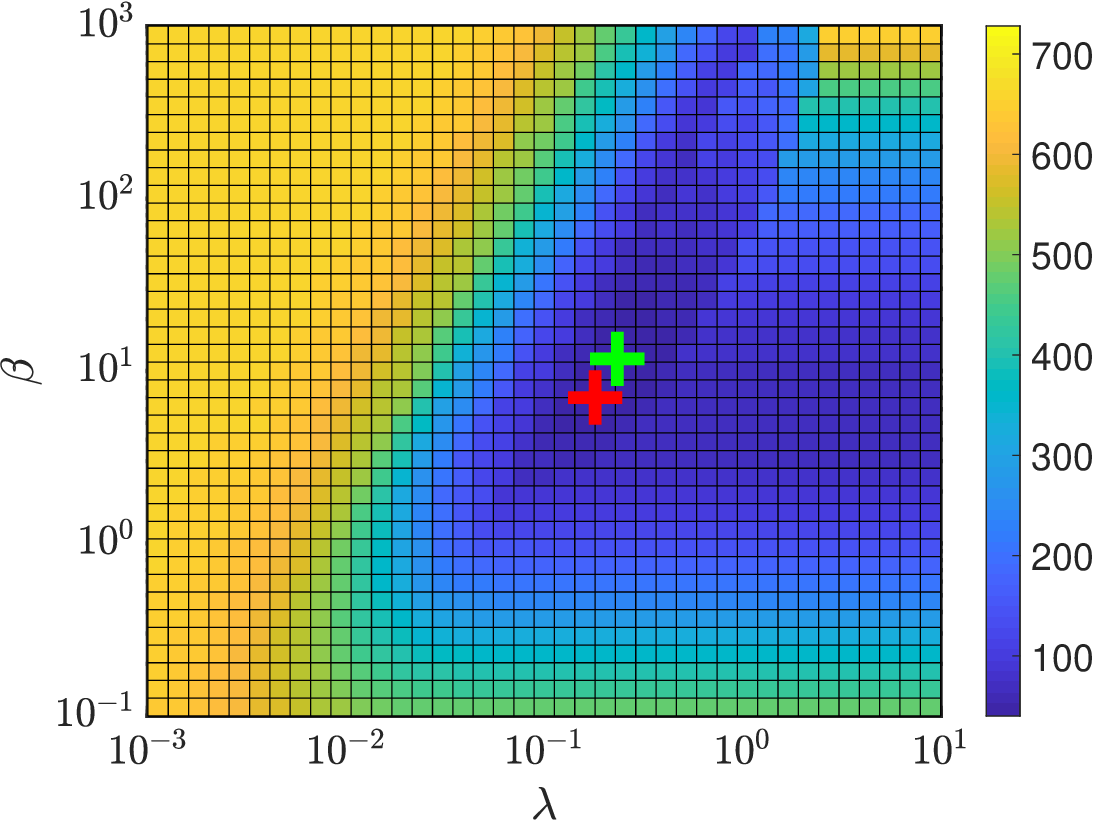}}
\subfloat[]{\includegraphics[width=.2\linewidth]{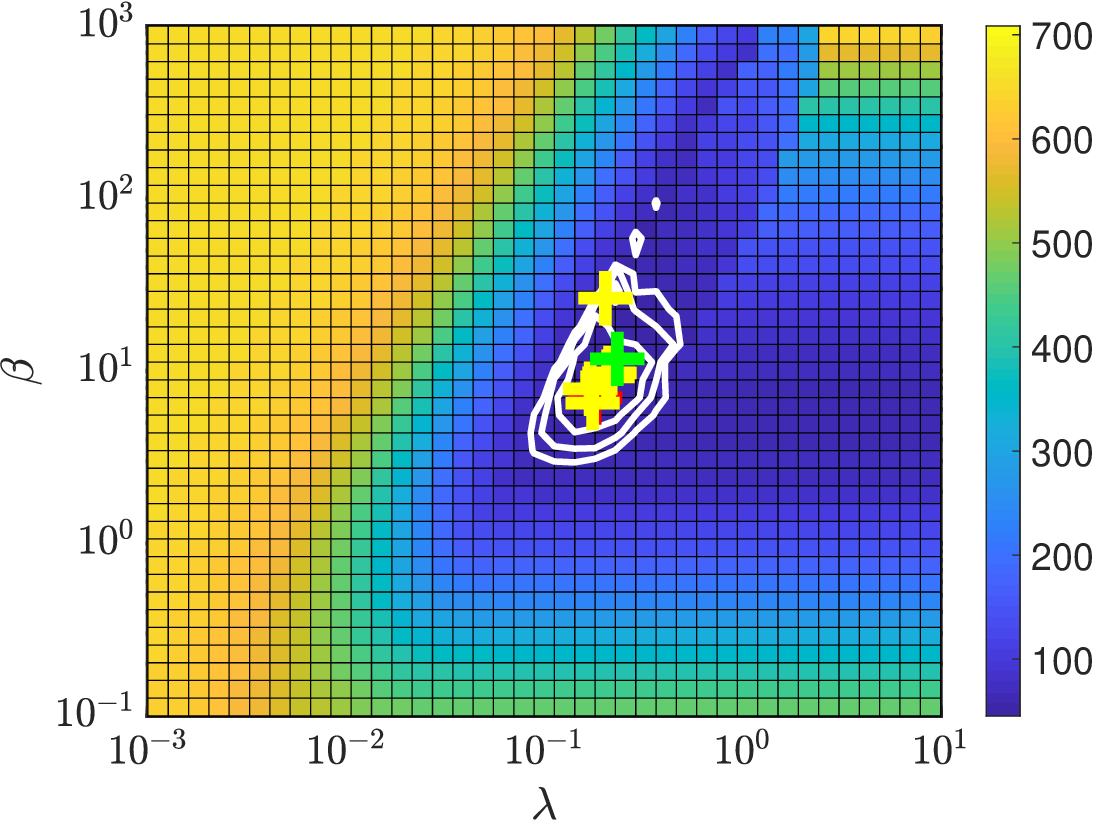}}
}
\vspace{3.5mm}
\centerline{
\subfloat[]{\includegraphics[width=.18\linewidth]{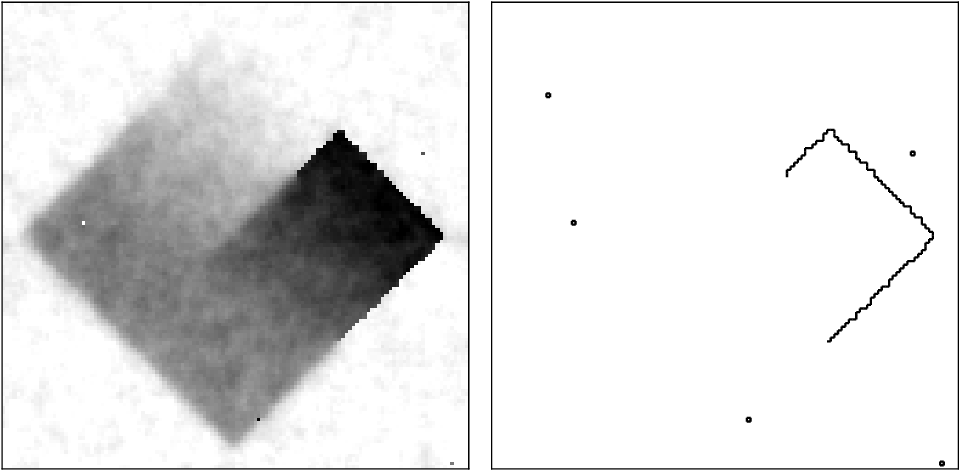}} 
\hspace{0.09in}
\subfloat[]{\includegraphics[width=.18\linewidth]{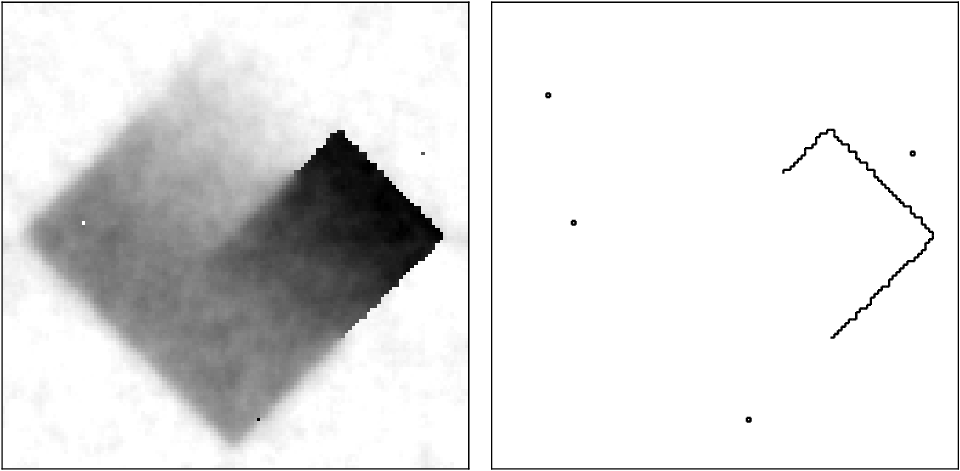}} 
\hspace{0.09in}
\subfloat[]{\includegraphics[width=.18\linewidth]{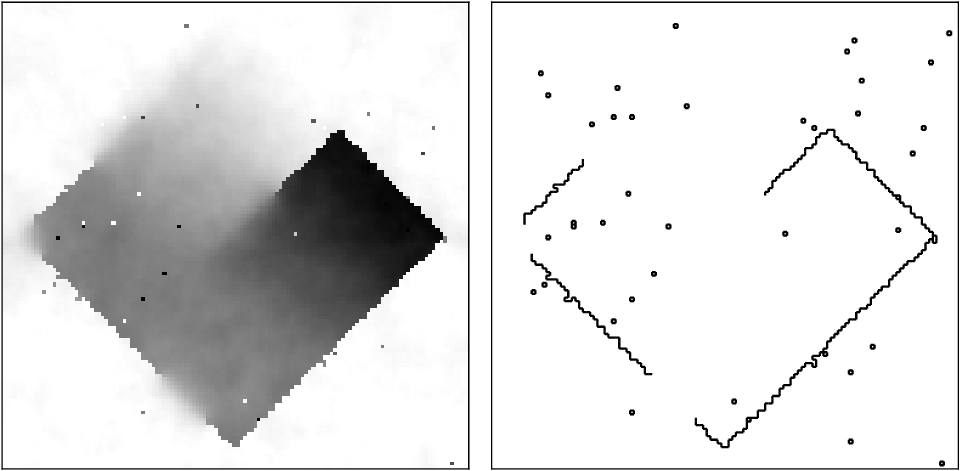}} 
\hspace{0.09in}
\subfloat[]{\includegraphics[width=.18\linewidth]{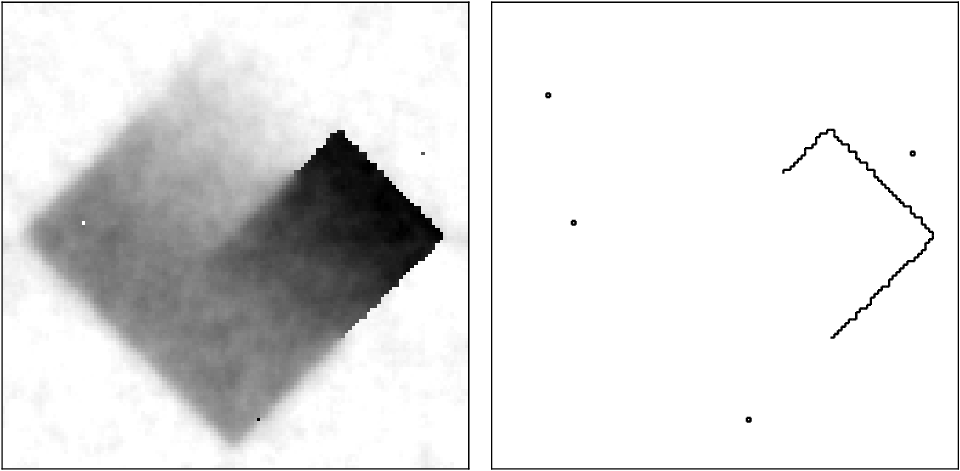}}
\hspace{0.09in}
\subfloat[]{\includegraphics[width=.18\linewidth]{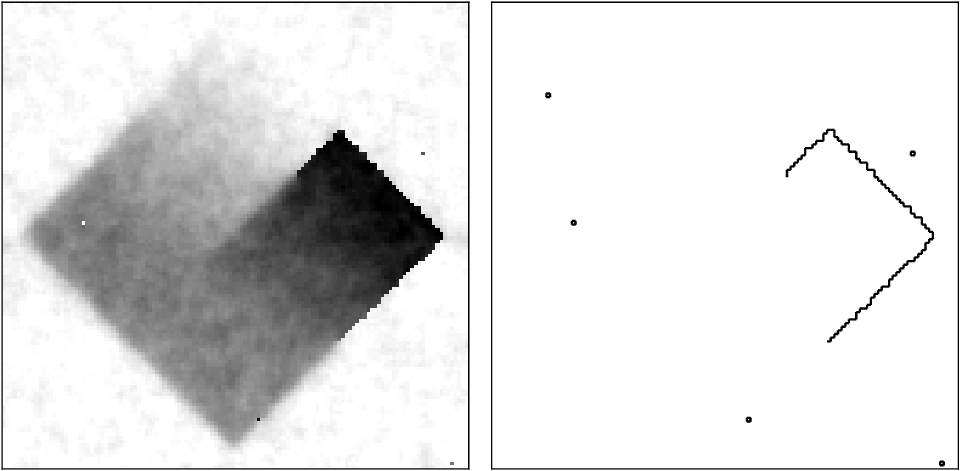}}
}
\caption{\textbf{Comparison between the quadratic error, standard and averaged SURE estimates for D-MS denoising of the image displayed in Fig.~\ref{fig:denoisedgeometries}(middle)}. \\ \textbf{1st row} -- Map on a logarithmic grid of $40 \times 40$ hyperparameters $\Theta = (\beta, \lambda)$: (a-c) $\mathrm{SURE}_{\epsilon, \boldsymbol{\delta}^{(r)}}(\boldsymbol{z}; \Theta \lvert \sigma^2)$ values for some realizations of the Monte Carlo vector, (d) $\mathrm{SURE}_{\epsilon, \boldsymbol{\Delta}}(\boldsymbol{z}; \Theta \lvert \sigma^2)$ values for $R=5$ realizations of the Monte Carlo vector  and  (e) quadratic error $\mathcal{Q}( \widehat{\boldsymbol{u}}(\boldsymbol{z}; \Theta )\lvert \overline{\boldsymbol{u}} )$ values  with level sets (black lines). \textbf{2nd row} -- Optimal solutions $\left(\widehat{\textbf{u}}(\boldsymbol{z};\Theta^{\textrm{Grid}}),\widehat{\textbf{e}}(\boldsymbol{z};\Theta^{\textrm{Grid}})\right)$ obtained from a grid search  over each map. The red  (resp. yellow and green) cross corresponds to the solution displayed in (j) (resp. (f)-(h) and (i)) associated with the minimum of the quadratic error grid (e) (resp. SURE estimate grids (a)-(c) and (d)).
}
\label{fig:sure_grids}
\end{figure*}

\begin{figure*}[!t]
\centering
\subfloat{\includegraphics[width=1.6in]{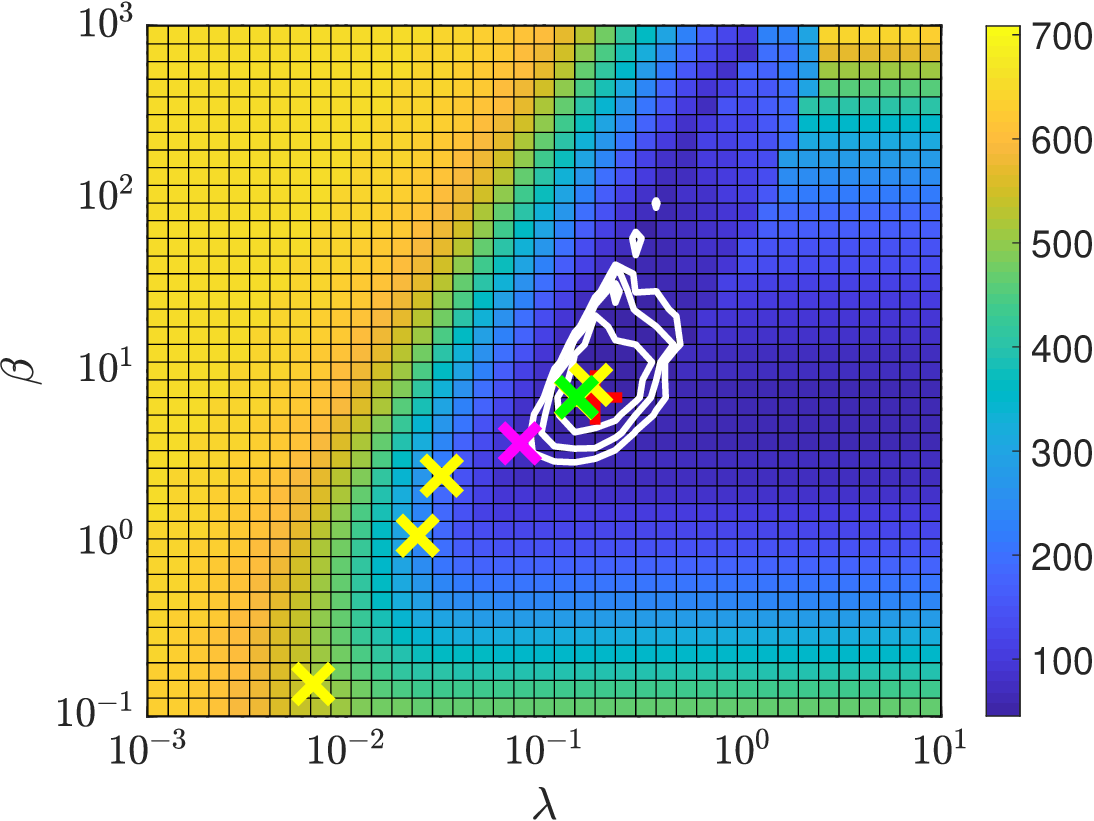}} \hspace{0.1in} 
\subfloat{\includegraphics[width=1.6in]{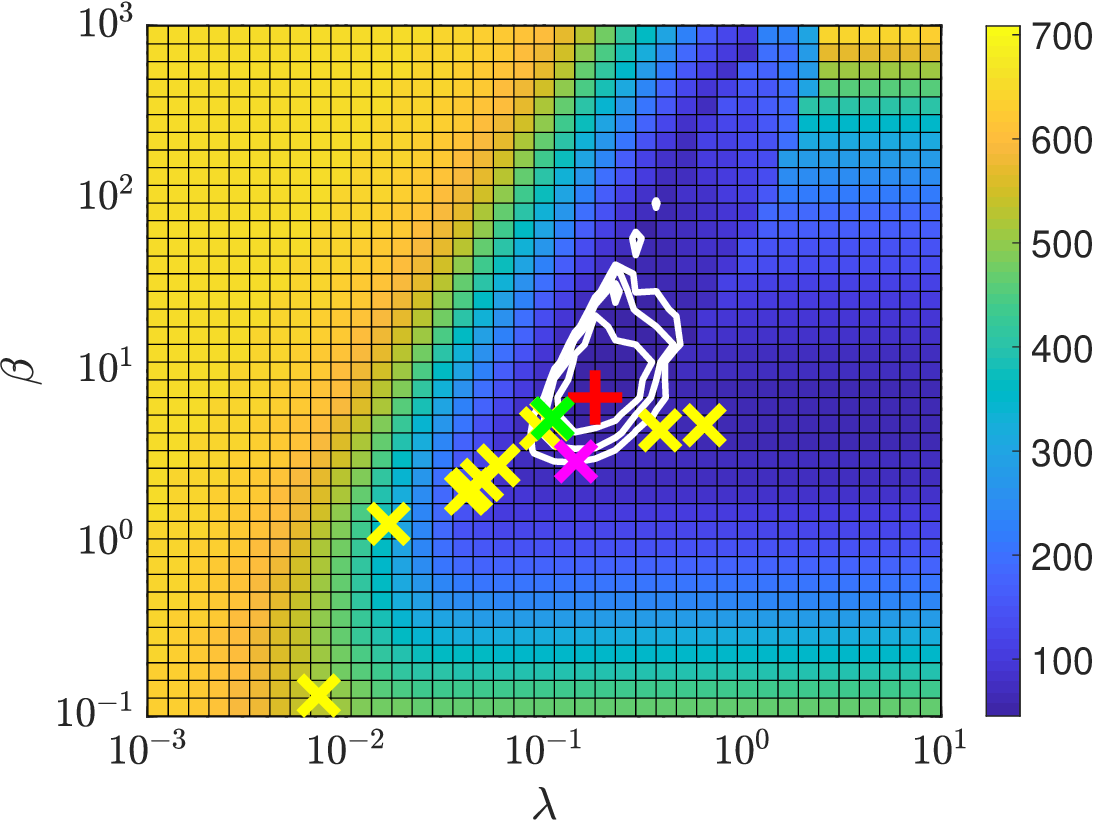}} \hspace{0.1in} 
\subfloat{\includegraphics[width=1.6in]{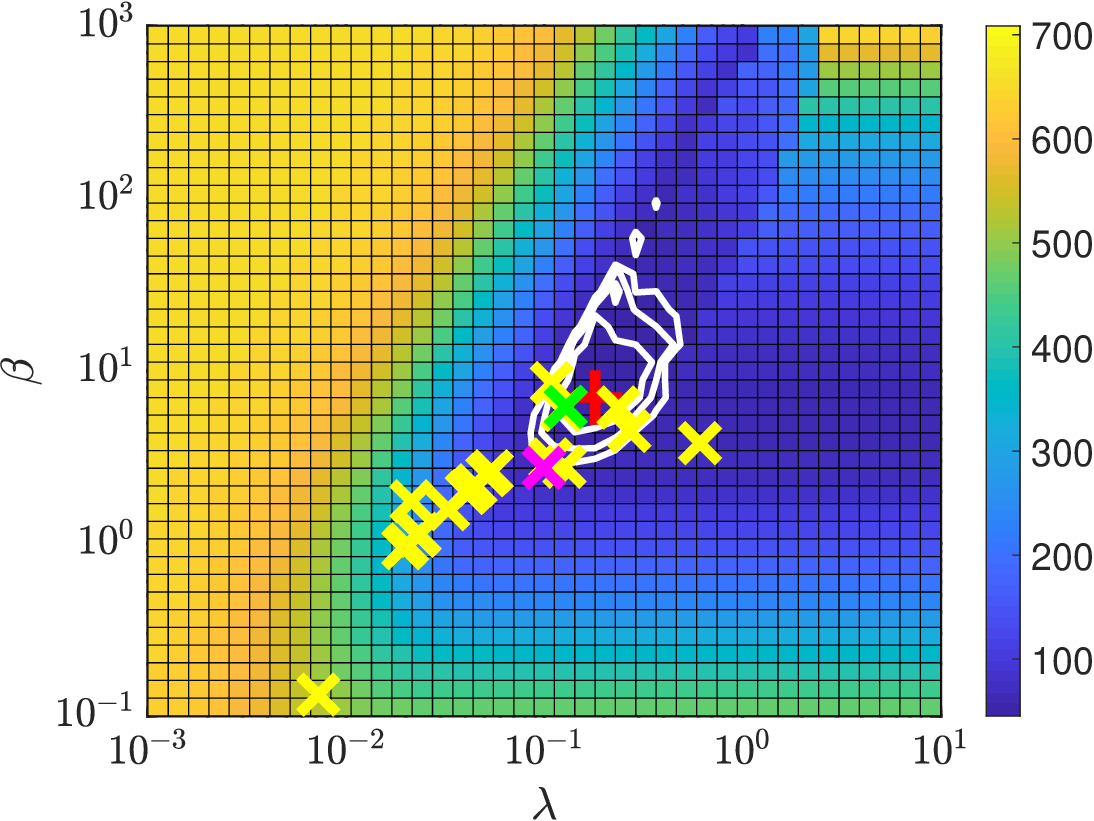}} 
\caption{\textbf{Impact of the number of realizations $R$ of the Monte Carlo vectors when selecting the hyperparameters with the methods described in Section~\ref{ss:algobfgs}}. (left) $R = 5$ , (middle) $R=10$  and (right) $R=20$. (yellow) \textit{Standard SUGAR D-MS} for different $\boldsymbol{\delta}^{(r)}$ leading to $\Theta^{*(r)}$, (pink) Mean over the $R$ realizations of  \textit{Standard SUGAR D-MS} leading to $\overline{\Theta^*}^{R}$, (green) \textit{Averaged SUGAR D-MS}, (red) optimum obtained by performing a grid search minimization of the quadratic error. For the 3 maps, the background displays the logarithmic grid of $40 \times 40$ hyperparameters $\Theta = (\beta,\lambda)$ of quadratic error  $\mathcal{Q}( \widehat{\boldsymbol{u}}(\boldsymbol{z}; \Theta )\lvert \overline{\boldsymbol{u}} )$ values with level sets (black lines).}
\label{fig:bfgs}
\end{figure*}

\begin{figure*}
\hspace{-1cm}{\scriptsize{\begin{tabular}{p{2.7cm}p{2.7cm}p{2.7cm}p{1.5cm}p{1.5cm}p{1.5cm}p{1.5cm}}
Degraded & SUGAR T-ROF & SUGAR D-MS & Original &Degraded & T-ROF &D-MS\\
 & (State-of-the-art) & (Proposed) & zoom &zoom &zoom &zoom \\
  &SSIM: 0.87 & SSIM: 0.86 &&&\\
    &Jacc: 0.31 & Jacc: 0.50 &&&\\
 \includegraphics[height=2cm,width=2.7cm]{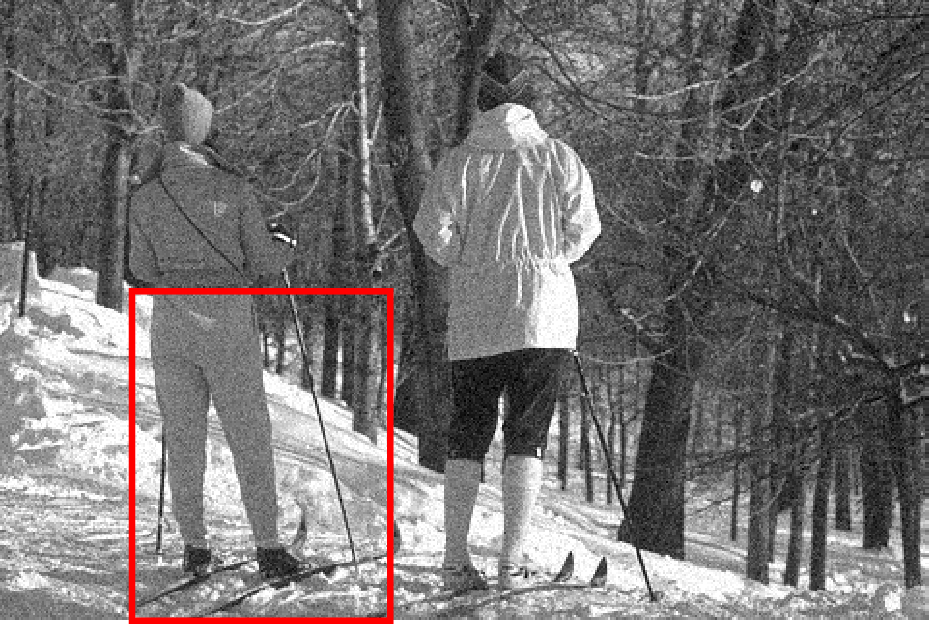}&
\includegraphics[height=2cm,width=2.7cm]{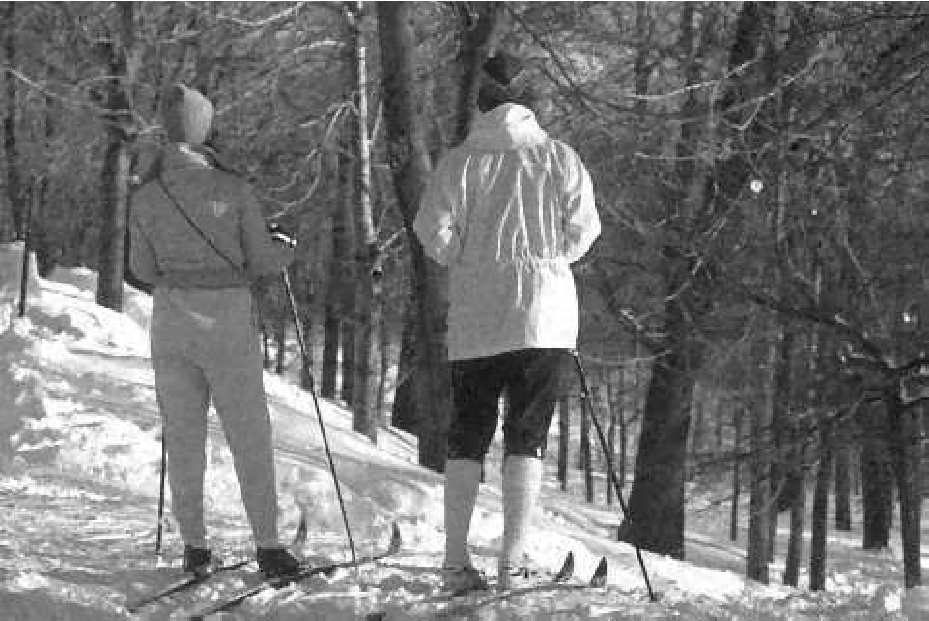}&
\includegraphics[height=2cm,width=2.7cm]{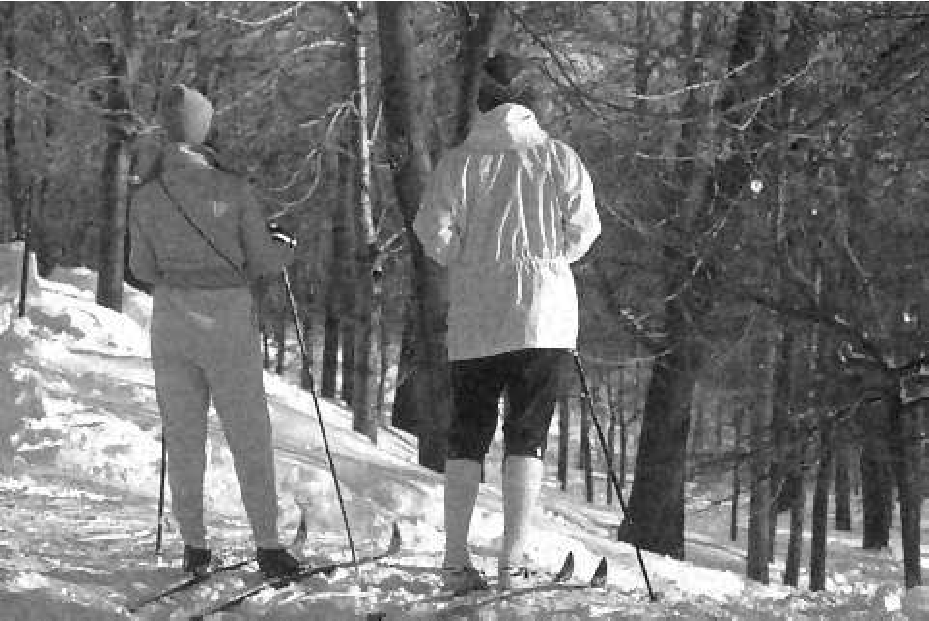}&
\includegraphics[height=2cm,width=1.5cm]{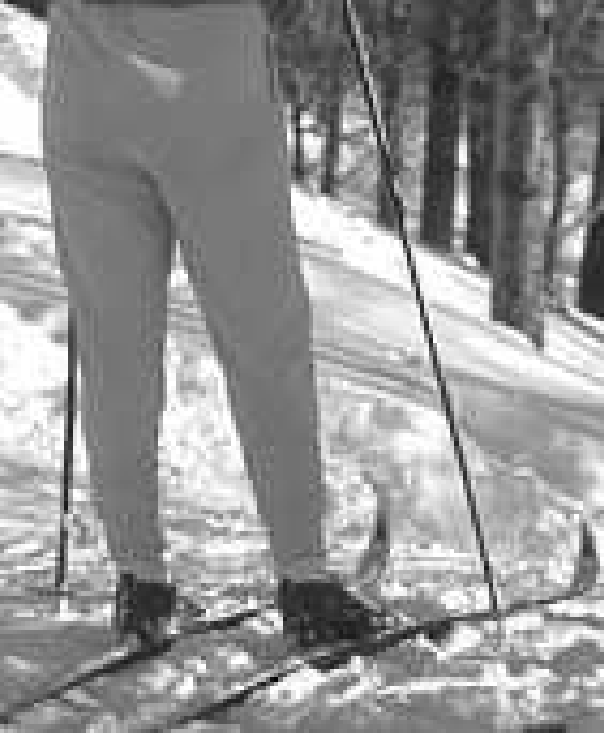}&
\includegraphics[height=2cm,width=1.5cm]{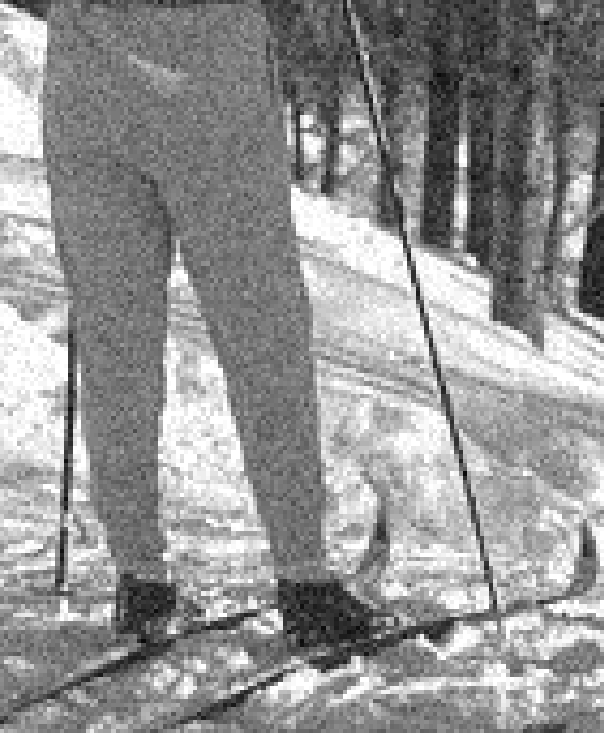}&
\includegraphics[height=2cm,width=1.5cm]{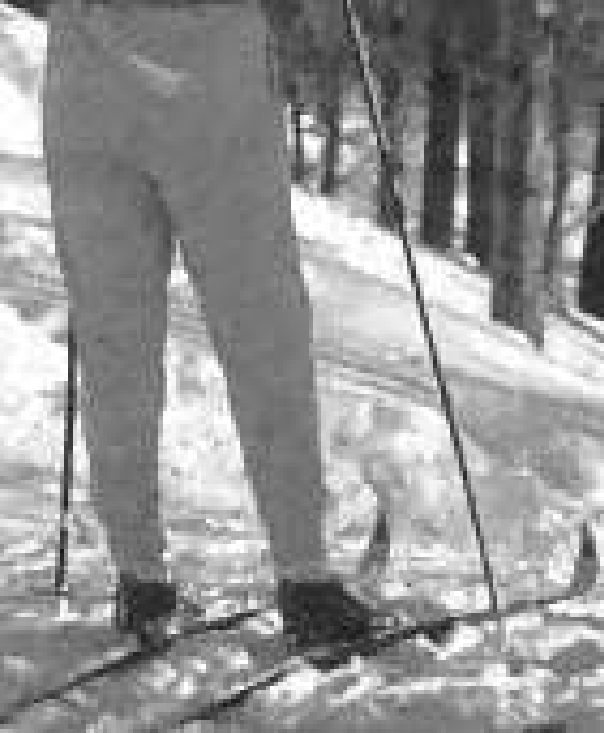}&
 \includegraphics[height=2cm,width=1.5cm]{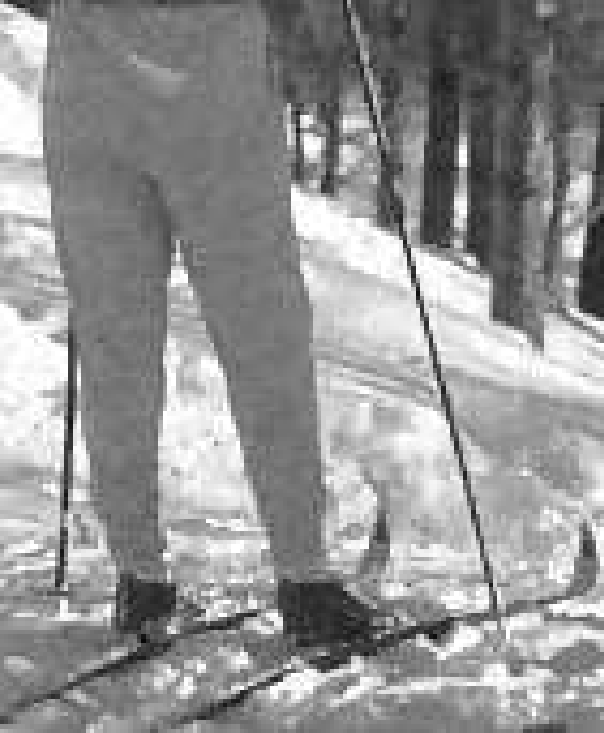}\\
&\includegraphics[height=2cm,width=2.7cm]{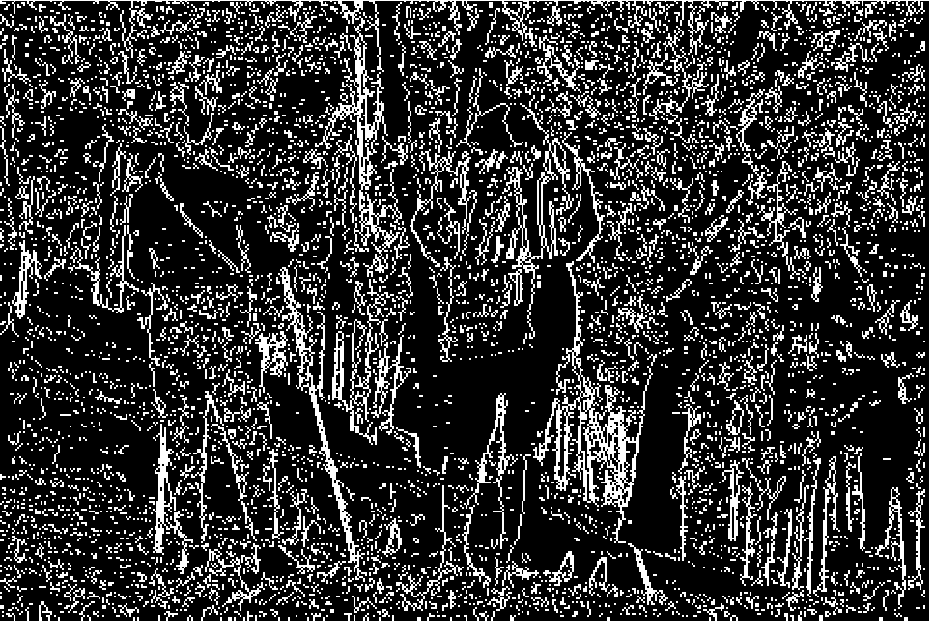}&
\includegraphics[height=2cm,width=2.7cm]{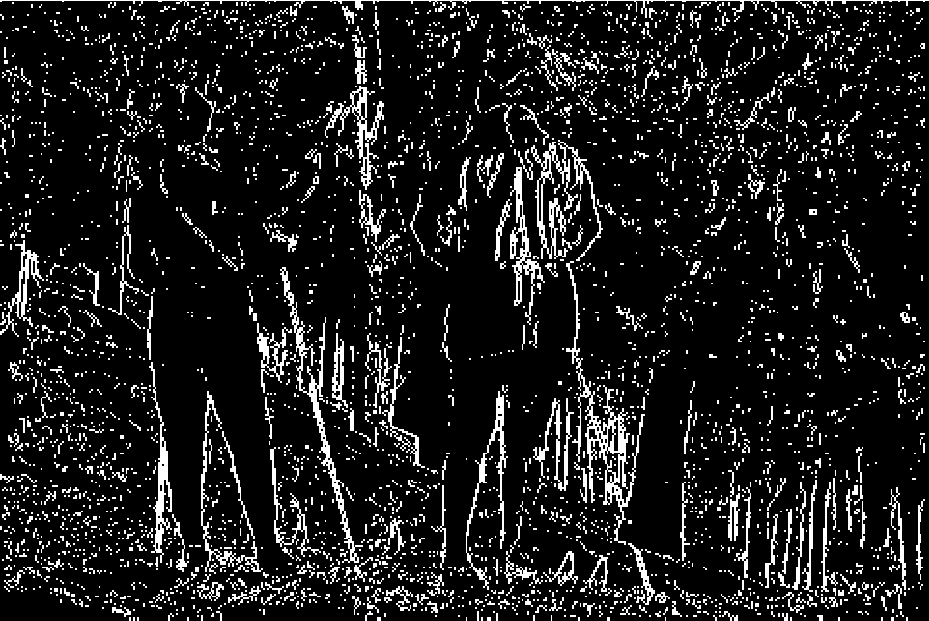} &
&
&\includegraphics[height=2cm,width=1.5cm]{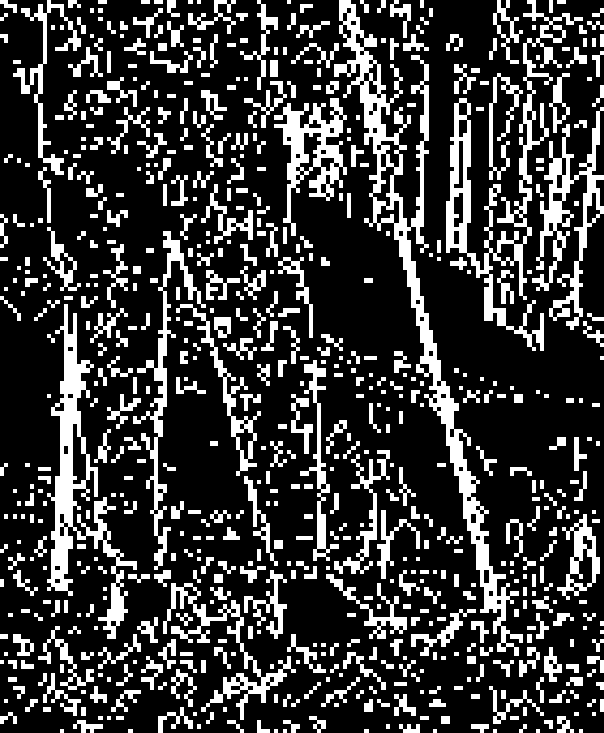}&
 \includegraphics[height=2cm,width=1.5cm]{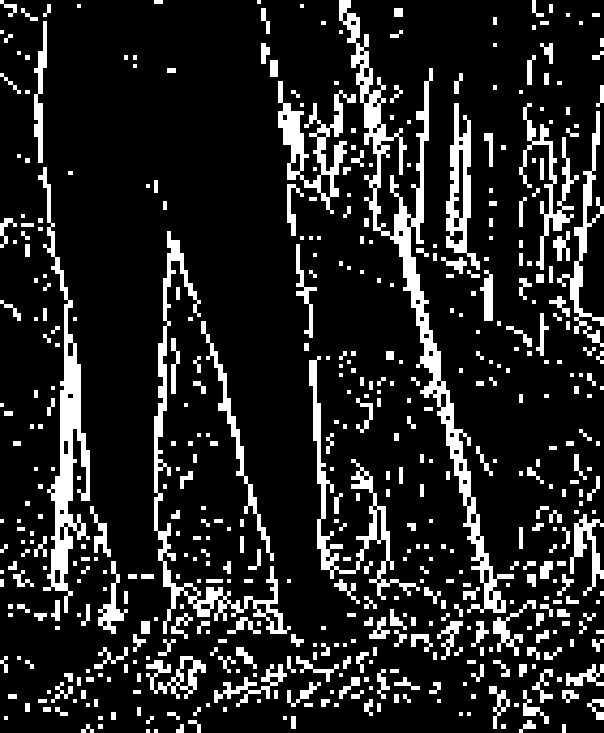}\\
\end{tabular}}}
\caption{Comparisons between SUGAR T-ROF and SUGAR D-MS for noisy images extracted from the BSD69 dataset \cite{Roth2009Fields}.\label{fig:real} }
\end{figure*}

\subsection{Comparison between \textit{Standard} and \textit{Averaged SUGAR D-MS}}
Fig. \ref{fig:bfgs} investigates the ability of the hyperparameter selection strategies proposed in Section~\ref{ss:algobfgs} for different numbers $R \in \{ 5, 10, 20 \}$ of Monte Carlo vector $\boldsymbol{\delta}^{(r)}$ to achieve the optimal hyperparameters minimizing the quadratic error $\mathcal{Q}( \widehat{\boldsymbol{u}}\lvert \overline{\boldsymbol{u}} )$. 

 The optimal hyperparameters $\Theta^{*(r)} = (\beta^{*(r)},\lambda^{*(r)})$ reached by the \textit{Standard SUGAR D-MS} are scattered (Fig.~\ref{fig:bfgs} left), probably due to a lack of accuracy of the estimator $\mathrm{SUGAR}_{\epsilon, \boldsymbol{\delta}} = \overline{\mathrm{SUGAR}}_{\epsilon, \boldsymbol{\Delta}}^{R=1}$. A first approach to alleviate the variability of the result is to carry out an averaging of $R$ hyperparameters $(\beta^{*(r)},\lambda^{*(r)})$ obtained by the \textit{Standard SUGAR D-MS} method: 
\begin{equation}
	\overline{\Theta^*}^{R} = (\overline{\beta^*}^{R}, \overline{\lambda^*}^{R} ) = \frac{1}{R} \sum_{r=1}^R (\beta^{*(r)},\lambda^{*(r)}).
\end{equation}
 As it can be observed in Fig. \ref{fig:bfgs} left, this improvement of the method remains unsatisfactory, compared to \textit{Averaged SUGAR D-MS} 
 which reaches more accurate hyperparameters.

The conclusions reached with this set of experiments  are twofold: first, we highlight that $R=5$ realizations are sufficient to achieve a good estimation of the optimal hyperparameters, second, we note that the proposed automated procedure is $20$ times faster compared to exhaustive search,  a grid search on averaged SURE requiring $60$ minutes of calculation, while \textit{Averaged SUGAR D-MS} requires $3$ minutes, when using \textsc{Matlab} R2018a and an Intel Core i5 processor.

Few estimated images obtained with the fully unsupervised parameter-free \textit{Averaged SUGAR D-MS} are provided in Fig.~\ref{fig:denoisedgeometries} (2nd row).

\subsection{Real-world images}
The proposed automated joint denoising and contour detection procedure \textit{Averaged SUGAR D-MS}  is evaluated on real-world images extracted from BSD69 dataset \cite{Roth2009Fields} degraded with a Gaussian noise with $\sigma=0.05$. In our experiments we set $R=5$ and $\sigma$ has been estimated from noisy data following Eq.~\eqref{eq:sigest} in Appendix~\ref{ssec:sig_est}.
Denoised images and contours provided by the proposed data-driven \textit{Averaged SUGAR D-MS} strategy are compared with those yield by SUGAR T-ROF (a two-step procedure, consisting in, first, a piecewise constant denoising with automated tuning of the regularization parameter~\cite{deledalle2014stein}, followed by an iterative thresholding procedure~\cite{cai2013multiclass}). In Fig.~\ref{fig:real}, we can observe that the denoising performance are very close for both procedures (in terms of SSIM, SUGAR T-ROF is slightly better)  while the contour detection is significantly improved with the SUGAR D-MS procedure (which is confirmed when computing Jaccard index w.r.t contours obtained from the original image). 

\section{Conclusion}

This work devises a procedure to automatically select the hyperparameters of the D-MS functional allowing to perform simultaneously image denoising and contour detection. This approach is fully unsupervised compared to alternative deep learning strategies such as \cite{bertasius2015deepedge} and reference therein. However, in a future work, it would probably benefit to combine D-MS functional with unfolded deep learning strategies in order to design efficient supervised combined denoising and contour detection approaches.

A \textsc{Matlab} toolbox implementing the proposed automated image denoising and contour detection procedure is publicly available\footnote{\url{https://github.com/charlesglucas/sugar_dms}}.

\appendix


\section{Notations} 
\label{sec:notations}
Let $\mathcal{H}$ a real Hilbert space, and $f : \mathcal{H} \rightarrow (+\infty,+\infty]$ a function which is proper, convex, and lower-semicontinuous and $\tau >0$ a real parameter, the proximity operator of $\tau f$ at point $\boldsymbol{v} \in \mathcal{H}$ is uniquely defined by $
\mathrm{prox}_{\tau f} (\boldsymbol{v}) = \underset{\boldsymbol{u}\in \mathcal{H}}{\arg\min} \, \frac{1}{2} \lVert \boldsymbol{u} - \boldsymbol{v} \rVert_2^2 + \tau f(\boldsymbol{v})$. Additionally, let  $\mathcal{G}$ be a real Hilbert space and 
let $\mathcal{A}\colon \mathcal{H} \to \mathcal{G}$ a Lipschitzian map, we denote by $L_{\mathcal{A}}>0$ the Lipschitz modulus of $\mathcal{A}$,  such that, for every $(x,y)\in \mathcal{H}\times \mathcal{H}$, $\Vert \mathcal{A}(x) - \mathcal{A}(y) \Vert \leq L_{\mathcal{A}} \Vert x - y \Vert$.  
Further,  for every $(x,y)\in \mathbb{R}\times \mathbb{R}$, we denote $\mathcal{I}_{x>y}= 1$ if $x>y $ and $0$ otherwise.
Finally, $\textbf{I}_N$ denotes the identity matrix acting on $\mathbb{R}^{N}$, and $\boldsymbol{1}_N$ (resp. $\boldsymbol{0}_N$) is the vector of $\mathbb{R}^N$ containing only ones (resp. zeros).

\section{State-of-the-art for contour detection in image processing}

This work focuses on performing \textit{jointly} piecewise smooth denoising and contour detection on images. 
In many classical approaches, image reconstruction is embedded into a variational formalism~\cite{golub1999tikhonov,tikhonov2013numerical}, which amounts to find a minimizer of a functional consisting of the sum of a data fidelity term and a prior penalization, i.e.,
\begin{align}
\label{eq:BZ}
\underset{\boldsymbol{u}}{\mathrm{minimize}} \, \frac{1}{2}\lVert \boldsymbol{u} - \boldsymbol{z}\rVert^2_2 + \gamma p(\boldsymbol{\mathrm{D}} \boldsymbol{u})
\end{align}
where $\gamma > 0$, $\boldsymbol{z} \in \mathbb{R}^{\lvert \Omega \rvert}$ denotes the observed degraded image, defined on a grid of pixels $\Omega$, and $\boldsymbol{\mathrm{D}}\colon \mathbb{R}^{\lvert \Omega \rvert} \to \mathbb{R}^{\vert \mathcal{E}\vert}$ is a discrete difference operator such that $\boldsymbol{\mathrm{D}}\boldsymbol{u}$ lives on a lattice of contours $\mathcal{E}$.
Appropriate choices of the penalization term $p$, yield e.g. the Potts functional, when $p = \lVert \cdot \rVert_0$, or the Blake and Zisserman functional \cite{blake1987visual,chambolle1995image}, corresponding to $p(\boldsymbol{\mathrm{D}} \boldsymbol{u}) = \sum_b \min\lbrace \lVert \boldsymbol{\mathrm{D}}_b \boldsymbol{u}\rVert_q^q, \chi^q\rbrace$, for some $q \in [1, \infty)$ and $\chi>0$, with $\boldsymbol{\mathrm{D}}_b$ being associated with several rows of $\boldsymbol{\mathrm{D}}$. 
In the same vein, considering a convex relaxation of Potts functional, contour detection can be obtained from the minimization of the Rudin-Osher-Fatemi (ROF) functional~\cite{rudin1992nonlinear}, which favors piecewise constant estimate when considering $p(\boldsymbol{\mathrm{D}}\boldsymbol{u}) =   \sum_b \Vert \boldsymbol{\mathrm{D}}_b \boldsymbol{u} \Vert_2$.  An alternative solution relies on a bi-convex  formulation that can trace back to the Mumford-Shah~\cite{mumford1989optimal} or Geman and Geman functionals~\cite{geman1984stochastic}, which may be written in the discrete variational formulation setting as:
\begin{equation} 
\label{eq:dms_univariate}
\underset{\boldsymbol{u} \in \mathbb{R}^{\lvert \Omega \rvert}, \boldsymbol{e}\in \mathbb{R}^{\vert \mathcal{E}\vert}}{\textrm{minimize}}   \frac{1}{2} \Vert \boldsymbol{u} - \boldsymbol{z} \rVert_2^2 + \beta \lVert (1-\boldsymbol{e}) \odot \boldsymbol{\mathrm{D}} \boldsymbol{u} \rVert_2^2 +  \lambda h(\boldsymbol{e}),
\end{equation}
where $\odot$ denotes the component-wise product, $h$ denotes a convex function enforcing sparsity  and $\beta>0$ and $\lambda>0$ are regularization parameters.
This Discrete Mumford-Shah (D-MS) functional provides a piecewise-smooth reconstructed image $\widehat{\boldsymbol{u}}$ as well as estimated sparse contours $\widehat{\boldsymbol{e}}$.

To achieve segmentation into $K$ regions, Cai and Steidl designed an iterated thresolding strategy~\cite{cai2013multiclass} applied as a post-processing onto the minimizer of ROF functional.
The resulting state-of-the-art two-step procedure,  referred to as Thresholded ROF (T-ROF), was proven to be equivalent to minimizing the $K$-region piecewise constant Mumford-Shah functional.
From this thresholded solution, it is then straightforward to identify the contours of the image.
However, such an \textit{indirect} contour extraction procedure restricts to \textit{closed} contours.
Fig.~\ref{fig:trof} shows a comparison between D-MS and T-ROF methods on a piecewise smooth image.
The Mumford-Shah estimate is piecewise smooth preserving the discontinuities of the image while the ROF estimate is piecewise constant, leading to staircasing effects.
We observe that T-ROF erroneously detects interfaces in areas on which the image is piecewise smooth, as opposed to the D-MS whose estimated contour variable is approximately zero everywhere except at the location of the actual signal discontinuity.

\begin{figure*}[!h]
\centering
\subfloat[Original image]{\includegraphics[width=0.2\linewidth]{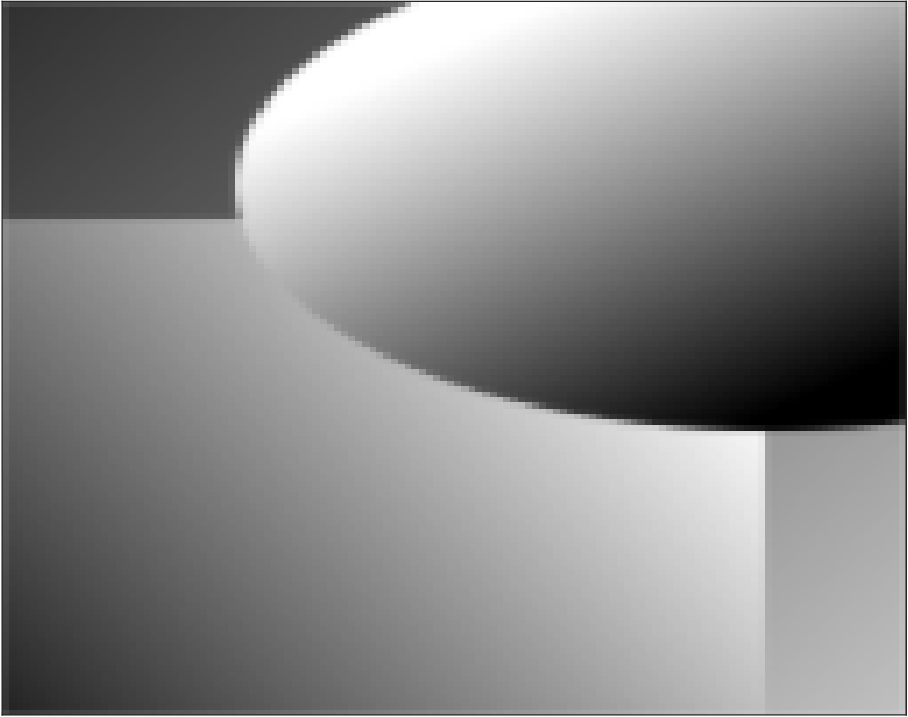}}
\hspace{0.2in}
\subfloat[Noisy Observation]{\includegraphics[width=0.2\linewidth]{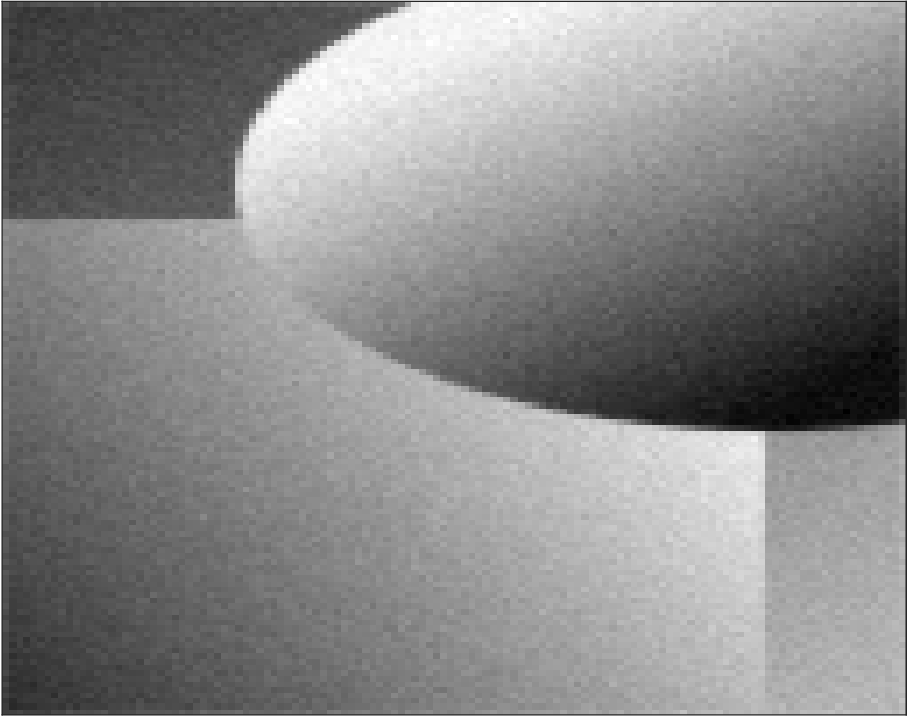}}
\hspace{0.2in}
\subfloat[T-ROF]{\includegraphics[width=0.2\linewidth]{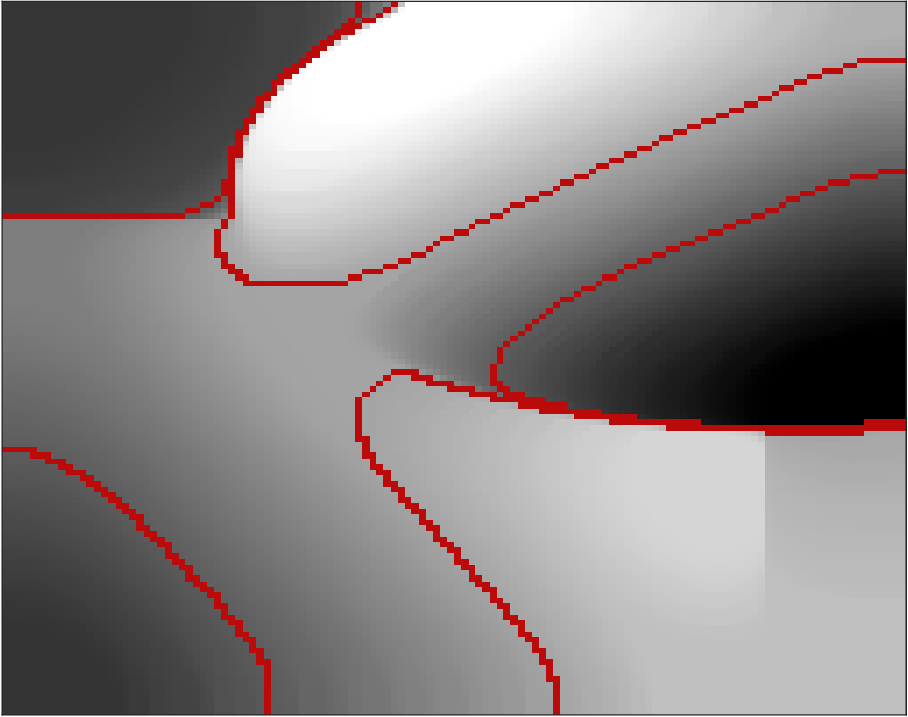}}
\hspace{0.2in}
\subfloat[D-MS]{\includegraphics[width=0.2\linewidth]{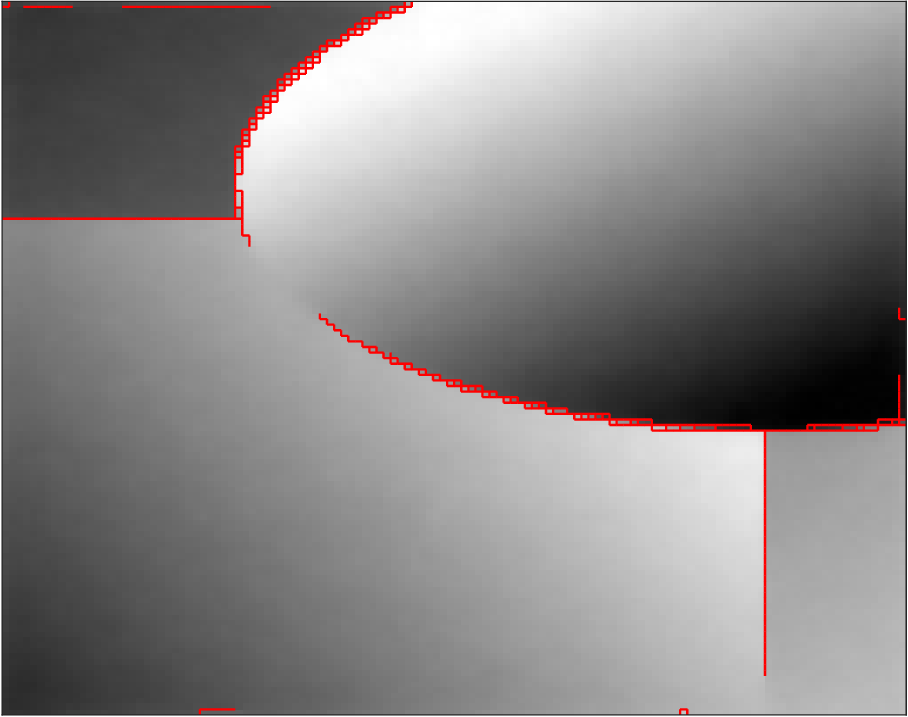}}
\caption{
Comparison of state-of-the-art \textit{convex} variational formulation T-ROF and the studied \textit{non-convex} D-MS performing image denoising and contour extraction. From left to right: (a)~Original noise-free piecewise smooth image, (b)~Observations $ \boldsymbol{z}$ corrupted by an additive Gaussian noise, (c) State-of-the-art ROF piecewise constant estimate and contours derived from thresholding into $K = 3$ regions (displayed in red), and (d)~Studied D-MS piecewise smooth approximation and estimated contours (displayed in red). 
}
\label{fig:trof}
\end{figure*}

\section{State-of-the-art for Hyperparameter selection}
\label{sec:soahyper}
 All aforementioned procedures for image denoising and contour detection involve \textit{hyperparameters}, e.g. $\beta$ and $\lambda$ in~\eqref{eq:dms_univariate}.
To reach satisfactory performance, the fine-tuning of these parameters is crucial.
Although central in signal and image processing, this difficult task is still an ongoing challenge, particularly for {variational} methods.

A first class of methods relying on hierarchical Bayesian approaches and has been widely used, both in signal and image processing~\cite{Molina_2001_image_resto_astronomy,Babacan_S_2009_ieee-tip_Variational_bbdutp, Dobigeon2007b,  vacar2019unsupervised}.
The drawbacks of Bayesian methods are that they rapidly become computationally heavy as the model for observed data gets more complicated, and their computational cost increases with the number of hyperparameters to be tuned. 
For specific 1D denoising problems, efficient hybrid variational/Bayesian strategies can be designed \cite{freconbayesian}.

Several other classes of methods, such as \textit{cross-validation} or Stein Unbiased Risk Estimate (SURE) formulation, can be formulated as a bilevel optimization problem. 
Cross-validation relies on a given labeled data set composed of noisy samples with their associated ground truth~\cite{stone1978cross,golub1979generalized}.
However, in several real-world applications, such as medical imaging~\cite{marin_mammographic_2017} or nonlinear physics problems~\cite{pascal2020parameter}, obtaining a large enough labeled dataset is very challenging, if not impossible. 
Hence, SURE, initially proposed in~\cite{Stein_C_1981_j-annals-statistics_estimation_mmnd}, has long been favored for its combined simplicity and efficiency. 
Stein-based hyperparameter strategies rely on an additive Gaussian noise model to design an estimate of the \textit{inaccessible} true risk, defined as the quadratic error between the estimate and ground truth.
The major advantage of these approaches is that they do not require to access ground truth.
Then, the selection of optimal hyperparameters is done by minimizing SURE and by making use of  Finite Difference strategies~\cite{ye1998measuring,shen2002adaptive} or/and Monte Carlo averaging~\cite{girard1989fast, ramani2008monte, deledalle2014stein}, to yield tractable and fast implementation of Stein-based risk estimates.

However, the strategy to find the optimal hyperparameters for a specific criterion has a huge impact on the solution both in terms of quality assessment and in terms of computational load. 
The most standard approach consists in computing a chosen error criterion over a grid of parameters~\cite{donoho1994ideal,ramani2008monte, eldar2008generalized}, and to select the parameter of the grid for which the error is minimal. 
Such a grid search procedure suffers from a high computation cost, especially when dealing with $L\geq 2$ regularization parameters.
To circumvent this difficulty, efficient automated minimization methods are required.
It was early envisionned by Chaux~\textit{et~al.}~\cite{chaux2008nonlinear}, who proposed and assessed numerically an empirical descent algorithm for automated choice of regularization parameters, but with no convergence guarantee. 
A deeper theoretical analysis was then provided by Deledalle~\textit{et~al.}~\cite{deledalle2014stein}, evidencing sufficient conditions so that Stein Unbiased Risk Estimate is differentiable with respect to hyperparameters, thus enabling to define the Stein Unbiased GrAdient of the Risk (SUGAR) estimator and to provide a practical implementation based on an iterative differentiation strategy.
Combining SUGAR with a quasi-Newton descent procedure, a fast algorithm was designed to achieve optimal  hyperparameters selection for objective functions of the form~\eqref{eq:BZ}.
This strategy,  later extended in~\cite{eldar2008generalized,pascal2020automated} for correlated noise,  proved its efficiency for texture segmentation~\cite{pascal2020automated}, piecewise linear signal denoising~\cite{pascal2020parameter}, and in spatial-spectral deconvolution for large multispectral data~\cite{ammanouil2019parallel}.

\section{Minimization of the discrete Mumford-Shah functional}
        \label{ssec:SL-PAM}
        
The D-MS functional introduced in Eq.~\eqref{eq:dms_univariate} being nonconvex, standard proximal algorithms~\cite{combettes2011proximal,parikh2014proximal,Bauschke_H_2011_book_con_amo} cannot be used directly for its minimization.
However, the fact that the functional is separately convex with respect to each variable advocates the use of alternating schemes.
Among the vast variety of existing alternating algorithms benefiting from convergence guarantees~\cite{attouch2010proximal,bolte2014proximal,foare2019semi}, a numerically efficient procedure for the minimization of D-MS like functionals appears to be the Semi-Linearized Proximal Alternating Minimization (SL-PAM) scheme proposed in~\cite{foare2019semi}, whose iterations in the general setting of Problem~\ref{pb:gen-SL-PAM} are recalled in Algorithm~\ref{alg:SLPAM}.

\begin{problem}[Nonconvex and nonsmooth minimization]
\label{pb:gen-SL-PAM}
Let 
$f : \mathbb{R}^{\lvert \Omega \rvert} \rightarrow (-\infty, +\infty]$, $h :\mathbb{R}^{\vert \mathcal{E}\vert } \rightarrow (-\infty, +\infty]$ two proper lower semi-continuous functions and $g : \mathbb{R}^{\lvert \Omega \rvert} \times \mathbb{R}^{\vert \mathcal{E}\vert } \rightarrow (-\infty +\infty]$ a $C^1$ function. Let $\lambda>0$ and $\beta>0$. We aim to estimate:
\begin{equation} \label{eq:dms}
(\widehat{\boldsymbol{u}},\widehat{\boldsymbol{e}}) \in  \!\!\!\! \underset{\boldsymbol{u} \in  \mathbb{R}^{\lvert \Omega \rvert}, \boldsymbol{e} \in  \mathbb{R}^{\vert \mathcal{E}\vert } }{\textrm{Argmin}}  \!\!\!\!  \Psi(\boldsymbol{u},\boldsymbol{e}) :=  f(\boldsymbol{u}) +  \beta g(\boldsymbol{u},\boldsymbol{e}) +   \lambda h(\boldsymbol{e}).
\end{equation}
\end{problem}

The algorithmic scheme SL-PAM (Algorithm~\ref{alg:SLPAM}) is an hybrid version between PAM~\cite{attouch2010proximal} and  PALM~\cite{bolte2014proximal}. 
The key ingredient for the efficiency of SL-PAM consists in avoiding the linearization 
with respect to the variable $\boldsymbol{e}^{[k]}$, enabling to choose larger descent steps. Under some technical assumptions, such as the existence of a closed-form expressions of the involved proximity operators, the sequence $\left( \boldsymbol{u}^{[k]},\boldsymbol{e}^{[k]}\right)_{k\in \mathbb{N}}$ converges toward a critical point of $\Psi(\boldsymbol{u}, \boldsymbol{e})$.

\begin{algorithm} [h!]
\caption{SL-PAM \label{alg:SLPAM}}
           \begin{algorithmic}
            \State \textbf{Initialization:} \; $\boldsymbol{u}^{[0]}=\boldsymbol{z}$, $\boldsymbol{e}^{[0]}=\boldsymbol{1}_{\lvert \mathcal{E}\vert}$,   $\gamma > 1$ and $\xi>0$.
             \State \textbf{While} $\vert \Psi(\boldsymbol{u}^{[k+1]},\boldsymbol{e}^{[k+1]})-\Psi(\boldsymbol{u}^{[k]},\boldsymbol{e}^{[k]}) \vert> \xi$ \\
             
$\left \lfloor \begin{array}{l}
            \mbox{Set $c_k = \gamma L_{\beta\nabla_u g(\cdot,\boldsymbol{e}^{[k]})}$ and $d_k > 0$ }\\
           \widetilde{\boldsymbol{u}}^{[k]} = \boldsymbol{u}^{[k]}-\frac{\beta}{c_k}\nabla_{\boldsymbol{u}}{g}(\boldsymbol{u}^{[k]},\boldsymbol{e}^{[k]})\\
          \boldsymbol{u}^{[k+1]} = \mathrm{prox}_{\frac{1}{c_k} f} \left(\widetilde{\boldsymbol{u}}^{[k]}\right)\\           
           \boldsymbol{e}^{[k+1]} = \mathrm{prox}_{\frac{1}{d_k} \left( \lambda h +  \beta{g}(\boldsymbol{u}^{[k+1]}, \boldsymbol{\cdot})\right)} \left(\boldsymbol{e}^{[k]}\right)
          \end{array}\right.$
          
             \end{algorithmic}
\end{algorithm}

The piecewise smooth image denoising and contour detection strategy defined by \eqref{eq:dms_univariate} and on which this paper focuses corresponds to a particularization of Problem~1. 
The three terms of the objective function $\Psi$ of Eq.~\eqref{eq:dms} are particularized to
    \begin{equation}
    \label{eq:DMSfunctions}
        \begin{cases}
            f(\boldsymbol{u}) = \displaystyle   \frac{1}{2} \| \boldsymbol{u} - \boldsymbol{z} \|_2^2,\\
           g(\boldsymbol{u},\boldsymbol{e}) = \displaystyle  \sum_{i=1}^{\vert\mathcal{E}\vert} (1 - e_i)^2 (\boldsymbol{\mathrm{D}}_i \boldsymbol{u})^2, \\
           h(\boldsymbol{e}) = \displaystyle   \sum_{i=1}^{\vert \mathcal{E}\vert} h_i( e_{i}) \\
        \end{cases}
    \end{equation}
where, for all $i \in \{ 1, \ldots, \vert \mathcal{E} \vert \}$,  $\boldsymbol{\mathrm{D}}_i$ denotes the $i^{th}$-row of the discrete gradient operator $\boldsymbol{\mathrm{D}}$, and $h_i: \mathbb{R} \mapsto (-\infty,+\infty]$ is a separable proper, lower semi-continuous, and  convex function having a proximal operator with known closed-form expression. The iterations of Algorithm~\ref{alg:SLPAM} specified to the minimization of D-MS lead to 
Algorithm~\ref{alg:dms0}~\cite{foare2019semi}.

For a detailed discussion of the convergence behavior depending on the choice of the descent steps $\gamma$ and $d_k$, the reader is referred to~\cite{foare2019semi}.
The most efficient setting appears to choose both of them the smallest possible.

\begin{algorithm}[htbp]
\caption{DMS-SLPAM to solve \eqref{eq:dms_univariate} \label{alg:dms0} }
       \begin{algorithmic}
           \State \textbf{Input:} \;Data $\boldsymbol{z}$. Set $\beta>0$, $\lambda>0$.
           \State \textbf{Initialization}: $\boldsymbol{u}^{[0]}=\boldsymbol{z}$, $\boldsymbol{e}^{[0]}=\boldsymbol{1}_{\lvert \mathcal{E}\vert}\in \mathbb{R}^{\vert \mathcal{E}\vert }$.
           \State Set $\gamma > 1$ and $\xi > 0$.
          \State \textbf{While} $\vert \Psi( \boldsymbol{u}^{[k+1]}, \boldsymbol{e}^{[k+1]})-\Psi(\boldsymbol{u}^{[k]}, \boldsymbol{e}^{[k]}) \vert > \xi $ \\
          
          $\left \lfloor \begin{array}{l}
          \mbox{Set $ c_k = \gamma \beta \Vert  \boldsymbol{\mathrm{D}} \Vert^2$ and $d_k >0 $. }\\
          \widetilde{\boldsymbol{u}}^{[k]}  = \boldsymbol{u}^{[k]}-\frac{\beta}{c_k}\nabla_{\boldsymbol{u}}g(\boldsymbol{u}^{[k]},\boldsymbol{e}^{[k]}) \\
            \boldsymbol{u}^{[k+1]} = \textrm{prox}_{\frac{1}{c_k} f} (\widetilde{\boldsymbol{u}}^{[k]}) \\
           \mbox{For all $  i \in \{1,\ldots ,\vert \mathcal{E} \vert\} $} \\
          \left \lfloor \begin{array}{l} \;
           \widetilde{e}_i^{[k]} =\frac{\beta \displaystyle (\boldsymbol{\mathrm{D}}_i \boldsymbol{u}^{[k+1]})^2 + \frac{d_k e_i^{[k]} }{2}}{\beta \displaystyle (\boldsymbol{\mathrm{D}}_i \boldsymbol{u}^{[k+1]})^2 + \frac{d_k}{2}} \\
             e_i^{[k+1]} = \mathrm{prox}_{ \frac{\lambda}{2 \beta (\boldsymbol{\mathrm{D}}_i \boldsymbol{u}^{[k+1]})^2 + d_k} {h_i}} (\widetilde{e}_i^{[k]})
              \end{array} \right. \\
            \end{array} \right.$
\end{algorithmic}
\end{algorithm}

\section{Risk estimation}
\label{ssec:SURE}

As previously discussed in introduction, many variational approaches for image restoration and contour detection consists in designing a parametric estimator $\widehat{\boldsymbol{u}}(\boldsymbol{z}; \Theta)$, e.g., defined as a minimizer of \eqref{eq:BZ} or \eqref{eq:dms_univariate},  which aims at providing the best possible estimate of a quantity of interest $\overline{\boldsymbol{u}}$ from noisy observations $\boldsymbol{z}$.
By construction, the quality of this estimate crucially relies  on the precise selection of the hyperparameters $\Theta$, which can be for instance the \textit{regularization parameters} $\beta$ and $\lambda$ in D-MS functional~\eqref{eq:dms_univariate}.\\

\noindent \textbf{Quadratic risk based parameter selection}\\
\label{sssec:param_selec}
\noindent The hyperparameters tuning task is commonly formulated as the minimization of the following \textit{quadratic risk}:
\begin{align}
\label{eq:def_quad_risk}
Q[\widehat{\boldsymbol{u}}](\Theta) =  \mathbb{E}[\Vert \widehat{\boldsymbol{u}}(\boldsymbol{z};\Theta) - \overline{\boldsymbol{u}} \Vert_2^2],
\end{align}
measuring the expected \textit{reconstruction} error made when estimating ground truth $\overline{\boldsymbol{u}}$ by $\widehat{\boldsymbol{u}}(\boldsymbol{z}; \Theta)$.
The expectation in Eq.~\eqref{eq:def_quad_risk} runs over the realizations of the noise corrupting $\boldsymbol{z}$. 

In practice, $\overline{\boldsymbol{u}}$ being unknown and the number of observed samples $\boldsymbol{z}$ being limited, if not reduced to one, the exact quadratic risk $Q[\widehat{\boldsymbol{u}}](\Theta)$ of Eq.~\eqref{eq:def_quad_risk} is not accessible.
Thus, the minimization of the quadratic risk $Q[\widehat{\boldsymbol{u}}](\Theta)$ is replaced by the minimization of some \textit{estimate} $\widehat{Q}(\boldsymbol{z};\Theta\lvert {\sigma^2})$ computed from a single noisy sample $\boldsymbol{z}$, not requiring the knowledge of ground truth but only some prior knowledge about the noise, e.g., its standard deviation $\sigma$ :
\begin{equation}
\label{eq:riskmin}
\widehat{\Theta} \in \underset{\Theta}{\mathrm{Argmin}} \;\widehat{Q}(\boldsymbol{z};\Theta\lvert {\sigma^2}).
\end{equation}

\noindent Then, the design of a fast \textit{gradient}-based hyperparameter selection strategy providing optimal hyperparameters from the minimization of~\eqref{eq:riskmin}  requires an unbiased estimate $\partial_{\Theta} \widehat{Q}(\boldsymbol{z}; \Theta \lvert \sigma^2)$ of the \textit{gradient} of the quadratic risk with respect to hyperparameters $\Theta$.
Such a general procedure is sketched in Algorithm~\ref{alg:mininize_sure}.

{\color{black}
      \begin{algorithm}[htbp]
            \caption{Automated selection of hyperparameters.\label{alg:mininize_sure}}
            \begin{algorithmic}
            \State \textbf{Input:} Data $\boldsymbol{z}$ and true or estimated $\sigma^2$.
            \State \textbf{Initialization:} Set $\Theta^{[0]} \in \mathbb{R}^L $. 
            \State \textbf{For } $t= 0$ \textbf{to} $T_{\max}-1$ \textbf{do}
            
            $\left \lfloor \begin{array}{l} 
             \mbox{Compute $ \widehat{Q}(\boldsymbol{z}; \Theta^{[t]}\ \lvert \sigma^2)$} \\
             \mbox{Compute $\partial_{\Theta} \widehat{Q}(\boldsymbol{z}; \Theta^{[t]}\ \lvert \sigma^2)$} \\
             \mbox{Update $\Theta^{[t]} $ to $\Theta^{[t+1]}$ \textit{via} a gradient} \\
             \mbox{descent step}
              \end{array} \right.$ \\
                   \State \textbf{Output:}  $\Theta^* =  \Theta^{[T_{\rm{max}}]}$
        	\end{algorithmic}
        \end{algorithm}
}

\noindent \textbf{Stein Unbiased Risk Estimate -- }
To address the fact that the ground truth $\overline{\boldsymbol{u}}$ is unknown, the pioneer work of Stein~\cite{Stein_C_1981_j-annals-statistics_estimation_mmnd} proposed an unbiased estimate of the quadratic risk, based on an i.i.d. Gaussian noise additive model in which the observations are supposed to write
\begin{align}
\label{eq:obs_model}
\boldsymbol{z} = \overline{\boldsymbol{u}} + \sigma \boldsymbol{\zeta}, \quad \boldsymbol{\zeta} \sim \mathcal{N}(\boldsymbol{0}_N, \textbf{I}_N)
\end{align}
with $N = \lvert \Omega \rvert$ is the number of pixels and $\sigma^2$ the \textit{known} variance of the noise. 
Then, under integrability and regularity assumptions,  together with the observation model~\eqref{eq:obs_model},  the so-called Stein Unbiased Risk Estimator (SURE) was derived in~\cite{Stein_C_1981_j-annals-statistics_estimation_mmnd}, and has then been intensively used in signal and image processing~\cite{chaux2008nonlinear,pesquet2009sure,ammanouil2018ada, pascal2020parameter,Pascal_B_2021_j-acha}.  
In most applications, the original Stein estimator is not usable directly and further strategies are necessary to yield a practical estimator.
The present work focuses on a strategy combining Finite Difference approximated differentiation and Monte Carlo averaging, which was first described by~\cite{ramani2008monte}. 
Making use of a Finite Difference step  $\epsilon > 0$ and a Monte Carlo vector $ \boldsymbol{\delta} \in \mathbb{R}^N$  drawn from $ \mathcal{N}(\boldsymbol{0}_N,\textbf{I}_N)$,  Finite Difference Monte Carlo (FDMC) SURE  is defined as:
    \begin{align}
    \label{eq:fdmcsure}
        \mathrm{SURE}_{\epsilon, \boldsymbol{\boldsymbol{\delta}}}(\boldsymbol{z};\Theta \lvert \sigma^2) := \| ( \widehat{\boldsymbol{u}}(\boldsymbol{z};\Theta) - \boldsymbol{z}) \|_2^2 +  \frac{2}{\epsilon} \langle \widehat{\boldsymbol{u}}(\boldsymbol{z} + \epsilon \boldsymbol{\delta};\Theta) - \widehat{\boldsymbol{u}}(\boldsymbol{z};\Theta), \sigma^2 \boldsymbol{\delta} \rangle - \sigma^2 N,
    \end{align}
Under the Lipschitzianity with respect to $\boldsymbol{z}$ of  $\widehat{\boldsymbol{u}}(\boldsymbol{z};\Theta)$ and the  natural unambiguity property $\widehat{\boldsymbol{u}}( \boldsymbol{0}_N,\Theta) = \boldsymbol{0}_N$, the \textit{true} inaccessible quadratic risk estimator~\eqref{eq:def_quad_risk} satisfies the following asymptotic unbiasedness property:
    \begin{align}
    \label{eq:sure_thm}
        \underset{\epsilon \longrightarrow 0}{\lim} \mathbb{E} [\mathrm{SURE}_{\epsilon, \boldsymbol{\delta}}(\boldsymbol{z};\Theta \lvert \sigma^2)] = Q[\widehat{\boldsymbol{u}}](\Theta),
    \end{align}
where 
 the expectation is to be understood on both the realizations of the observation noise $\boldsymbol{\zeta}$ appearing in Eq.~\eqref{eq:obs_model}, and the realizations of the Monte Carlo vector $\boldsymbol{\delta}$. 
 Eq.~\eqref{eq:sure_thm} ensures that, for small enough Finite Difference step $\epsilon$, and 
 provided that $N$ is large enough so that the Monte Carlo strategy is relevant, a minimizer of $\mathrm{SURE}_{\epsilon, \boldsymbol{\delta}}(\boldsymbol{z};\Theta \lvert \sigma^2)$ is an approximately optimal set of hyperparameters in terms of quadratic risk.\\

\noindent \textbf{Risk estimate minimization -- }    
The gradient-based strategy sketched at Algorithm~\ref{alg:mininize_sure} when \begin{equation}
\label{eq:sure}
\widehat{Q}(\boldsymbol{z};\Theta\lvert {\sigma^2}) = \mathrm{SURE}_{\epsilon, \boldsymbol{\delta}}(\boldsymbol{z};\Theta \lvert \sigma^2)
\end{equation} relies on the FDMC Stein Unbiased GrAdient Risk (SUGAR) estimate  defined as:
    \begin{align}
    \label{eq:fdmcsugar}
           \!\!\!\! \!\!\mathrm{SUGAR}_{\epsilon, \boldsymbol{\delta}}(\boldsymbol{z};\Theta \lvert \sigma^2 ) = 2\partial_{\Theta}\widehat{\boldsymbol{u}}(\boldsymbol{z};\Theta)^*(\widehat{\boldsymbol{u}}(\boldsymbol{z};\Theta)-\boldsymbol{z}) +  \frac{2}{\epsilon} \left(\partial_{\Theta}\widehat{\boldsymbol{u}}(\boldsymbol{z} + \epsilon \boldsymbol{\delta};\Theta) -\partial_{\Theta}\widehat{\boldsymbol{u}}(\boldsymbol{z};\Theta) \right)^* \sigma^2 \boldsymbol{\delta},
        \end{align}
where $\partial_{\Theta}\widehat{\boldsymbol{u}}(\boldsymbol{z};\Theta)$ denotes the Jacobian of the parametric estimator $\widehat{\boldsymbol{u}}(\boldsymbol{z};\Theta)$ with respect to the hyperparameters $\Theta$. The first proposal of such Stein Unbiased GrAdient Risk (SUGAR) estimate  was formulated by~\cite{deledalle2014stein} for i.i.d. Gaussian noise, and then extended in~\cite{pascal2020automated} for correlated noise. The main difficulty when it comes to practical implementation  is to evaluate the Jacobian matrices. 
In \cite{deledalle2014stein,pascal2020automated}, the authors proposed an efficient implementation when $\widehat{\boldsymbol{u}}(\boldsymbol{z};\Theta)$ is estimated from the resolution of a convex minimization problem of the form~\eqref{eq:BZ} while in this contribution we extend it in the context of interface detection involving a minimization problem such as \eqref{eq:dms_univariate} solved with SL-PAM described in Section~\ref{ssec:SL-PAM}.

Under technical assumptions such as Lipschitzianity of $\widehat{\boldsymbol{u}}(\boldsymbol{z};\Theta)$ with respect to $\Theta$ and $\boldsymbol{z}$, 
it has been proved in \cite{deledalle2014stein} that the quadratic risk estimator~\eqref{eq:fdmcsure} is weakly differentiable with respect to $\Theta$ and its gradient is exactly the gradient estimator recalled  in~\eqref{eq:fdmcsugar}, i.e.,
\begin{align}
\label{eq:prop1}
\partial_{\Theta} \mathrm{SURE}_{\epsilon, \boldsymbol{\delta}}(\boldsymbol{z};\Theta \lvert \sigma^2 ) = \mathrm{SUGAR}_{\epsilon, \boldsymbol{\delta}}(\boldsymbol{z};\Theta \lvert \sigma^2 ).
\end{align}
Eq.~\eqref{eq:prop1} ensures that the gradient estimate $\mathrm{SUGAR}_{\epsilon, \boldsymbol{\delta}}(\boldsymbol{z};\Theta \lvert \sigma^2 )$ is indeed the gradient of the quadratic risk estimate $\mathrm{SURE}_{\epsilon, \boldsymbol{\delta}}(\boldsymbol{z};\Theta \lvert \sigma^2 )$ with respect to hyperparameters $\Theta$, justifying the use of the gradient descent approach of Algorithm~\ref{alg:mininize_sure} to solve a particular instance of Problem~\eqref{eq:riskmin} when $\widehat{Q}(\boldsymbol{z};\Theta\lvert {\sigma^2})$ is defined by \eqref{eq:sure}. 
Additionally, FDMC SUGAR estimator introduced in~\eqref{eq:fdmcsugar} is an asymptotically unbiased estimator of the gradient of the \textit{true} quadratic risk, i.e.
\begin{equation}
\label{eq:sugar_thm}
\underset{\epsilon \longrightarrow 0}{\textrm{lim}} \mathbb{E}[\mathrm{SUGAR}_{\epsilon, \boldsymbol{\delta}}(\boldsymbol{z};\Theta \lvert \sigma^2 )] = \partial_\Theta Q[\widehat{\boldsymbol{u}}](\Theta),
\end{equation}
where $\partial_{\Theta} Q[\widehat{\boldsymbol{u}}](\Theta)$ is the \textit{true} inaccessible gradient of quadratic risk with respect to hyperparameters $\Theta$, and the expectation is to be understood on both the realizations of the observation noise $\boldsymbol{\zeta}$ appearing in Eq.~\eqref{eq:obs_model} and the realizations of the Monte Carlo vector $\boldsymbol{\delta}$.
The asymptotic unbiasedness of the gradient estimate ensures that the risk profile around its minimum is well enough reproduced by Stein-like estimates so that Algorithm~\ref{alg:mininize_sure} can be reasonably supposed to output a good approximation of the \textit{true} optimal hyperparameters.

\section{Iterative differentiation of SL-PAM for D-MS }\label{sec:diffSLPAM}

\label{sec:iterdiff}

        \subsection{Update of $\partial_\theta \tilde{\boldsymbol{u}}^{[k]}$}
        
               For the function $g$ given in Eq. \eqref{eq:DMSfunctions}, the update rule of $\widetilde{\boldsymbol{u}}^{[k]}$ reads:
        \begin{equation}
            \begin{aligned}
            \widetilde{\boldsymbol{u}}^{[k]} & = \boldsymbol{u}^{[k]}-\frac{\beta}{c_k}\nabla_{u} g(\boldsymbol{u}^{[k]},\boldsymbol{e}^{[k]}) \\
            & = \boldsymbol{u}^{[k]}-\frac{2\beta}{c_k} \displaystyle \sum_{i=1}^{\vert \mathcal{E} \vert}  (1-e^{[k]}_i)^2 \ \boldsymbol{\mathrm{D}}_i^*\boldsymbol{\mathrm{D}}_i \boldsymbol{u}^{[k]}.
            \end{aligned}
        \end{equation}
        
        The update of $\tilde{\boldsymbol{u}}^{[k]}$ can be written:
        \begin{equation}
            \tilde{\boldsymbol{u}}^{[k]} = \Gamma ( \boldsymbol{u}^{[k]},\boldsymbol{e}^{[k]},\tau^{[k]} ),
        \end{equation}
        where
         \begin{equation}
        \begin{cases}
         \Gamma(\mathbf{u},\mathbf{e},\tau) = \mathbf{u}-\tau \displaystyle \sum_{i=1}^{\vert \mathcal{E} \vert}  (1-\mathrm{e}_i)^2 \ \boldsymbol{\mathrm{D}}_i^*\boldsymbol{\mathrm{D}}_i \mathbf{u}, \\
           \tau^{[k]} = \frac{2\beta}{c_k}.
         \end{cases}
          \end{equation}

\noindent The derivative $\partial_{\theta} \mathbf{v}$ of $\mathbf{v} = \Gamma(\mathbf{u},\mathbf{e},\tau)$, for $\theta \in \{\beta,\lambda\}$, is:
        \begin{align}
            &\begin{aligned}
            \partial_{\theta}\mathbf{v}= 
            & \partial_{\theta}\mathbf{u} - \tau  \displaystyle \sum_{i=1}^{\vert \mathcal{E} \vert}  (1-\mathrm{e}_i)^2 \ \boldsymbol{\mathrm{D}}_i^*\boldsymbol{\mathrm{D}}_i \partial_{\theta} \mathbf{u} \\
            & + 2\tau \displaystyle \sum_{i=1}^{\vert \mathcal{E} \vert}  (1-\mathrm{e}_i) \partial_{\theta} \mathrm{e}_i \ \boldsymbol{\mathrm{D}}_i^*\boldsymbol{\mathrm{D}}_i \mathbf{u} \\
           & - \partial_\theta \tau \displaystyle \sum_{i=1}^{\vert \mathcal{E} \vert} (1-\mathrm{e}_i)^2 \mathbf{D}_i^*\mathbf{D}_i \mathbf{u},
            \end{aligned}
        \end{align}
        and, since $c_k =   \beta \gamma\Vert D \Vert^2$, for $\theta \in \{\beta,\lambda\}$,
        \begin{equation}
           \partial_\theta \tau^{[k]} =  \partial_\theta \left(\frac{2\beta}{c_k} \right) = \partial_\theta \left(\frac{2}{\gamma\Vert D \Vert^2 } \right) =0.
        \end{equation}
        
        \subsection{Update of $\partial_\theta \boldsymbol{u}^{[k+1]}$ }
        
                The function $f$ given in Eq. \eqref{eq:DMSfunctions} has a proximal operator with a closed form expression. Thus the update rule of $\boldsymbol{u}^{[k]}$ can be explicitely expressed as follows:
        \begin{equation}
            \boldsymbol{u}^{[k+1]} = \textrm{prox}_{\frac{1}{c_k} f} (\widetilde{\boldsymbol{u}}^{[k]}) = \frac{c_k \widetilde{\boldsymbol{u}}^{[k]}+\boldsymbol{z}}{c_k+1}.
        \end{equation}

\noindent For the update of $\boldsymbol{u}^{[k+1]}$, we thus have:   
        \begin{equation}
            \boldsymbol{u}^{[k+1]} = \Gamma(\tilde{\boldsymbol{u}}^{[k]},\boldsymbol{0}_{\vert \mathcal{E} \vert},\tau^{[k]}),
        \end{equation}
        where
        \begin{equation}
        \begin{cases}
         \Gamma(\mathbf{u},\mathbf{e},\tau) = 
          \frac{\tau \mathbf{u}+\boldsymbol{z}}{\tau+1}, \\
         \tau^{[k]} = c_k.
         \end{cases}
         \end{equation}
        
\noindent The derivative $\partial_{\theta} \mathbf{v}$ of $\mathbf{v} = \Gamma(\mathbf{v},\mathbf{e},\tau)$, for $\theta \in \{\beta,\lambda\}$, is:
            \begin{equation}
            \partial_{\theta} \mathbf{v} 
             = \frac{\tau}{\tau+1}\partial_{\theta} \mathbf{u} + \frac{\mathbf{u}-\boldsymbol{z}}{(\tau+1)^2} \partial_{\theta} \tau,
            \end{equation}
       and $\partial_{\beta} \tau^{[k]} = \gamma \| \boldsymbol{\mathrm{D}} \|^2$ and $\partial_{\lambda} \tau^{[k]} = 0$.

        \subsection{Update of $\partial_\theta \tilde{\boldsymbol{e}}^{[k]}$}

	 In Algorithm~\ref{alg:dms0}, the parameter for the update of $\boldsymbol{e}^{[k+1]}$ is set to $d_k = \beta \overline{d}  $ where $\overline{d} = \eta \Vert \boldsymbol{\mathrm{D}} \Vert_2^2$. This choice is discussed in \cite{foare2019semi} and gives good numerical results for some values of $\eta$. This setting simplifies the computation of the derivatives due to the linear dependance of $d_k$ with $\beta$. Indeed, the update rule of  $\widetilde{e}_{i}^{[k]}$, for every $i\in\{1,\ldots, \vert \mathcal{E} \vert\}$, can be rewritten as follows:
        \begin{align}
            \widetilde{e}_{i}^{[k]} = \frac{ (\boldsymbol{\mathrm{D}}_i \boldsymbol{u}^{[k+1]} )^2 + \frac{\overline{d} e_{i}^{[k]} }{2}}{(\boldsymbol{\mathrm{D}}_i \boldsymbol{u}^{[k+1]} )^2 + \frac{\overline{d}}{2}}.
        \end{align}
        
        Thus, for every $i\in\{1,\ldots, \vert \mathcal{E} \vert\}$,
        \begin{equation}
            \tilde{e}_{i}^{[k]}  = \Gamma_i(\boldsymbol{u}^{[k+1]},\boldsymbol{e}^{[k]},\tau^{[k]}),
        \end{equation}
        where 
        \vspace{-1cm}
        
        \begin{equation}
        \begin{cases}
        \Gamma_i(\mathbf{u},\mathbf{e},\tau) = \frac{(\boldsymbol{\mathrm{D}}_i \mathbf{u})^2 + \frac{\tau}{2} \mathrm{e}_i }{(\boldsymbol{\mathrm{D}}_i \mathbf{u})^2 + \frac{\tau}{2}}, \\
        \tau^{[k]}= \overline{d}.
        \end{cases}
        \end{equation}

        The derivative $ \partial_{\theta} \mathbf{v} $ of $\mathbf{v} = \Gamma(\mathbf{u},\mathbf{e},\tau)$, for $\theta \in \{\beta,\lambda\}$ and for $i\in\{1,\ldots, \vert \mathcal{E} \vert\}$, is:
       \begin{equation}
            \begin{aligned}
            \partial_{\theta} \mathrm{v}_i 
             = &  \frac{2 \boldsymbol{\mathrm{D}}_i \mathbf{u} \boldsymbol{\mathrm{D}}_i \partial_{\theta} \mathbf{u} }{ \left(\boldsymbol{\mathrm{D}}_i \mathbf{u} \right)^2 + \frac{\tau}{2} }  -  \frac{ 2  \boldsymbol{\mathrm{D}}_i \mathbf{u} \boldsymbol{\mathrm{D}}_i \partial_{\theta} \mathbf{u} \left[ \left(\boldsymbol{\mathrm{D}}_i \mathbf{u} \right)^2 + \frac{\tau}{2} \mathrm{e}_i  \right]}{ \left[  \left(\boldsymbol{\mathrm{D}}_i \mathbf{u} \right)^2 + \frac{\tau}{2} \right]^2} \\
              & +\frac{ \frac{\tau}{2} \partial_{\theta} \mathrm{e}_i}{  \left(\boldsymbol{\mathrm{D}}_i \mathbf{u} \right)^2 + \frac{\tau}{2}}  \\
              & +  \frac{ \frac{\partial_{\theta} \tau}{2} \mathrm{e}_i}{ \left(\boldsymbol{\mathrm{D}}_i \mathbf{u} \right)^2 + \frac{\tau}{2} }  -  \frac{ [(\boldsymbol{\mathrm{D}}_i \mathbf{u})^2 + \frac{\tau}{2} \mathrm{e}_i]\frac{\partial_{\theta} \tau}{2} }{ \left[  \left(\boldsymbol{\mathrm{D}}_i \mathbf{u} \right)^2 + \frac{\tau}{2} \right]^2}, \\
            \end{aligned}
        \end{equation}
        and $\partial_\theta \tau^{[k]} = 0$, which yields to the result in \cite[Proposition~2]{Lucas_CG_2022_sivp}.
        \subsection{Update of $\partial_\theta \boldsymbol{e}^{[k+1]}$ }
        
        For the update of $e_i^{[k+1]}$, for every $i\in\{1,\ldots, \vert \mathcal{E} \vert\}$, the fonction $h$ in Eq. \eqref{eq:DMSfunctions} is chosen with $h_i = \vert . \vert$. This latter function corresponds to the common $\ell_1$-norm penalization of the contour. The setting $d_k =  \beta \overline{d} $ simplifies this update which now reads:
        \begin{equation}
            e_i^{[k+1]} = \textrm{prox}_{\phi_i^{[k+1]} \vert . \vert}( \widetilde{e}_i^{[k]} ), \quad \phi_i^{[k+1]} = \phi_i(\boldsymbol{u}^{[k+1]}; \tau),
        \end{equation}
where
        \begin{equation}
        \begin{cases}
            \phi_i(\mathbf{u}, \tau) = \frac{\tau}{\left[ 2  \left(\boldsymbol{\mathrm{D}}_i \mathbf{u} \right)^2 + \overline{d} \right]}, \quad \tau = \frac{\lambda}{\beta},\\
            \textrm{prox}_{\phi_i(\mathbf{u},\tau) \vert . \vert}(\mathrm{e}_i) = \textrm{max}(0,1-\frac{\phi_i(\mathbf{u}; \tau)}{\vert \mathrm{e}_i \vert}) \mathrm{e}_i.
        \end{cases}
        \end{equation}       
        
 \noindent       For every $i\in\{1,\ldots, \vert \mathcal{E} \vert\}$,
         \begin{equation}
            e_i^{[k+1]} = \Gamma_i(\boldsymbol{u}^{[k+1]},\boldsymbol{\tilde{e}}^{[k]},\tau^{[k]}),
        \end{equation}
          where 
          \begin{equation}
          \begin{cases}
          \Gamma_i(\mathbf{u},\mathbf{e},\tau) = \textrm{prox}_{\phi_i(\mathbf{u};\tau) \vert . \vert}( \mathrm{e}_i),\\
           \tau^{[k]} = \frac{\lambda}{\beta}.
          \end{cases} 
          \end{equation}

\noindent        The derivative $ \partial_{\theta} \mathbf{v} $ of $\mathbf{v} = \Gamma(\mathbf{u},\mathbf{e},\tau)$, for $\theta \in  \{\beta,\lambda\}$ and for $i\in\{1,\ldots, \vert \mathcal{E} \vert\}$, is:
        \begin{equation}
      \begin{aligned}
        \partial_{\theta}\mathrm{v}_i  
        = & - \partial_{\mathbf{u}} \phi_i \partial_{\theta} \mathbf{u}  \frac{\mathrm{e}_i }{\vert \mathrm{e}_i \vert} \mathcal{I}_{\vert \mathrm{e}_i \vert > \phi_i(\mathbf{u},\tau) } +\partial_{\theta} \mathrm{e}_i \mathcal{I}_{\vert \mathrm{e}_i \vert > \phi_i(\mathbf{u},\tau) } \\
        & - \frac{\partial_\theta \tau}{ \left[ 2 \left(\boldsymbol{\mathrm{D}}_i \mathbf{u} \right)^2 + \overline{d} \right]}  \frac{\mathrm{e}_i}{\vert \mathrm{e}_i \vert} \mathcal{I}_{\vert \mathrm{e}_i \vert > \phi_i(\mathbf{u},\tau)},
        \end{aligned}
        \end{equation}
                with Jacobian matrices product
        \begin{align}
        \partial_{\mathbf{u}} \phi_i  \partial_{\theta} \mathbf{u}
        &= -\frac{\tau }{ \left[ 2  \left( \boldsymbol{\mathrm{D}}_i \mathbf{u} \right)^2 + \overline{d} \right]^2}  \left(  4    \boldsymbol{\mathrm{D}}_i \mathbf{u}  \boldsymbol{\mathrm{D}}_i  \partial_{\theta} \mathbf{u}   \right), 
        \end{align}
        and $\partial_\beta \tau^{[k]} = \frac{1}{\beta}$ and $\partial_\lambda \tau^{[k]} = -\frac{\lambda}{\beta^2} $.

\section{Proof of Proposition~2}
\label{sec:proof2}

For each Monte Carlo vector $\boldsymbol{\delta}^{(r)}$, the FDMC $\mathrm{SURE}_{\epsilon, \boldsymbol{\delta}^{(r)}}$ and $\mathrm{SUGAR}_{\epsilon, \boldsymbol{\delta}^{(r)}}$ estimates, defined at Eq.~\eqref{eq:fdmcsure}~and~\eqref{eq:fdmcsugar} are asymptotically unbiased and $\mathrm{SUGAR}_{\epsilon, \boldsymbol{\delta}^{(r)}}$, is the derivative of $\mathrm{SURE}_{\epsilon, \boldsymbol{\delta}^{(r)}}$, w.r.t. $\Theta$. 
Then, by linearity of both the limit $\lim_{\epsilon \rightarrow 0}$ and the summation over the $R$ Monte Carlo vectors, 
 the \textit{Monte Carlo averaged} estimates $\overline{\mathrm{SURE}}^R_{\epsilon, \boldsymbol{\Delta}}(\boldsymbol{z};\Theta \lvert \sigma^2 )$ and $\overline{\mathrm{SUGAR}}^R_{\epsilon, \boldsymbol{\Delta}}(\boldsymbol{z};\Theta \lvert \sigma^2 )$ are also unbiased and $\overline{\mathrm{SUGAR}}^R_{\epsilon, \boldsymbol{\Delta}}(\boldsymbol{z};\Theta \lvert \sigma^2 )$ is the derivative of $\overline{\mathrm{SURE}}^R_{\epsilon, \boldsymbol{\Delta}}(\boldsymbol{z};\Theta \lvert \sigma^2 )$ w.r.t. $\Theta$.

\section{Additional experiments}
\label{sec:add_exp}
\subsection{Algorithmic setup}

\noindent \textbf{SL-PAM} -- The minimization of the D-MS functional~\eqref{eq:dms_univariate} providing estimates of both the piecewise smooth image and its salient contours, is performed running Algorithm~\ref{alg:dms0}.
The stopping criterion, based on the objective function increments, is set to $\xi = 10^{-4}$, while the descent steps are tuned manually so as to obtain the fastest convergence, leading to $\gamma = 1.01$ and $d_k = \eta \beta  \Vert \boldsymbol{\mathrm{D}} \Vert_2^2$
with $\eta = 1.01 \times 10^{-3}$ 
following~\cite{foare2019semi}.

\noindent \textbf{Stein estimators} -- FDMC SURE~\eqref{eq:fdmcsure} is computed with a Finite Difference step 
\begin{align}
\label{eq:FDeps}
\epsilon = 2 \frac{\sigma}{N^{\alpha}}, \quad 0 < \alpha < 1
\end{align}
where $\sigma$ is the standard deviation of the noise on the observed image $\boldsymbol{z} \in \mathbb{R}^N$. 
Formula~\eqref{eq:FDeps} derives from a heuristic reasoning developed in \cite{deledalle2014stein}, in the context of $\ell_1$-norm penalization.
The dependency of the Finite Difference step on the size of the data  is controlled via the exponent $\alpha$, which is fixed at $\alpha = 0.3$ for all the numerical simulations.
In the systematic numerical experiments, four values of the number $R$ of realizations of the Monte Carlo vector $\boldsymbol{\boldsymbol{\delta}}^{(r)}$ are envisioned and systematically compared: $R \in \{1, 5,10,20\}$.

\noindent \textbf{BFGS algorithm} -- To perform the risk minimization described in Algorithm~\ref{alg:mininize_sure} for different choices of $ \widehat{Q}(\boldsymbol{z}; \Theta\ \lvert \sigma^2)$ and $\partial_{\Theta} \widehat{Q}(\boldsymbol{z}; \Theta\ \lvert \sigma^2)$ 
we used the GRadient-based Algorithm for Non-Smooth Optimization, implemented in GRANSO toolbox\footnote{\texttt{http://www.timmitchell.com/software/GRANSO/}}, consisting of the low memory BFGS quasi-Newton algorithm proposed in~\cite{curtis2017bfgs}
with box constraints, enabling to enforce positivity of $\beta$ and $\lambda$. 
The maximal number of iterations of BFGS Algorithm~\ref{alg:mininize_sure} is set to $T_{\max} = 20$, while the stopping criterion on the gradient norm is set to $10^{-8}$.
Further, it is well-documented that the initialization of quasi-Newton algorithms might drastically impact their convergence.
Hence, we propose a model-based strategy for initializing Algorithm~\ref{alg:mininize_sure}. 
Inspired from the initialization strategies proposed in~\cite{deledalle2014stein,pascal2020automated},  hyperparameters $\Theta = (\beta, \lambda)$ 
are initialized as 
\begin{align}
\beta^{(0)}  = \frac{N \sigma^2}{4\|\boldsymbol{\mathrm{D}}\boldsymbol{z}\|_2^2} \quad \text{and} \quad  \lambda^{(0)} = \frac{\beta^{(0)} \|\boldsymbol{\mathrm{D}}\boldsymbol{z}\|_2^2}{2 N},
\end{align}
while, for $\kappa = 0.9$,  the initial approximated inverse Hessian involved in the BFGS strategy is set to
{\footnotesize{\begin{equation}
{\footnotesize{\mathrm{H}^{(0)} = \mathrm{diag} \left( \left\lvert \frac{ \kappa\beta^{(0)} }{\partial_{\beta}\widehat{Q}(\boldsymbol{z}; \Theta^{(0)} \lvert \sigma^2)} \right\rvert,\left\lvert \frac{\kappa\lambda^{(0)}}{\partial_{\lambda}\widehat{Q}(\boldsymbol{z}; \Theta^{(0)}\ \lvert \sigma^2)} \right\rvert  \right).}}
\end{equation}}}

\begin{figure*}[!t]
\centering
\subfloat{\includegraphics[width=0.125\linewidth]{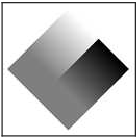}}
\hspace{0.1in}
\subfloat{\includegraphics[width=0.125\linewidth]{losange_1e-2.eps}}
\hspace{0.1in}
\subfloat{\includegraphics[width=0.125\linewidth]{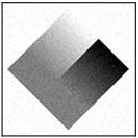}}
\hspace{0.1in}
\subfloat{\includegraphics[width=0.125\linewidth]{losange_5e-2.eps}}
\hspace{0.1in}
\subfloat{\includegraphics[width=0.125\linewidth]{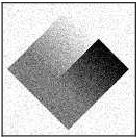}}
\hspace{0.1in}
\subfloat{\includegraphics[width=0.125\linewidth]{losange_10e-2.eps}}\\
\addtocounter{subfigure}{-6}
\subfloat[$\sigma = 0$]{\includegraphics[width=0.125\linewidth]{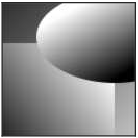}}
\hspace{0.1in}
\subfloat[$\sigma = 0.01$]{\includegraphics[width=0.125\linewidth]{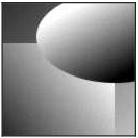}}
\hspace{0.1in}
\subfloat[$\sigma = 0.03$]{\includegraphics[width=0.125\linewidth]{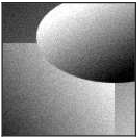}}
\hspace{0.1in}
\subfloat[$\sigma = 0.05$]{\includegraphics[width=0.125\linewidth]{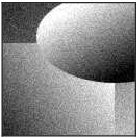}}
\hspace{0.1in}
\subfloat[$\sigma = 0.07$]{\includegraphics[width=0.125\linewidth]{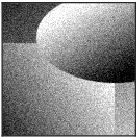}}
\hspace{0.1in}
\subfloat[$\sigma = 0.1$]{\includegraphics[width=0.125\linewidth]{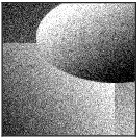}}
\caption{Piecewise smooth grey level images ($\sigma = 0$) corrupted by i.i.d. Gaussian noise with level $\sigma \in \{0.01,0.03,0.05,0.07,0.1\}$.}
\label{fig:geometries}
\end{figure*}

\begin{figure*}[t!]
\centering
\subfloat{\includegraphics[width=0.175\linewidth]{losange_denoised_1e-2.eps}}
\subfloat{\includegraphics[width=0.175\linewidth]{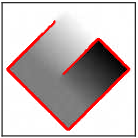}}
\subfloat{\includegraphics[width=0.175\linewidth]{losange_denoised_5e-2.eps}}
\subfloat{\includegraphics[width=0.175\linewidth]{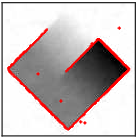}}
\subfloat{\includegraphics[width=0.175\linewidth]{losange_denoised_10e-2.eps}}\\
\subfloat{\includegraphics[width=0.175\linewidth]{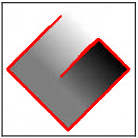}}
\subfloat{\includegraphics[width=0.175\linewidth]{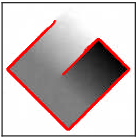}}
\subfloat{\includegraphics[width=0.175\linewidth]{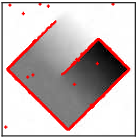}}
\subfloat{\includegraphics[width=0.175\linewidth]{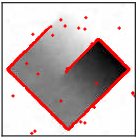}}
\subfloat{\includegraphics[width=0.175\linewidth]{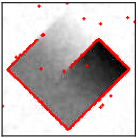}} \\
\subfloat{\includegraphics[width=0.175\linewidth]{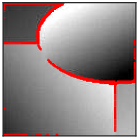}}
\subfloat{\includegraphics[width=0.175\linewidth]{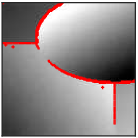}}
\subfloat{\includegraphics[width=0.175\linewidth]{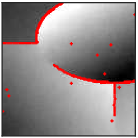}}
\subfloat{\includegraphics[width=0.175\linewidth]{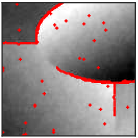}}
\subfloat{\includegraphics[width=0.175\linewidth]{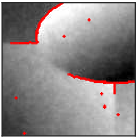}}\\
\addtocounter{subfigure}{-15}
\subfloat[$\sigma = 0.01$]{\includegraphics[width=0.175\linewidth]{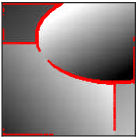}}
\subfloat[$\sigma = 0.03$]{\includegraphics[width=0.175\linewidth]{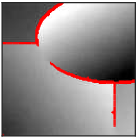}}
\subfloat[$\sigma = 0.05$]{\includegraphics[width=0.175\linewidth]{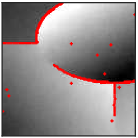}}
\subfloat[$\sigma = 0.07$]{\includegraphics[width=0.175\linewidth]{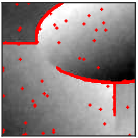}}
\subfloat[$\sigma = 0.1$]{\includegraphics[width=0.175\linewidth]{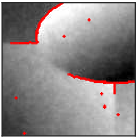}}
\caption{D-MS estimates $\widehat{\boldsymbol{u}}$ and $\widehat{\boldsymbol{e}}$ (superimposed in red) of piecewise smooth grey level images corrupted by i.i.d. Gaussian noise with noise level $\sigma$ displayed in Fig.~\ref{fig:geometries}. The D-MS hyperparameters are selected with the proposed  \textit{Averaged SUGAR D-MS} using either the true standard deviation $\sigma$ (first and third rows)  or the estimated standard deviation $\widehat{\sigma}$ (second and fourth rows). }
\label{fig:denoisedgeometries}
\end{figure*}

\subsection{Performance w.r.t noise level}
We now focus on the \textit{Averaged SUGAR D-MS}  for $R=5$ and assess its performance for the different geometries and noise levels displayed in Fig.~\ref{fig:geometries}.
Averaged PSNR for 10 realizations of the noise are reported in Table~\ref{table_SNR_noises}, 
denoised images and detected contours are displayed in Fig.~\ref{fig:denoisedgeometries} (1st and 3rd rows) for one realization of the noise.

As expected, the PSNR decreases as the noise level increases. Further, large error bars are observed mainly due to 
the variability of the gradient descent scheme in Algorithm~\ref{alg:mininize_sure}, the realizations of the Monte Carlo vector, or the SL-PAM non-convex minimization procedure.

\begin{table}[t!]
\caption{PSNR values with $95\%$ confidence interval for true noise level $\sigma$ and for estimated noise level $\hat{\sigma}$.}
\label{table_SNR_noises}
\centering
\setlength{\tabcolsep}{1mm}
\begin{tabular}{|c||c|c||c|c|}
\hline
\multirow{2}{*}{$\sigma$ } & \multicolumn{2}{c||}{Losange} &\multicolumn{2}{c|}{Ellipse}  \\
\hhline{|~|----|}  
& True $\sigma$ & Estimated $\hat{\sigma}$ &True $\sigma$ & Estimated  $\hat{\sigma}$\\
\hline
$0.01$ &  $43.85 \pm 0.06$ & $43.64 \pm 0.12$ & $41.77 \pm 0.06$ & $41.08 \pm 0.25$  \\
\hline
$0.03$ & $38.94 \pm 0.12$ & $38.89 \pm 0.13$ & $31.80 \pm 2.73$ & $33.64 \pm 2.23$ \\
\hline
$0.05$ &  $34.33 \pm 0.70$ & $34.71 \pm 0.12$ & $26.05 \pm 2.92$ & $26.37 \pm 2.80$ \\
\hline
$0.07$ & $31.60 \pm 0.52$ & $30.55 \pm 1.78$  & $26.73 \pm 2.46$ & $22.64 \pm 2.68$ \\
\hline
$0.1$ & $28.86 \pm 0.47$  & $27.24 \pm 1.84$ & $21.25 \pm 3.03$ & $22.61 \pm 3.33$  \\
\hline
\end{tabular}
\end{table}

\begin{figure*}
\hspace{-1cm}{\scriptsize{\begin{tabular}{p{2.7cm}p{2.7cm}p{2.7cm}p{1.5cm}p{1.5cm}p{1.5cm}p{1.5cm}}
Degraded & SUGAR T-ROF & SUGAR D-MS & Original &Degraded & T-ROF &D-MS\\
 & (State-of-the-art) & (Proposed) & zoom &zoom &zoom &zoom \\
  &SSIM: 0.85  & SSIM: 0.83 &&&\\
    &Jacc: 0.38 & Jacc: 0.43 &&&\\
 \includegraphics[height=2cm,width=2.7cm]{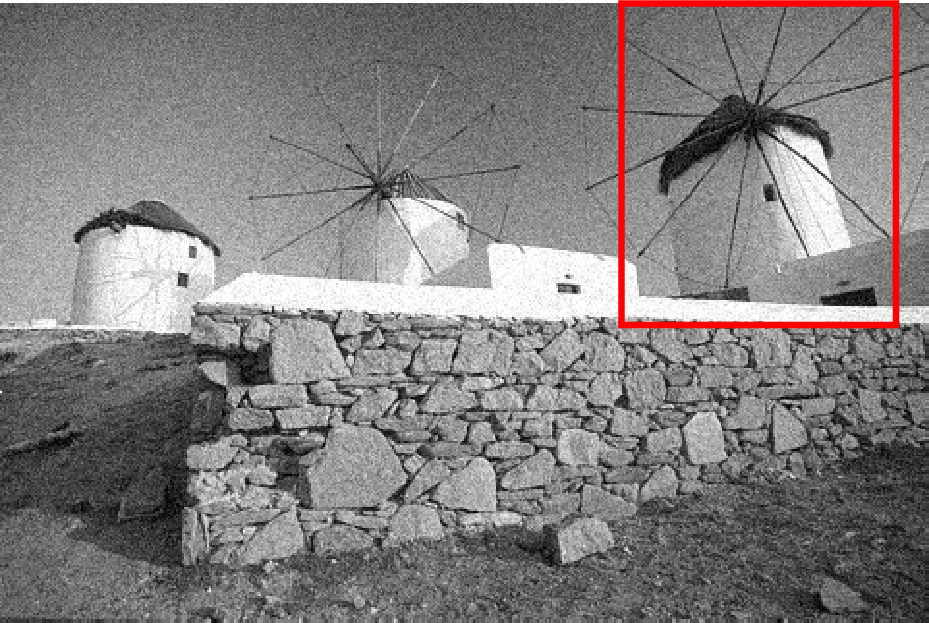}&
\includegraphics[height=2cm,width=2.7cm]{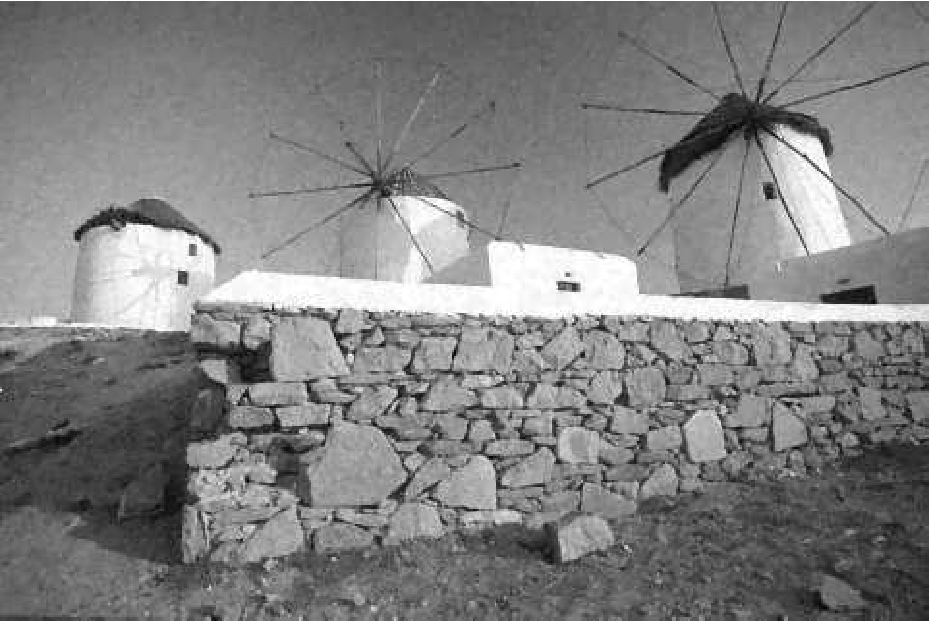}&
\includegraphics[height=2cm,width=2.7cm]{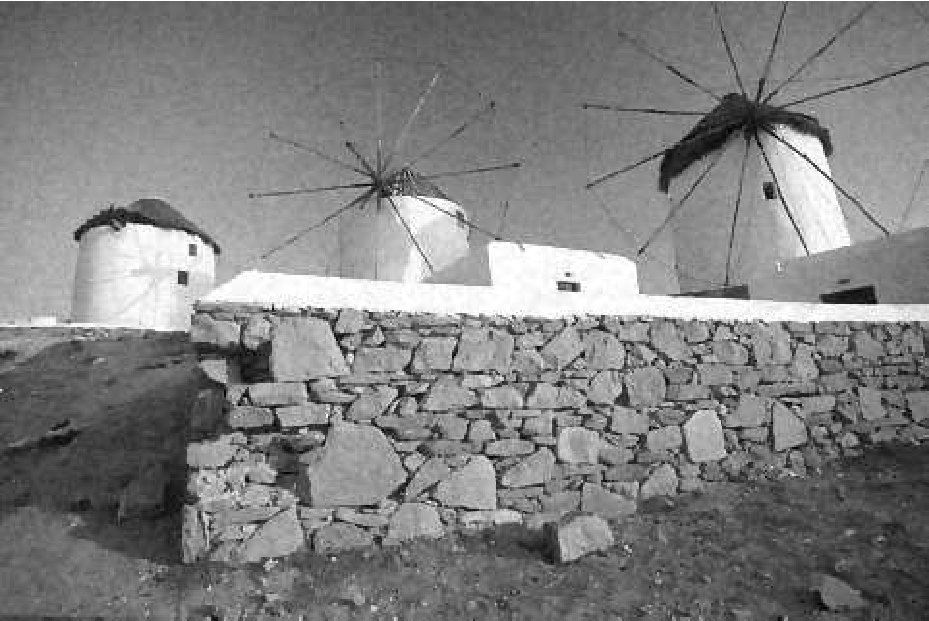}&
\includegraphics[height=2cm,width=1.5cm]{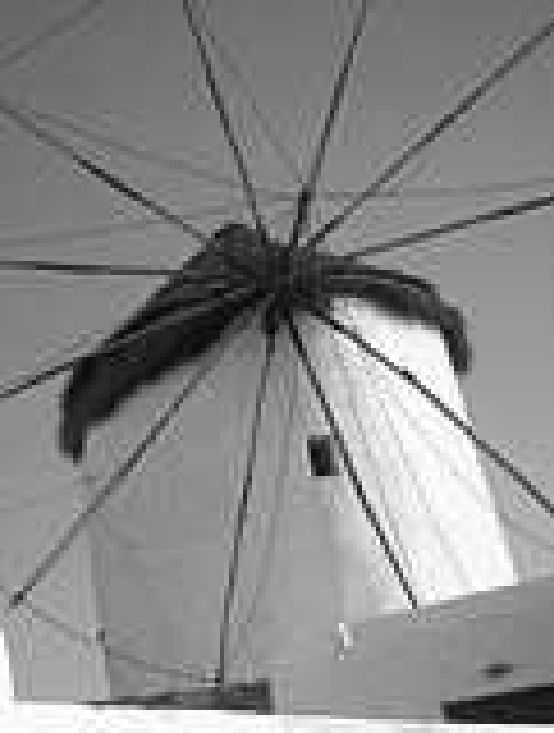}&
\includegraphics[height=2cm,width=1.5cm]{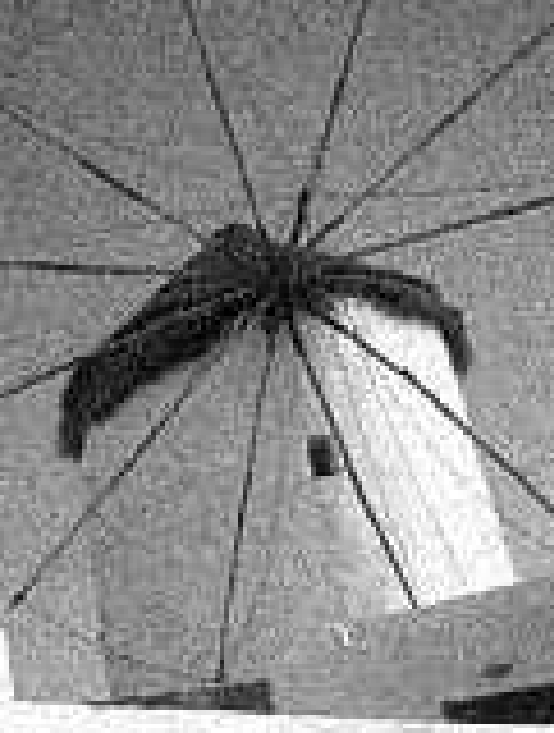}&
\includegraphics[height=2cm,width=1.5cm]{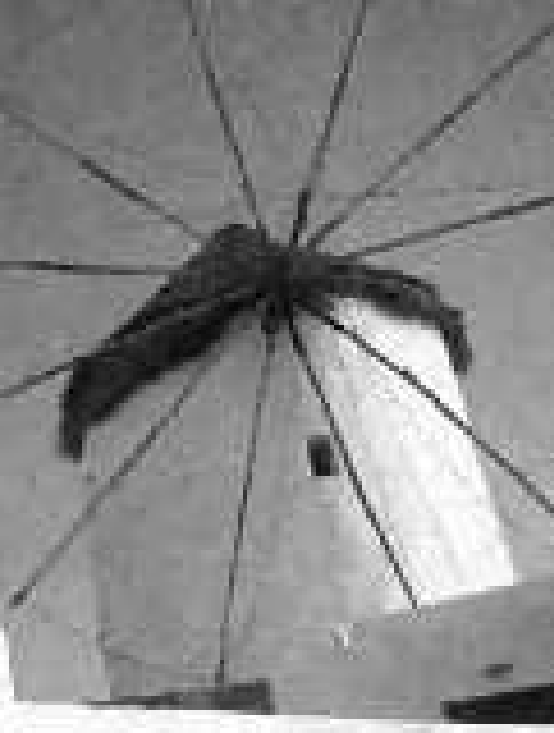}&
 \includegraphics[height=2cm,width=1.5cm]{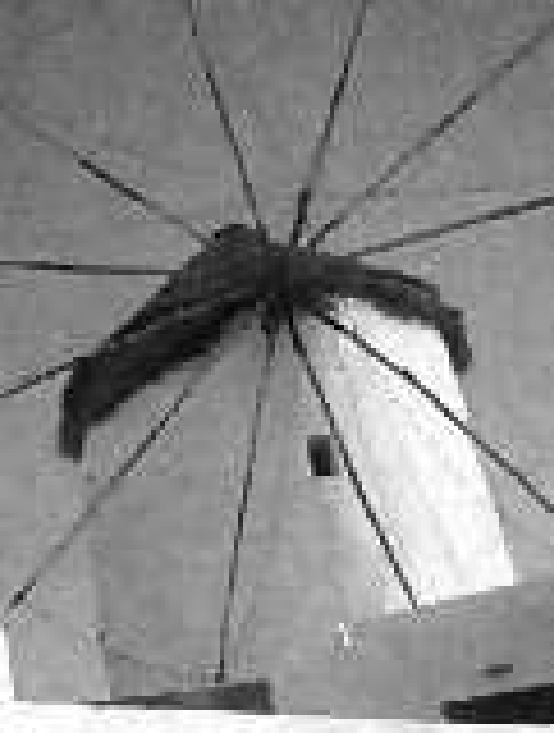}\\
&\includegraphics[height=2cm,width=2.7cm]{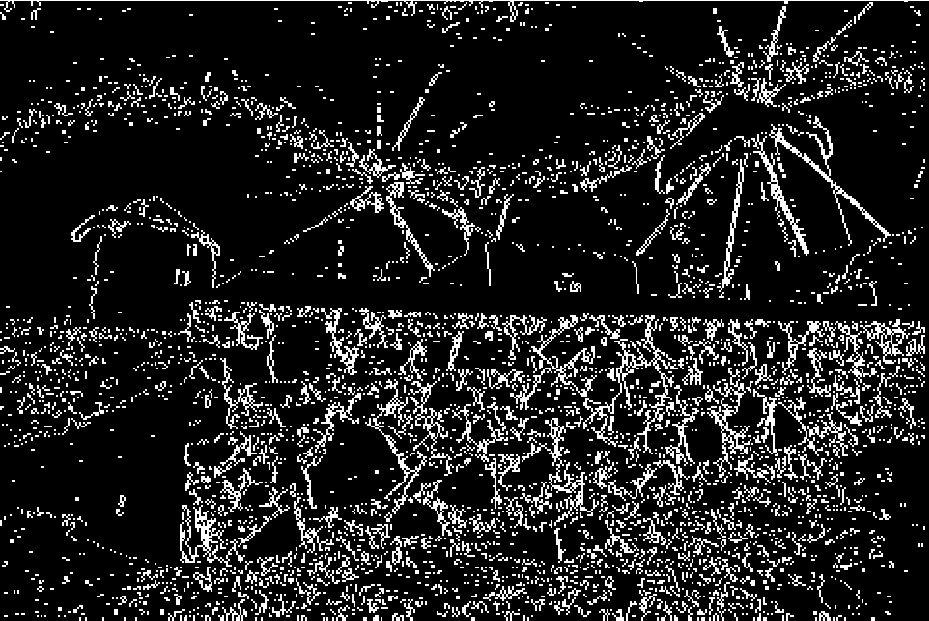}&
\includegraphics[height=2cm,width=2.7cm]{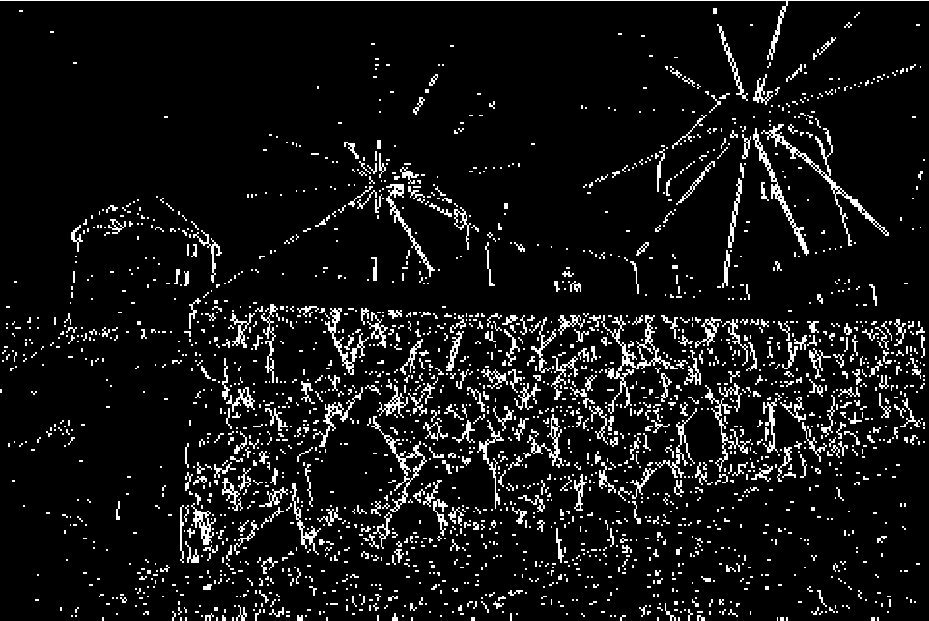} &
&
&\includegraphics[height=2cm,width=1.5cm]{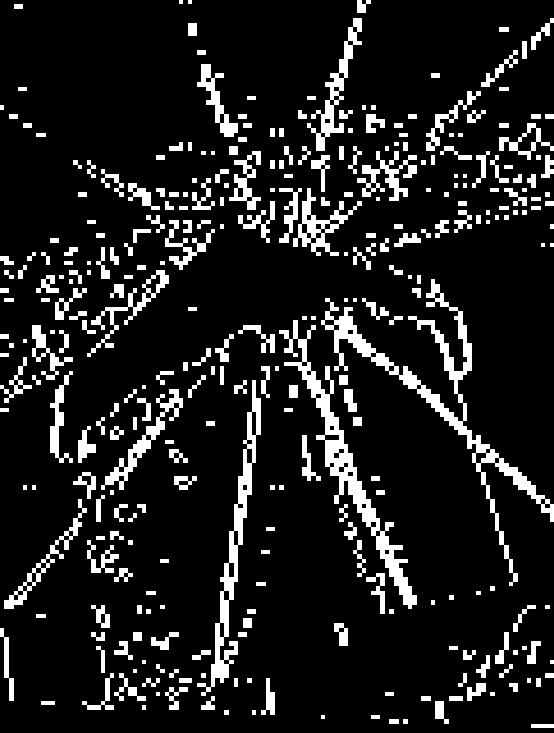}&
 \includegraphics[height=2cm,width=1.5cm]{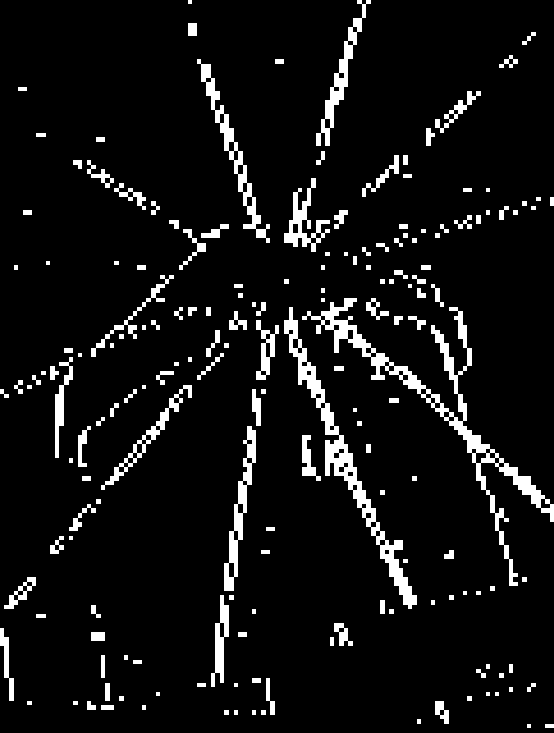}\\ 
&SSIM: 0.86  & SSIM: 0.84  &&&\\
&Jacc: 0.34 &Jacc: 0.61 &&&\\
 \includegraphics[height=2cm,width=2.7cm]{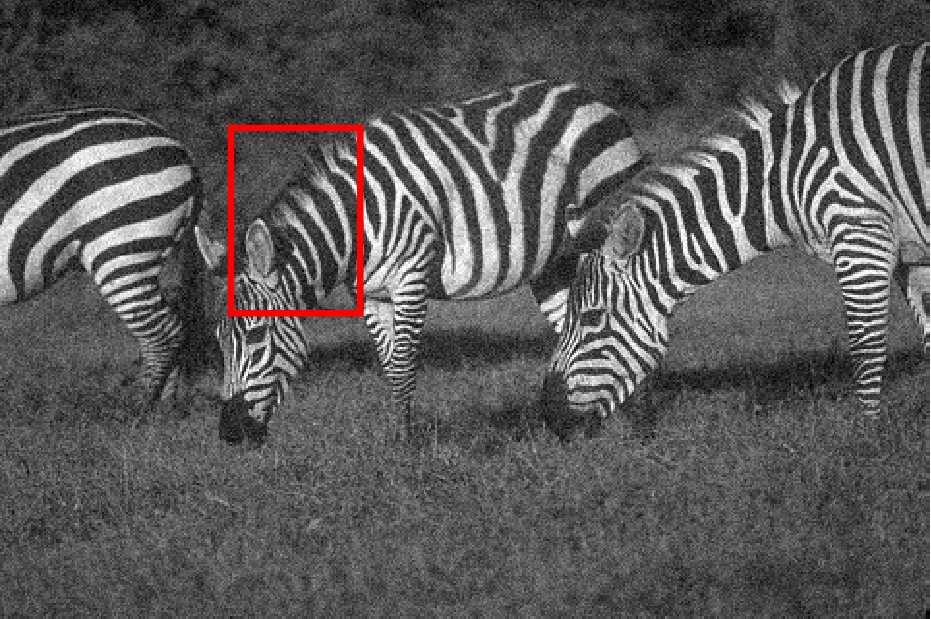}&
\includegraphics[height=2cm,width=2.7cm]{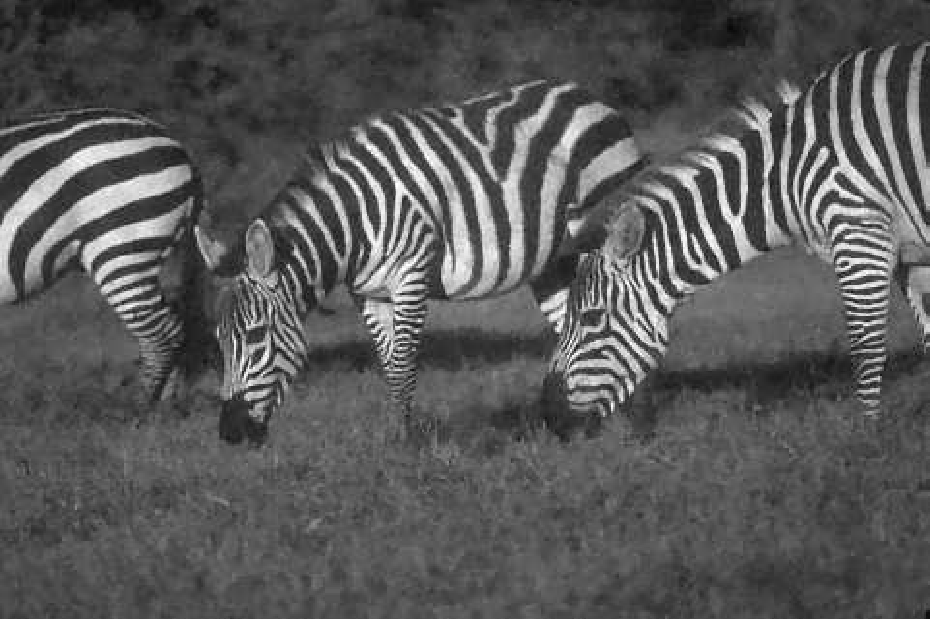}&
\includegraphics[height=2cm,width=2.7cm]{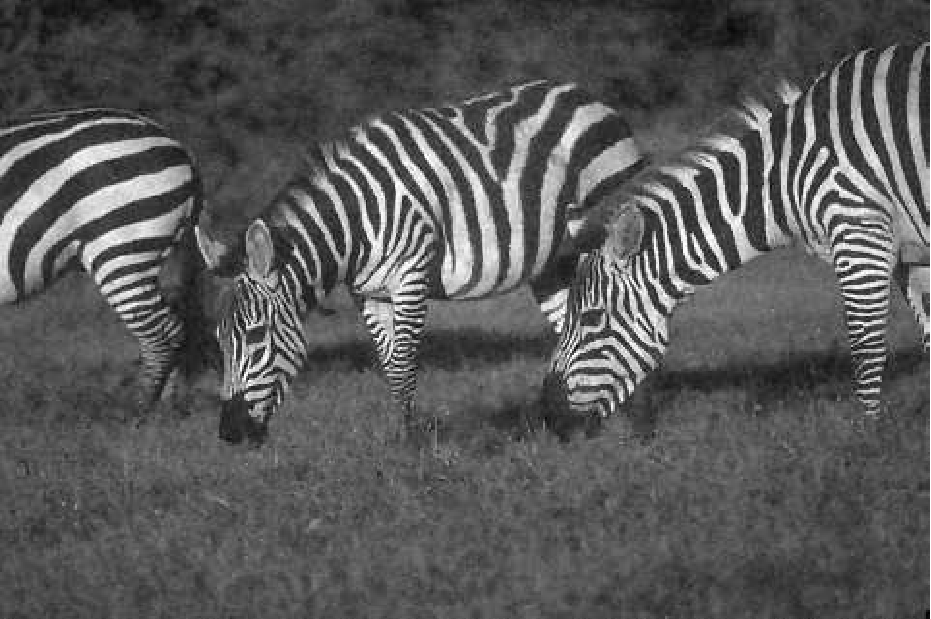}&
\includegraphics[height=2cm,width=1.5cm]{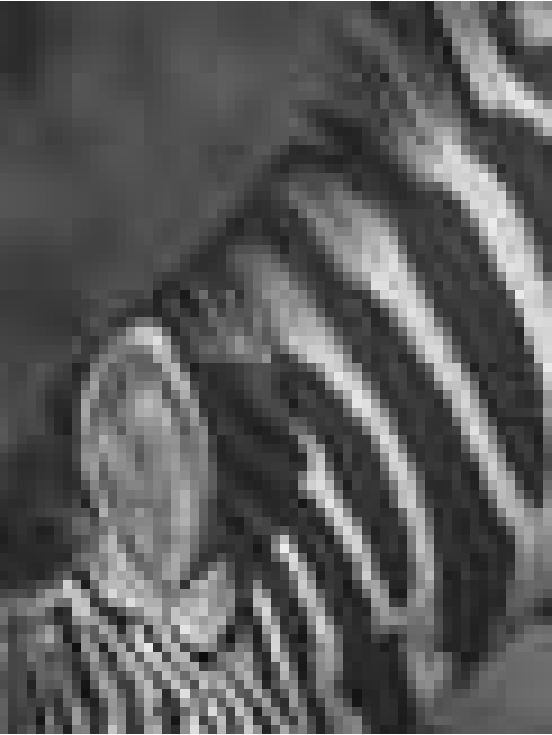}&
\includegraphics[height=2cm,width=1.5cm]{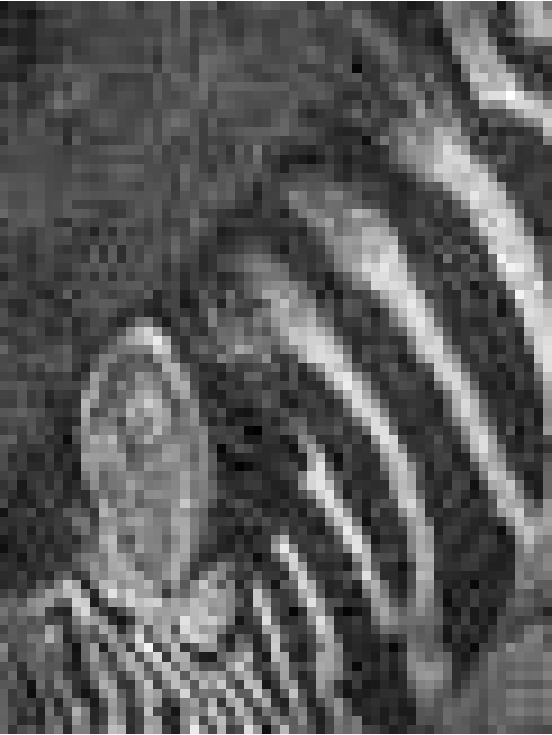}&
\includegraphics[height=2cm,width=1.5cm]{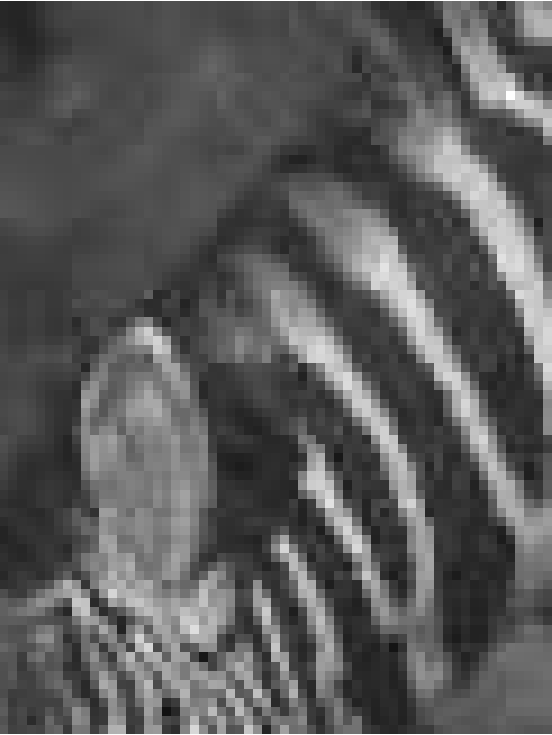}&
 \includegraphics[height=2cm,width=1.5cm]{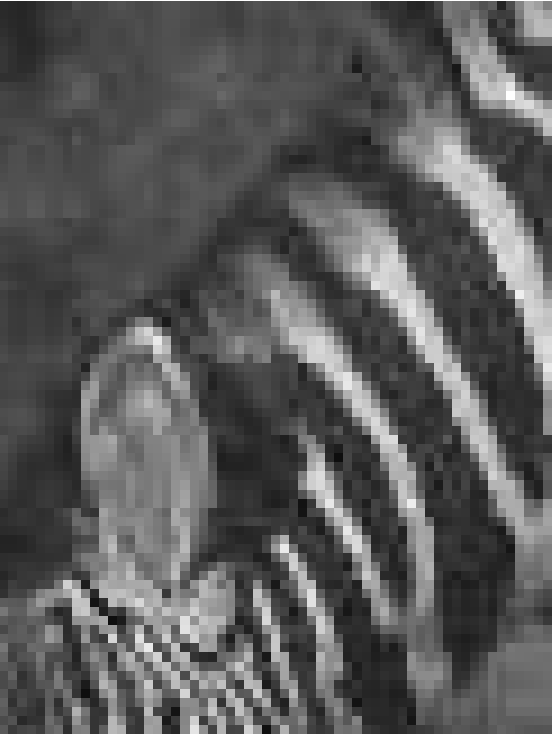}\\
&\includegraphics[height=2cm,width=2.7cm]{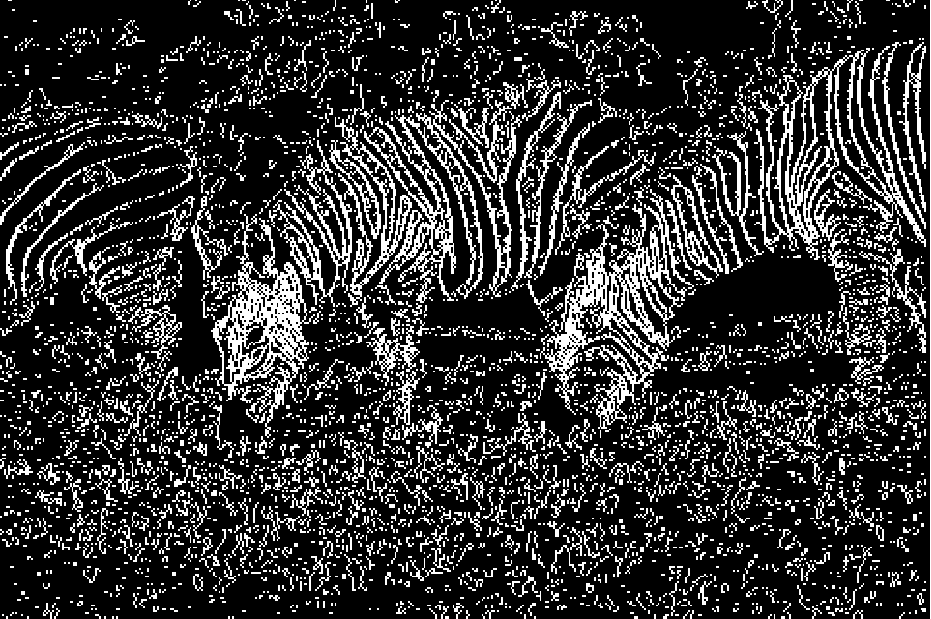}&
\includegraphics[height=2cm,width=2.7cm]{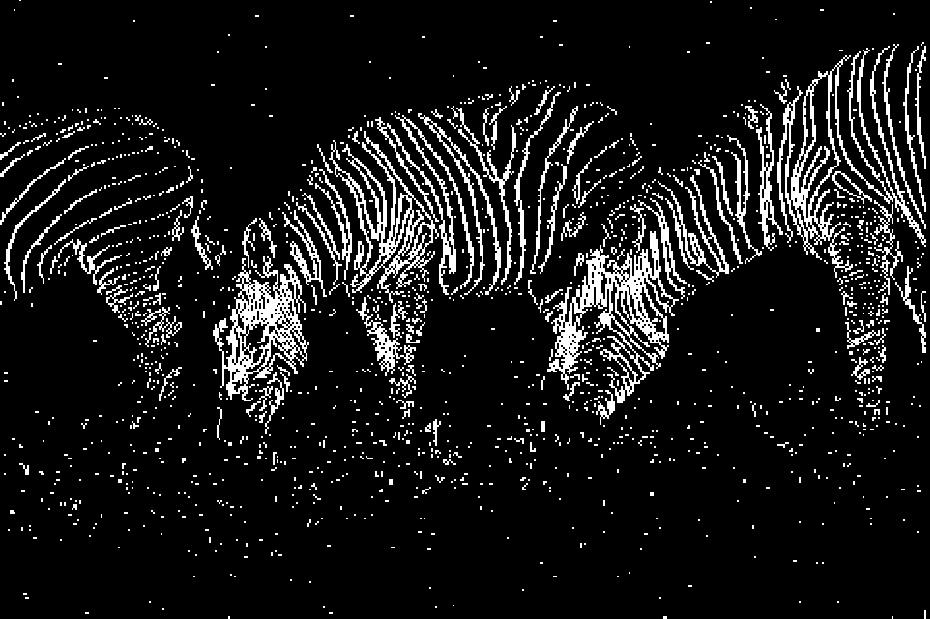} &
&
&\includegraphics[height=2cm,width=1.5cm]{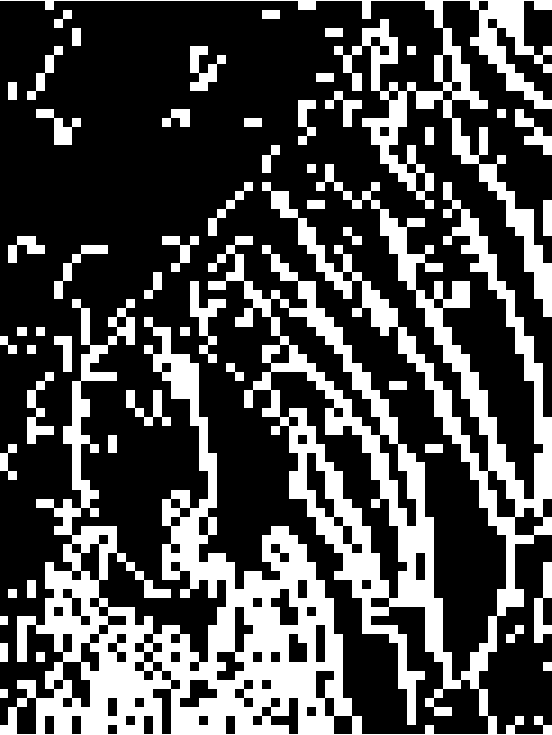}&
 \includegraphics[height=2cm,width=1.5cm]{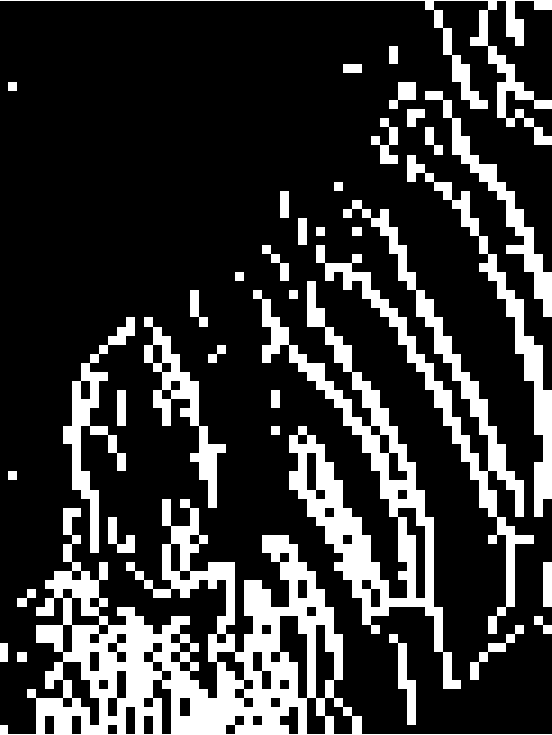}\\ 
   &SSIM: 0.87 & SSIM: 0.85 &&&\\
      & Jacc: 0.35 &Jacc: 0.49 &&&\\
 \includegraphics[height=2cm,width=2.7cm]{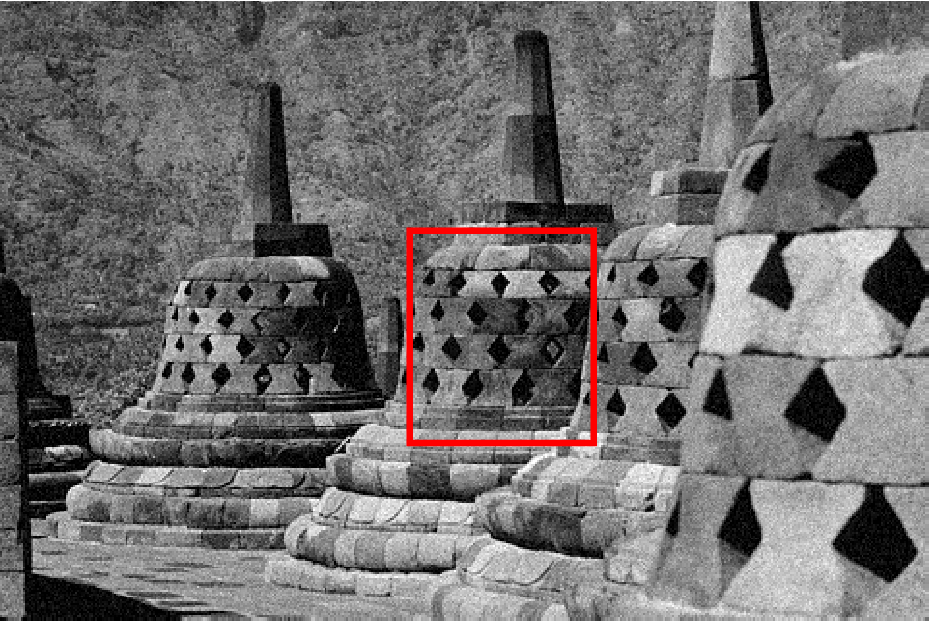}&
\includegraphics[height=2cm,width=2.7cm]{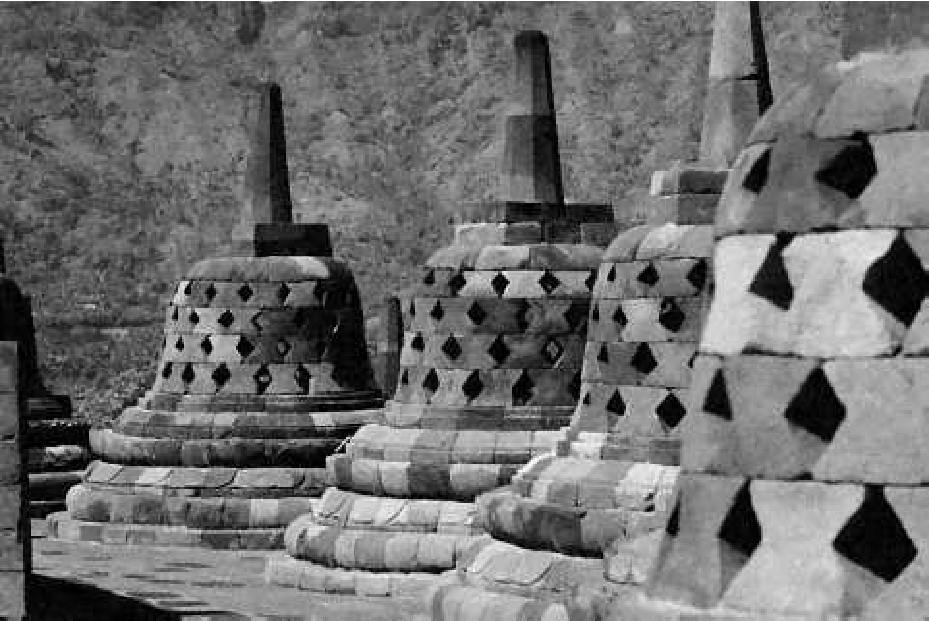}&
\includegraphics[height=2cm,width=2.7cm]{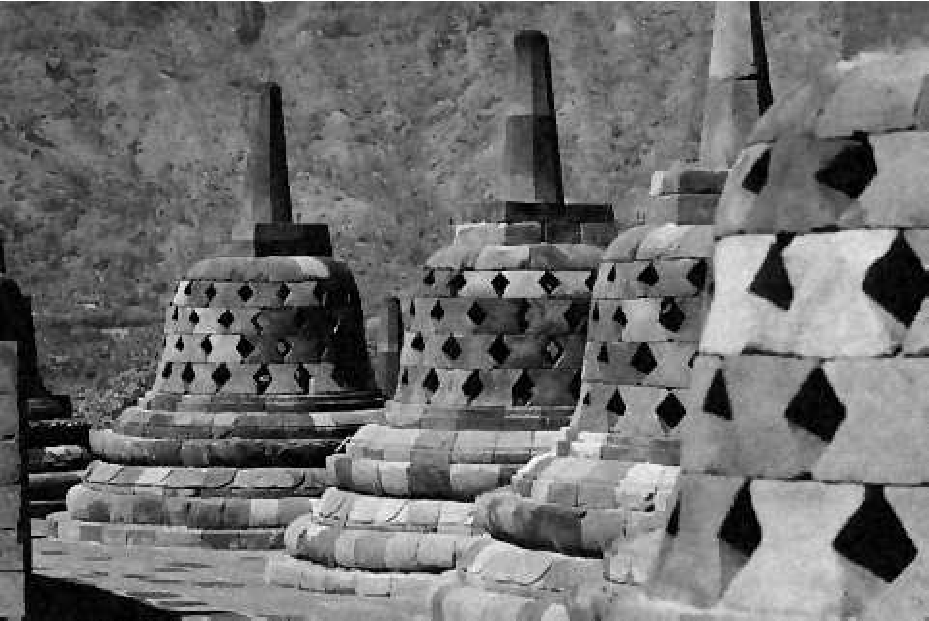}&
\includegraphics[height=2cm,width=1.5cm]{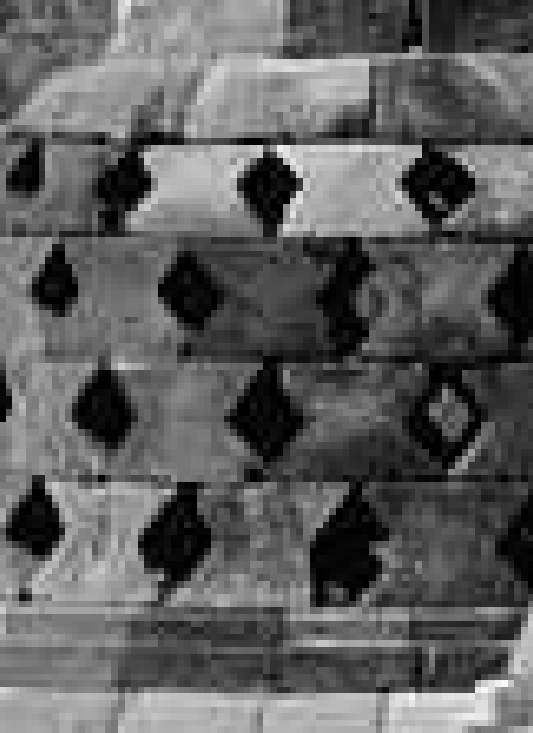}&
\includegraphics[height=2cm,width=1.5cm]{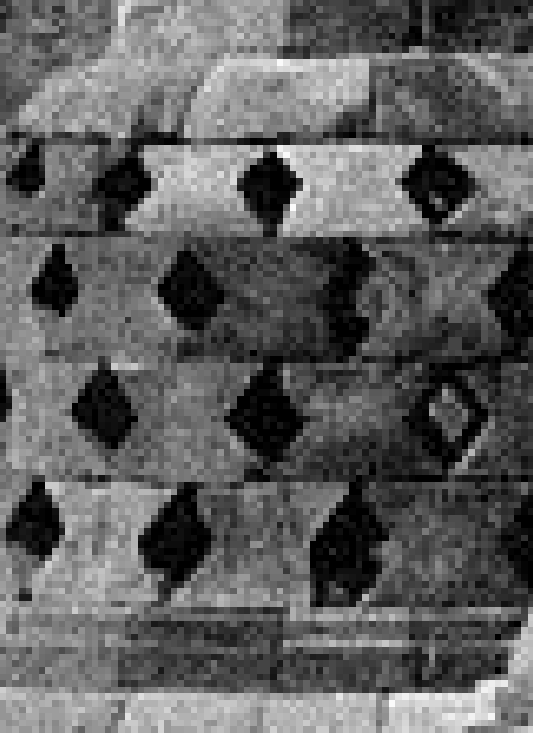}&
\includegraphics[height=2cm,width=1.5cm]{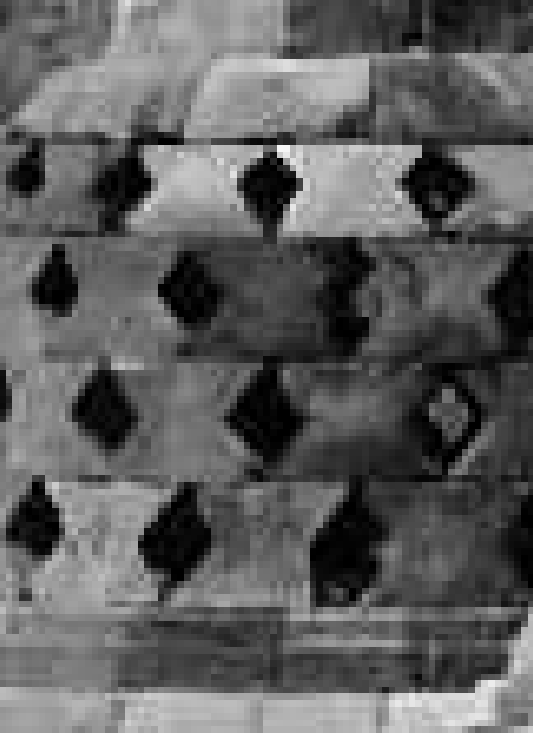}&
 \includegraphics[height=2cm,width=1.5cm]{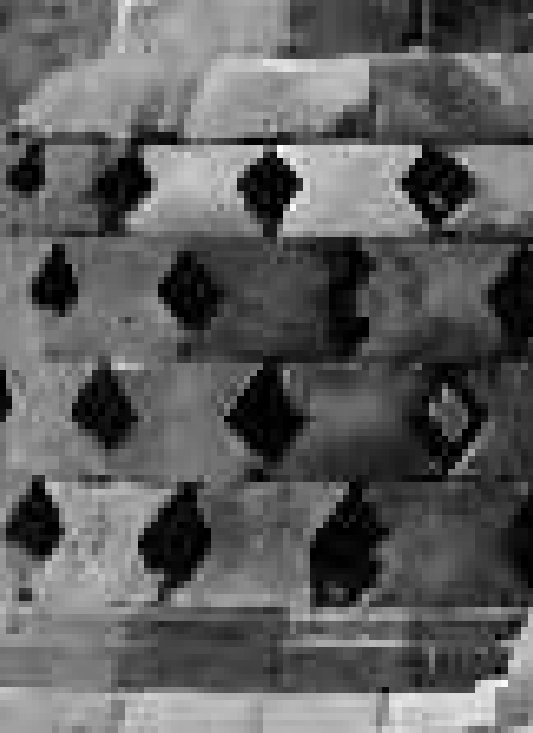}\\
&\includegraphics[height=2cm,width=2.7cm]{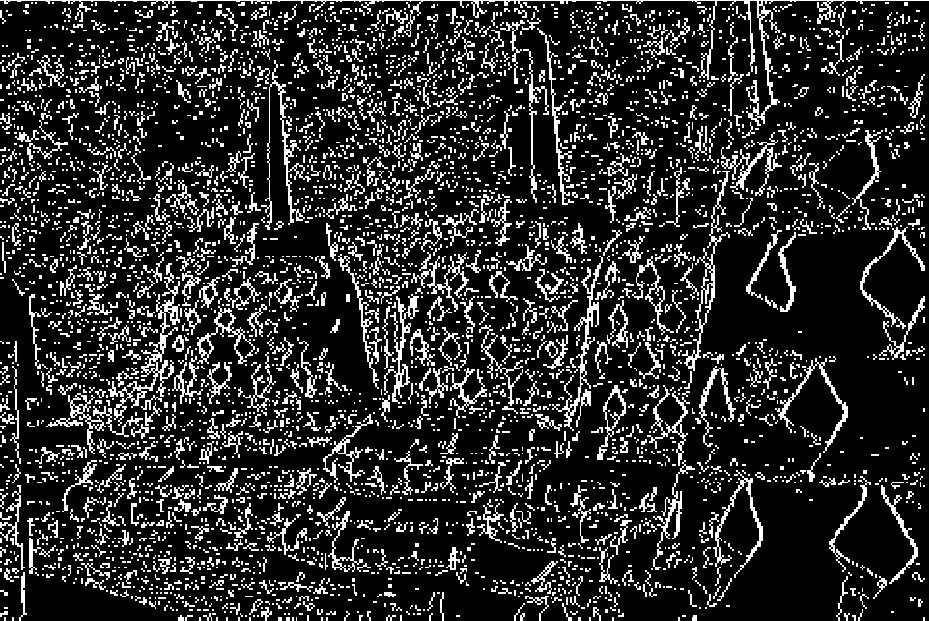}&
\includegraphics[height=2cm,width=2.7cm]{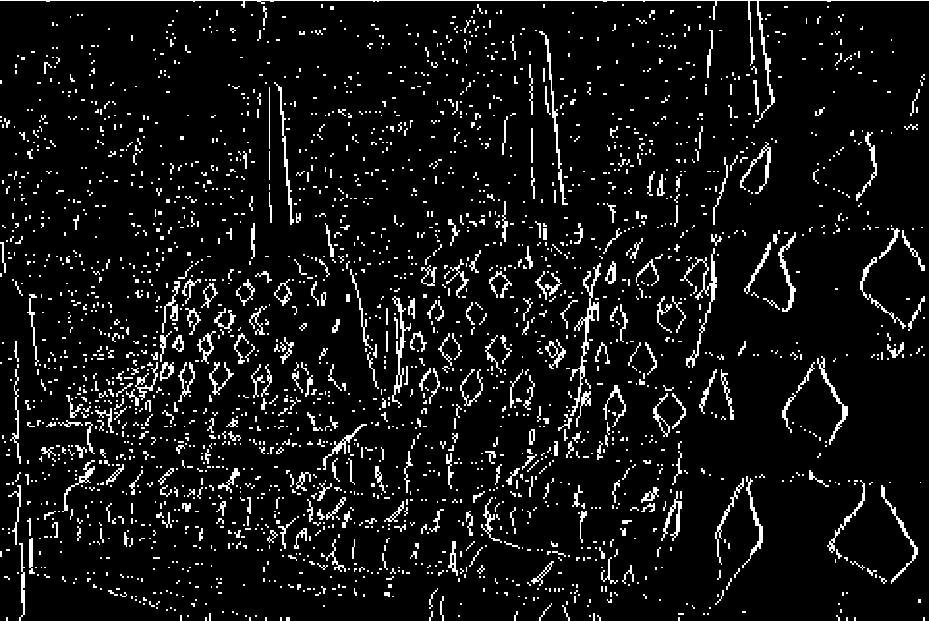} &
&
&\includegraphics[height=2cm,width=1.5cm]{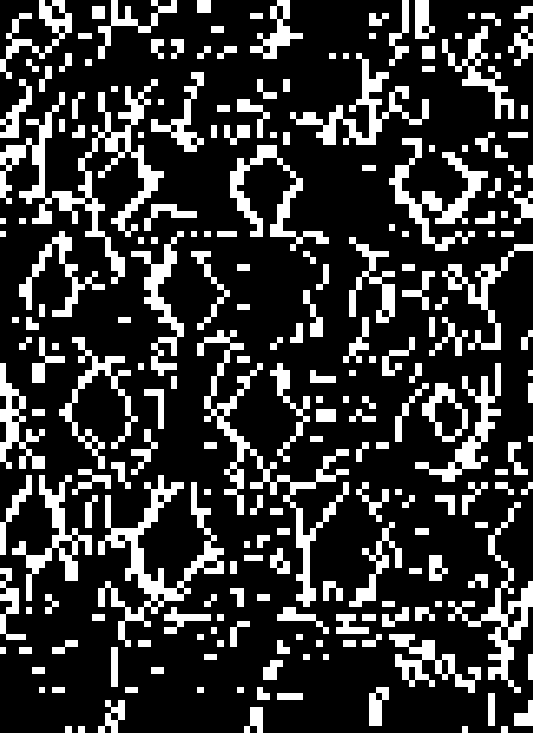}&
 \includegraphics[height=2cm,width=1.5cm]{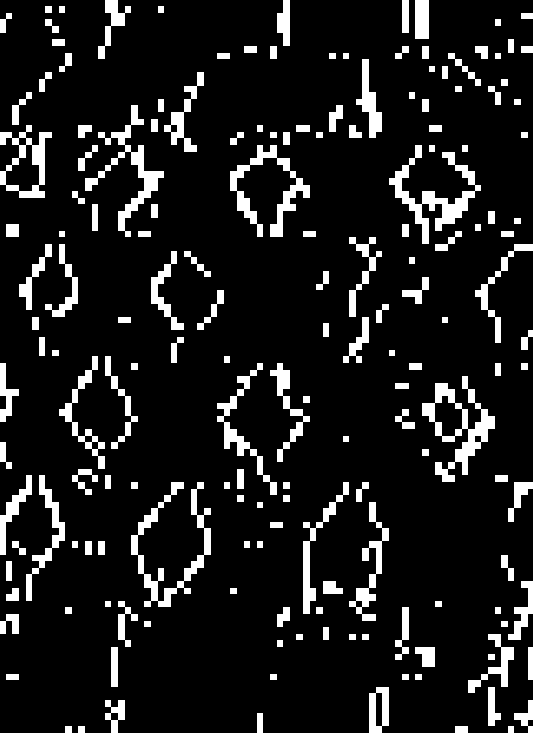}
\end{tabular}}}
\caption{Comparisons between SUGAR T-ROF and SUGAR D-MS for noisy images extracted from the BSD69 dataset \cite{Roth2009Fields}.\label{fig:real} }
\end{figure*}

\subsection{Impact of the estimation of $\sigma$}
\label{ssec:sig_est}
On real data, the noise level $\sigma$ needed to implement \textit{Averaged SUGAR D-MS} is unknown.  
The most usual method to estimate the noise standard deviation is the \textit{median absolute deviation} (MAD) of 2D discrete wavelet coefficients \cite{Donoho_D_1995_tit_den_st}:
\begin{equation}
	\hat{\sigma} = \frac{\textrm{MAD}\left(\left\lbrace\lvert \psi_{H,k} \rvert,\lvert \psi_{V,k} \rvert,\lvert\psi_{D,k}\rvert \vert k \in \lbrace1, \ldots, \frac{N}{4} \rbrace \right\rbrace \right)}{0.6745},
	\label{eq:sigest}
\end{equation}
where $\psi_{H,k},\psi_{V,k},\psi_{D,k}$ are the three wavelets coefficents (horizontal, vertical and diagonal) at the finest scale. 
Table~\ref{table_SNR_noises} 
presents the PSNR values with the estimated noise level $\hat\sigma$. In addition, some estimates are depicted in Fig.~\ref{fig:denoisedgeometries} (2nd and 4th rows).

The results with either estimated or true noise level 
are similar, 
attesting the robustness of \textit{Averaged SUGAR D-MS} to the estimation of the noise level.

\subsection{Real-world images}
The proposed automated joint denoising and contour detection procedure \textit{Averaged SUGAR D-MS}  is evaluated on real-world images extracted from BSD69 dataset \cite{Roth2009Fields} degraded with a Gaussian noise with $\sigma=0.05$. In our experiments we set $R=5$ and $\sigma$ has been estimated from noisy data following \eqref{eq:sigest}.
Denoised images and contours provided by the proposed data-driven \textit{Averaged SUGAR D-MS} strategy are compared with those yield by SUGAR T-ROF (a two-step procedure, consisting in, first, a piecewise constant denoising with automated tuning of the regularization parameter~\cite{deledalle2014stein}, followed by an iterative thresholding procedure~\cite{cai2013multiclass}). In Fig.~\ref{fig:real}, we can observe that the denoising performance are very close for both procedures (in terms of SSIM, SUGAR T-ROF is slightly better)  while the contour detection is significantly improved with the SUGAR D-MS procedure (which is confirmed when computing Jaccard index w.r.t contours obtained from the original image).

\end{document}